\documentclass[11pt,a4paper]{article}
\usepackage[a4paper,margin=1.05in]{geometry}
\usepackage{amsmath,amssymb,amsfonts,amsthm,mathtools,bm}
\usepackage{mathrsfs}
\usepackage{graphicx}
\usepackage{booktabs}
\usepackage{enumitem}
\usepackage{xcolor}
\usepackage{hyperref}
\usepackage[nameinlink,noabbrev]{cleveref}
\usepackage{microtype}
\usepackage{float}
\usepackage[section]{placeins}
\usepackage{array}
\usepackage{caption}
\usepackage{subcaption}
\usepackage{algorithm}
\usepackage{algpseudocode}

\hypersetup{colorlinks=true,linkcolor=blue!50!black,citecolor=blue!50!black,urlcolor=blue!50!black}

\newcommand{\cH}{\mathcal H}

\newcommand{\cL}{\mathcal L}
\newcommand{\cR}{\mathcal R}
\newcommand{\cF}{\mathcal F}
\newcommand{\cA}{\mathcal A}
\newcommand{\cB}{\mathcal B}

\newcommand{\cV}{\mathcal V}
\newcommand{\cE}{\mathcal E}
\newcommand{\bbE}{\mathbb E}
\newcommand{\bbP}{\mathbb P}
\newcommand{\R}{\mathbb R}

\newcommand{\Tr}{\operatorname{Tr}}
\newcommand{\rank}{\operatorname{rank}}
\newcommand{\Comm}{\operatorname{Comm}}
\newcommand{\Herm}{\operatorname{Herm}}
\newcommand{\Var}{\operatorname{Var}}
\newcommand{\Cov}{\operatorname{Cov}}
\newcommand{\diag}{\operatorname{diag}}

\newcommand{\op}{\mathrm{op}}
\newcommand{\HS}{\mathrm{HS}}
\newcommand{\eff}{\mathrm{eff}}
\newcommand{\shot}{\mathrm{shot}}
\newcommand{\br}{\mathrm{br}}
\newcommand{\sym}{\mathrm{sym}}
\newcommand{\comm}{\mathrm{comm}}
\newcommand{\var}{\mathrm{var}}
\newcommand{\disc}{\mathrm{disc}}
\newcommand{\gen}{\mathrm{gen}}
\newcommand{\opt}{\mathrm{opt}}
\newcommand{\bias}{\mathrm{bias}}
\newcommand{\val}{\mathrm{val}}

\newcommand{\dd}{\mathrm d}
\newcommand{\norm}[1]{\left\lVert #1\right\rVert}
\newcommand{\abs}[1]{\left\lvert #1\right\rvert}
\newcommand{\ip}[2]{\left\langle #1,#2\right\rangle}

\newcommand{\srmscore}{\mathfrak S}
\newcommand{\Rad}{\operatorname{Rad}}
\newcommand{\TV}{\operatorname{TV}}
\newcommand{\argmin}{\operatorname*{arg\,min}}

\newcommand{\fit}{\mathrm{fit}}
\newcommand{\rel}{\mathrm{rel}}
\newcommand{\valsplit}{\mathrm{val}}
\newcommand{\sh}{\mathrm{sh}}

\theoremstyle{plain}
\newtheorem{theorem}{Theorem}[section]
\newtheorem{proposition}[theorem]{Proposition}
\newtheorem{lemma}[theorem]{Lemma}
\newtheorem{corollary}[theorem]{Corollary}
\theoremstyle{definition}
\newtheorem{definition}[theorem]{Definition}
\newtheorem{assumption}[theorem]{Assumption}
\theoremstyle{remark}
\newtheorem{remark}[theorem]{Remark}

\crefname{theorem}{theorem}{theorems}
\Crefname{theorem}{Theorem}{Theorems}
\crefname{proposition}{proposition}{propositions}
\Crefname{proposition}{Proposition}{Propositions}
\crefname{lemma}{lemma}{lemmas}
\Crefname{lemma}{Lemma}{Lemmas}
\crefname{corollary}{corollary}{corollaries}
\Crefname{corollary}{Corollary}{Corollaries}
\crefname{definition}{definition}{definitions}
\Crefname{definition}{Definition}{Definitions}
\crefname{assumption}{assumption}{assumptions}
\Crefname{assumption}{Assumption}{Assumptions}
\crefname{remark}{remark}{remarks}
\Crefname{remark}{Remark}{Remarks}

\setlist[itemize]{leftmargin=1.5em}
\setlist[enumerate]{leftmargin=1.7em}

\title{\bfseries Statistical Symmetry Release for Equivariant Quantum Learning}
\author{Zeyu Chen}
\date{}

\begin{document}
\maketitle

\begin{abstract}
Hard symmetry constraints reduce model complexity, but can also erase label information. Statistical symmetry release determines when finite data and quantum measurements justify relaxing such a constraint, which directions to open, and how far to move. We connect global signal detection to local, loss-dependent improvement. A two-copy twirl--swap gate estimates task information in the symmetry-breaking complement with a dimension-independent copy count under paired-state and group-unitary access; reweighting the same records resolves representation sectors. An exact duality distinguishes this Hilbert--Schmidt signal from the larger signal accessible to bounded-outcome readouts. Local improvement is governed by the release gradient and a loss-corrected double-commutator matrix. Simultaneous confidence bounds convert empirical direction selection into certified descent, using either shared Pauli measurements or scalar probes with state-independent truncation bounds. Gaussian testing lower bounds quantify the cost of searching over unknown directions in the calibrated local experiment. Independent validation controls adaptively generated models, and a fast squared-loss bound preserves the approximation--estimation rate of a nested release path. On an eight-qubit Ising model, shared measurements certify release with 6300 times fewer shots than the specified scalar estimator on the tested budget grids. Quotient quantum natural gradient then controls parameter redundancy during training. Together, these results turn symmetry relaxation into a statistically justified model-selection decision.
\end{abstract}


\section{Introduction}

Symmetry determines which distinctions a quantum learning model can represent. Enforcing it can reduce model complexity and improve trainability under suitable architecture and task assumptions \cite{Schatzki2024,McClean2018,Cerezo2021Cost,Ragone2024LieBP}. An overly restrictive symmetry also creates a precise failure mode: states with the same group twirl are indistinguishable to every invariant single-copy readout. Under global spin flip, two classes polarized in opposite directions have identical twirls, although their odd magnetization reveals the label. Optimization within the invariant class leaves this information inaccessible.

The resulting question is when a finite data set and measurement budget justify relaxing the constraint. Two decisions enter. A global test asks whether the symmetry-breaking readout space contains label information; a local test asks whether an available circuit perturbation reduces the chosen loss. Their distinction is essential: a global signal may lie outside the implementable tangent directions, and searching over noisy local estimates can produce apparent improvement even when none exists.

Statistical symmetry release connects these decisions through quantities that can be measured and validated. A two-copy twirl--swap gate estimates the squared task signal over Hilbert--Schmidt-normalized breaking observables, with a copy count independent of Hilbert dimension under the specified access model. Character reweighting resolves its symmetry sectors from the same records. We also characterize the strongest bounded-outcome breaking signal as a trace-norm distance to the commutant. This duality explains the rank-dependent conversion between the two readout normalizations and makes a retention decision meaningful at a declared signal scale.

For local release, the descent gradient and the negative loss Hessian measure the benefit of an implementable perturbation. A linear loss observable $C$ gives the directional curvature $\Tr([B,[B,C]]\rho_0)$, while nonlinear supervised losses require a chain-rule correction. A simultaneous confidence bound remains valid after selecting a direction from the estimated curvature matrix and supplies a quantitative descent step. We make the certificate computable through shared Pauli records and through scalar loss probes with an operator bound on finite-difference bias. On the tested eight-qubit Ising instance, the shared-record estimator certifies the leading release direction at $10^4$ total shots, compared with $6.3\times10^7$ for the specified scalar estimator.

The decision extends from one step to a family of trained models. A structure specifies a backbone symmetry, an activated breaking subspace, and soft amplitudes. Gaussian testing bounds quantify the cost of searching unknown directions in a calibrated local experiment. Independent validation then compares the adaptively generated predictors. For bounded squared loss, the excess-loss variance relation gives a fast validation bound that preserves the approximation--estimation rate of a nested release path. These guarantees compare the realized candidate family and certify local steps within their stated smoothness range.

Statistical release complements symmetry discovery and geometric training. Symmetry discovery \cite{CompanionDiscovery} supplies candidate physical groups and label actions. The quotient quantum natural gradient (QNG) developed in \cite{ChenQNGQuotient2026} builds on information geometry \cite{Amari1998,BraunsteinCaves1994,Stokes2020} and handles state-preserving parameter redundancy after release. Measurement evidence determines model enlargement, validation evaluates its predictive benefit, and quotient geometry controls the ensuing update. Quantum decision theory separates the statistical information from the measurement access used to extract it \cite{Holevo1973,Helstrom1976}.

\subsection{Relation to existing quantum equivariant learning}

Group-invariant models enforce a prescribed action on the predictor \cite{Larocca2022}. Representation theory organizes the corresponding operator and circuit spaces \cite{Ragone2022}, and equivariant channels supply compatible layers \cite{Nguyen2024EQNN}. Fixed permutation symmetry can yield trainability and generalization guarantees \cite{Schatzki2024}, while quantum convolutional architectures exploit locality and hierarchy \cite{CongChoiLukin2019}. Horizontal quantum gates provide geometric directions that relax strict equivariance \cite{Wiersema2025}. These constructions specify what a symmetry-compatible or relaxed model can express. The release decision studied here determines whether its additional capacity is supported by the task and the available measurements.

Quantum symmetry tests estimate fidelity or Hilbert--Schmidt asymmetry \cite{LaBordeRethinasamyWilde2023,Bandyopadhyay2023}, and coherent projector circuits resolve symmetric and other representation sectors \cite{LaBordeRethinasamyWilde2024}. We use this measurement perspective for a label-conditional class difference, then connect the detected signal to loss curvature and independent model validation. The statistical ingredients---Gaussian testing, simultaneous concentration, and structural-risk comparison---retain their standard roles. Their combination identifies the distinct evidence required for global signal detection, a local release step, and selection of the trained predictor.

\section{Problem Setting: Symmetry as a Model-Selection Variable}
\label{sec:setup}

Treating symmetry as a variable requires a candidate set to select from, a measurement interface that can score candidates, and an accounting of what each choice costs. The candidate classes and access assumptions fix the scope of every guarantee below.

\subsection{Data, group actions, and candidate structures}

All quantum states act on a finite-dimensional Hilbert space $\cH$ of dimension $d$. The Hilbert--Schmidt inner product is $\ip{A}{B}=\Tr(A^\dagger B)$, with unnormalized trace; $\|A\|_1$, $\|A\|_{\HS}$, and $\|A\|_{\op}$ denote the trace, Hilbert--Schmidt, and operator norms. Matrix eigenvalues are ordered decreasingly. Groups used for twirling are finite or compact, their representations are unitary and continuous, and subgroup integrals use normalized counting or Haar measure on closed subgroups. All risks are minimized. Positive release curvature means a favorable second-order loss variation, with first-order terms and activation costs treated separately.

Let $Z_i=(\rho_i,y_i)$ be independent and identically distributed training examples. The input $\rho_i$ is either a quantum state or a classical datum after a quantum feature map; $y_i$ is a label. A candidate group $G$ acts on $\cH$ by $U_g$. For a retained subgroup $H\leq G$, the hard invariant Hermitian readout space is
\begin{equation}
\cA_H=\Comm(H)\cap\Herm(\cH),\qquad
\Comm(H)=\{O:[O,U_h]=0\ \text{for all}\ h\in H\}.
\end{equation}
For a linear readout $f_O(\rho)=\Tr(O\rho)$, hard $H$-symmetry restricts $O\in\cA_H$. For a unitary ansatz, the same formal role is played by $H$-equivariant layers, intertwiners, orbit averages, or sectorwise trainable blocks. The statistical results below require only that each candidate structure produces an auditable finite-dimensional or effectively finite-dimensional class $\cF_a$.

Let $B_1,\ldots,B_p$ be Hermitian breaking generators. A local released state family around a hard backbone is
\begin{equation}
\rho_\theta(\beta)=e^{-iB(\beta)}\rho_\theta e^{iB(\beta)},
\qquad
B(\beta)=\sum_{a=1}^p\beta_aB_a.
\label{eq:release-family}
\end{equation}
Soft release is represented by gates
\begin{equation}
\beta_j=s_j(\alpha_j)u_j,\qquad
s_j(\alpha_j)=\frac{1}{1+\alpha_j},\qquad \alpha_j\in[0,\infty],
\label{eq:soft-scale}
\end{equation}
with the convention $s_j(\infty)=0$. Thus $\alpha_j=\infty$ is hard retention in direction $j$, while $\alpha_j=0$ is full local release. A gate changes the feasible class when the ungated coordinates $u$ obey a fixed radius constraint, as in \cref{thm:soft-effective-dimension}; with unrestricted $u$, a finite nonzero gate is only a reparameterization.

\begin{definition}[Structure variable]
A releasable equivariant model is indexed by
\begin{equation}
 a=(H,S,\alpha),
\end{equation}
where $H$ is the symmetry imposed on the hard backbone, $S\subseteq\operatorname{span}\{B_j\}_{j=1}^p$ is the activated breaking subspace, and $\alpha$ contains its soft-release penalties. Write $\cF_a$ for the resulting predictor class and define $R(f)=\bbE\ell(f(\rho),Y)$, $\widehat R_N(f)=N^{-1}\sum_i\ell(f(\rho_i),y_i)$, and $R(a)=\inf_{f\in\cF_a}R(f)$. A trained output is denoted $\widehat f_a$; its risk can exceed $R(a)$.

\end{definition}

The index $H$ records the reference hard constraint. Opening an $H$-breaking direction generally reduces the exact symmetry of the resulting model. For unitary release, a sufficient residual subgroup is $\{h\in H:[U_h,B]=0\text{ for every active }B\in S\}$, provided the backbone and readout respect it. A positive soft gate changes the allowed amplitude, while an exactly closed gate removes that direction altogether.

\subsection{Sample splitting, measurement access, and adaptive candidate families}
\label{subsec:splitting}

The main statistical risk in adaptive release is post-selection bias. If the same data are used to search over breaking directions and to report the final risk, maximizing curvature over the dictionary systematically amplifies noise. The high-probability guarantees in this paper therefore use three independent data streams,
\begin{equation}
\mathcal D=\mathcal D_{\fit}\sqcup \mathcal D_{\rel}\sqcup \mathcal D_{\valsplit}.
\end{equation}
The fitting stream \(\mathcal D_{\fit}\) trains the parameters inside each candidate structure. The release stream \(\mathcal D_{\rel}\), together with its shot budget, discovers candidate breaking subspaces and estimates curvature, Gram, and Fisher matrices. The validation stream \(\mathcal D_{\valsplit}\) is used only for the final structural-risk comparison and for post-selection confidence intervals. A random candidate family \(\widehat{\mathscr A}\) may depend on \(\mathcal D_{\fit}\cup\mathcal D_{\rel}\), but after conditioning on these two streams the validation stream must remain independent. The fitting and release streams include their auxiliary measurement records and training randomness. Conditioning on this complete prevalidation information gives the independence required in \cref{thm:post-selection-oracle}.

We use three measurement-access modes. In the \emph{paired-copy gate mode}, two independently prepared states are measured through the SWAP observable, implemented by destructive Bell measurements for qubit registers. This mode estimates a global quadratic breaking mass and its representation-sector decomposition before a direction is chosen. In the \emph{matrix-element mode}, one estimates
\begin{equation}
O_{ab}=\frac12\bigl([B_a,[B_b,C]]+[B_b,[B_a,C]]\bigr),
\qquad
K_{ab}=\Tr(O_{ab}\rho_0),
\end{equation}
which is appropriate when \(O_{ab}\) decomposes into low-weight Pauli sums. In the \emph{scalar-probe mode}, one implements \(e^{-i\delta B}\), measures the scalar losses \(\cL(\delta)\) and \(\cL(-\delta)\), and uses a central difference to estimate directional curvature. Every mode supplies a simultaneous confidence radius before its output enters the release criterion. The ambient probe family must be fixed independently of those shots, or the confidence event must cover the entire family from which it is selected. Empirical eigendirections within a simultaneously estimated matrix can reuse its shots; proposing additional observables adaptively requires a fresh batch or an enlarged uniform bound. Paired-copy access is used as a gate when its predicted boundary is favorable.

\begin{assumption}[Post-selection validity discipline]
\label[assumption]{ass:postselection}
The main finite-sample results are stated under the following conditions.
\begin{enumerate}
\item The loss is bounded or sub-exponential, and in the linear readout results it is Lipschitz in the prediction.
\item The candidate structure family is finite before validation, or can be replaced by a finite \(\varepsilon\)-net with its risk-approximation error included. A random family \(\widehat{\mathscr A}\) may depend on \(\mathcal D_{\fit}\cup\mathcal D_{\rel}\) only.
\item The curvature estimator satisfies a simultaneous event such as \(\|\widehat K-K\|_{\op}\le r_\delta\), generated by classical shadows, grouped Pauli measurements, scalar finite differences, or randomized Hessian sketches.
\item For QNG guarantees, the physical Fisher matrix in the horizontal space of the released structure has a gap \(\lambda_{\min}^+(F_S)>0\), certified by the truncation rule of \cite{ChenQNGQuotient2026}. If this gap is below the estimation noise, the algorithm reports the branch as statistically ambiguous rather than interpreting it as a physical release direction; an exact zero mode and a physical mode at the estimation-noise scale are not separable from the estimated Fisher matrix alone \cite{ChenQNGQuotient2026}.
\end{enumerate}
\end{assumption}

\subsection{Twirling blind spots and the release complement}
\label{subsec:twirling-blind-spots}

For a finite or compact retained group $H$, define the Heisenberg twirl and Schr\"odinger twirl by
\begin{equation}
\mathcal T_H(A)=\int_H U_h^\dagger A U_h\,\dd h,
\qquad
\mathcal T_H^\ast(\rho)=\int_H U_h\rho U_h^\dagger\,\dd h,
\label{eq:twirl}
\end{equation}
with normalized Haar measure, or the uniform average in the finite case. The range of $\mathcal T_H$ is $\Comm(H)$.

\begin{proposition}[Twirling blind spot and release complement]
\label[proposition]{prop:twirling-blind-spot}
Let $O\in\Comm(H)$ and let $\rho,\rho'$ be states on $\cH$.
\begin{enumerate}
\item Hard $H$-invariant readouts depend only on the twirled state:
\begin{equation}
\Tr(O\rho)=\Tr\bigl(O\mathcal T_H^\ast(\rho)\bigr).
\label{eq:hard-sees-twirl}
\end{equation}
Consequently, if $\mathcal T_H^\ast(\rho)=\mathcal T_H^\ast(\rho')$, then no hard $H$-invariant readout can distinguish $\rho$ from $\rho'$.
\item In squared-loss supervised learning, any label component orthogonal to the $\sigma$-algebra generated by the twirled feature $\mathcal T_H^\ast(\rho)$ contributes irreducible hard-symmetry bias. Among all predictors measurable with respect to that twirled feature, the optimum is the conditional expectation of the label given the feature. A restricted linear readout class can have additional approximation error. General nonlinear invariant functions need not factor through the single-copy twirl.
\item A release dictionary can be audited by decomposing every candidate generator as
\begin{equation}
B_a=\mathcal T_H(B_a)+(I-\mathcal T_H)B_a.
\end{equation}
The commutant component $\mathcal T_H(B_a)$ preserves $H$ and can be assigned to the backbone. The physical local release component is $(I-\mathcal T_H)B_a$ modulo the null space of the Fisher Gram $G^F$, the quantum Fisher information of the release tangents defined in \eqref{eq:gram}. If $(I-\mathcal T_H)B_a$ has zero $G^F$ norm at the probing state, that candidate is locally statistically indistinguishable and cannot be certified by a pointwise release statistic.
\end{enumerate}
\end{proposition}

\begin{proof}
For $O\in\Comm(H)$,
\begin{equation}
\Tr\bigl(O\mathcal T_H^\ast(\rho)\bigr)
=\int_H\Tr(O U_h\rho U_h^\dagger)\,\dd h
=\int_H\Tr(U_h^\dagger O U_h\rho)\,\dd h
=\Tr(O\rho),
\end{equation}
which proves the blind-spot claim. The squared-loss statement is the Hilbert-space projection theorem applied to the closed subspace of functions measurable with respect to the twirled feature. Finally, $\mathcal T_H$ is the orthogonal projector onto $\Comm(H)$ under the Hilbert--Schmidt inner product. Its complement is therefore the part of the generator that transforms in nontrivial $H$ sectors. A zero Fisher-Gram norm means that the induced tangent vector has zero local statistical length at $\rho_0$, so no local estimator at that state can distinguish it from a gauge or null direction.

\end{proof}

\begin{corollary}[Binary separation by a released symmetry sector]
\label[corollary]{cor:binary-release-separation}
Let $Y\in\{+1,-1\}$ have equal prior probabilities and class-conditional states $\rho_+$ and $\rho_-$. Suppose
\begin{equation}
\mathcal T_H^\ast(\rho_+)=\mathcal T_H^\ast(\rho_-).
\end{equation}
Then every two-outcome positive operator-valued measure $\{E_+,E_-\}$ with $E_\pm\in\Comm(H)$ has success probability exactly $1/2$. If there exists a Hermitian released observable $A$ with $\norm{A}_{\op}\le1$ and
\begin{equation}
\Delta_A=\Tr(A\rho_+)-\Tr(A\rho_-)>0,
\end{equation}
then measuring $A$ independently on $M$ copies and thresholding the empirical mean at $(\Tr(A\rho_+)+\Tr(A\rho_-))/2$ has average misclassification probability at most
\begin{equation}
\exp\left(-\frac{M\Delta_A^2}{8}\right).
\label{eq:binary-release-error}
\end{equation}
For $0<m\le1$, the spin-flip example $H=\{I,X^{\otimes n}\}$ and
\begin{equation}
\rho_y=\left(\frac{I+ymZ}{2}\right)^{\otimes n},\qquad y\in\{+1,-1\},
\end{equation}
the two twirls coincide, every hard spin-flip-invariant measurement is at chance for the sign label, and the released order parameter $A=M_z/n$, with $M_z=\sum_{i=1}^n Z_i$, has $\Delta_A=2m$. Here $X_i,Y_i,Z_i$ denote single-site Pauli matrices; $M_x=\sum_iX_i$ and $M_y=\sum_iY_i$ use the same convention.
\end{corollary}

\begin{proof}
Let $\omega=\mathcal T_H^\ast(\rho_+)=\mathcal T_H^\ast(\rho_-)$. Since $E_+\in\Comm(H)$ and $E_-=I-E_+$ also lies in $\Comm(H)$, \cref{prop:twirling-blind-spot} gives
\begin{equation}
\frac12\Tr(E_+\rho_+)+\frac12\Tr(E_-\rho_-)
=\frac12\Tr(E_+\omega)+\frac12\Tr((I-E_+)\omega)=\frac12.
\end{equation}
For the released observable, each shot outcome lies in $[-1,1]$. Under either class, Hoeffding's inequality \cite{Hoeffding1963} bounds the probability that the empirical mean crosses the midpoint threshold by $\exp(-M\Delta_A^2/8)$. The spin-flip example follows from $XZX=-Z$ and $\Tr((M_z/n)\rho_y)=ym$.
\end{proof}

Release is useful when the task law uses information erased by the twirl and the available measurements resolve that information above sampling uncertainty.

\subsection{The two costs of retaining or releasing symmetry}
\label{subsec:two-costs}

For squared loss, let $f_\star(x)=\bbE[Y\mid X=x]$. Let $\Pi_H$ be the $L^2(P_X)$ orthogonal projection onto a closed linear subspace of $H$-equivariant predictors. Its irreducible symmetry bias is
\begin{equation}
\cE^{\sym}_{\bias}(H)=\norm{f_\star-\Pi_Hf_\star}_{L^2(P_X)}^2.
\label{eq:projection-bias}
\end{equation}
For an $H$-invariant input law and a unitary label action $V_g$, define $(T_gf)(x)=V_gf(g^{-1}x)$. This is a unitary action on $L^2(P_X)$, and its invariant subspace is the equivariant class. The identity $f_\star-T_gf_\star=(I-T_g)(f_\star-\Pi_Hf_\star)$ gives $\norm{f_\star-\Pi_Hf_\star}\ge\tfrac12\norm{f_\star-T_gf_\star}$ for every $g\in H$ \cite{CompanionDiscovery}. Thus one broken group element already forces approximation bias.

If $H$ is too small, or if $S$ is too large, the readout and tangent spaces become larger. Let $\{A_j\}_{j=1}^{m_H}$ be a Hilbert--Schmidt orthonormal basis for $\cA_H$ and define
\begin{equation}
\Phi_H(\rho)=(\Tr(A_1\rho),\ldots,\Tr(A_{m_H}\rho))\in\R^{m_H},
\qquad
\widehat\Sigma_H=\frac1N\sum_{i=1}^N\Phi_H(\rho_i)\Phi_H(\rho_i)^\top.
\end{equation}
The data-dependent commutant complexity used throughout the paper is
\begin{equation}
\widehat{\mathfrak C}_{\comm}(H)=R_H\sqrt{\frac{\Tr\widehat\Sigma_H}{N}},
\label{eq:effective-comm}
\end{equation}
with readout radius $R_H$. If the covariance is unavailable, the conservative replacement is $R_H\sqrt{m_H/N}$.

Release also has shot cost. If a release gradient or curvature is estimated from $M$ measurements with noise scale $\sigma^2$, an effective breaking dimension $D_{\br}(S,\alpha)$ contributes
\begin{equation}
\cE_{\shot}(S,\alpha)\simeq \frac{\sigma^2}{2M}D_{\br}(S,\alpha).
\label{eq:shot-cost}
\end{equation}
Here $D_{\br}(S,\alpha)$ denotes the number of hard-opened breaking coordinates, or its soft effective analogue defined in \eqref{eq:soft-shot-dim}. The scale of the soft penalty is fixed by the following local benchmark, which also supplies the effective-dimension formula used in \cref{thm:soft-effective-dimension}.

\begin{proposition}[Finite-shot ridge scale of the soft release]
\label[proposition]{prop:ridge-scale}
Work in a local quadratic surrogate around the released optimum,
\begin{equation}
\cL_{\mathrm{quad}}(u)=\cL_0+\tfrac12(u-u_\star)^\top \mathsf Q(u-u_\star),
\qquad
\mathsf Q=\mathsf Q_\parallel\oplus\mu I_{D_{\br}},
\qquad \mathsf Q_\parallel\succ0,
\qquad \mu>0,
\end{equation}
where the first block spans the $D_{\sym}$ retained coordinates and the second the $D_{\br}$ breaking coordinates, and write $u_\perp$ for the breaking block of $u$. Train with the soft penalty $\tfrac{\lambda}{2}\norm{u_\perp}_2^2$, write $\alpha=\lambda/\mu$ and $s=s(\alpha)=1/(1+\alpha)$ as in \eqref{eq:soft-scale}, and let the additive gradient noise obey $\bbE\xi=0$ and $\Cov(\xi)=\sigma^2\mathsf Q/M$. Define the noisy penalized solution by $\widehat u_\alpha=(\mathsf Q+\lambda P_\perp)^{-1}(\mathsf Q u_\star-\xi)$, where $P_\perp$ projects onto the breaking coordinates. Put $\mathcal B^2=\tfrac12\mu\norm{u_{\star\perp}}_2^2>0$ for the task signal carried by the breaking block. Then for every $M>0$ the expected excess risk is exactly
\begin{equation}
\bbE\bigl[\cL_{\mathrm{quad}}(\widehat u_\alpha)\bigr]-\cL_0
=\mathcal B^2(1-s)^2
+\frac{\sigma^2}{2M}\bigl(D_{\sym}+s^2D_{\br}\bigr)
,
\label{eq:ridge-risk}
\end{equation}
and the minimizer over $\alpha\ge0$ is
\begin{equation}
\alpha_\star=\frac{\sigma^2D_{\br}}{2M\mathcal B^2},
\qquad
s_\star=\frac{1}{1+\alpha_\star}.
\label{eq:alpha-star}
\end{equation}
If the breaking block is opened with individual gates $s_j$ rather than one common gate, the variance term becomes $\tfrac{\sigma^2}{2M}\bigl(D_{\sym}+\sum_js_j^2\bigr)$, which is the effective breaking dimension
\begin{equation}
D_{\br}(S,\alpha)=\sum_{j}s_j(\alpha_j)^2 .
\label{eq:soft-shot-dim-origin}
\end{equation}
\end{proposition}

\begin{proof}
The breaking block of the solution is $\widehat u_{\alpha\perp}=s u_{\star\perp}-s\xi_\perp/\mu$. Its mean error contributes $\tfrac12\mu(1-s)^2\norm{u_{\star\perp}}^2=\mathcal B^2(1-s)^2$, and its covariance contributes $s^2\sigma^2D_{\br}/(2M)$. On the retained block, the error covariance is $\sigma^2\mathsf Q_\parallel^{-1}/M$, so its risk contribution is $\tfrac12\Tr(\mathsf Q_\parallel\sigma^2\mathsf Q_\parallel^{-1}/M)=\sigma^2D_{\sym}/(2M)$. Adding the terms gives \eqref{eq:ridge-risk}; individual gates give \eqref{eq:soft-shot-dim-origin}. Differentiation with respect to $s$ yields $-2\mathcal B^2(1-s)+\sigma^2sD_{\br}/M=0$, proving \eqref{eq:alpha-star}.
\end{proof}

For one breaking coordinate, this calculation agrees with the scalar Fisher-mode ridge rule in \cite{ChenQNGQuotient2026}: if the mode has Fisher eigenvalue $f$, gradient component $c$, and gradient-noise variance $\sigma_g^2$, its optimal damping is $\gamma_\star=\sigma_g^2 f/c^2$. The conventions coincide under $\alpha=\gamma/f$, $\sigma_g^2=\sigma^2f/M$, and $\mathcal B^2=c^2/(2f)$. The block calculation makes the dependence on the number and amplitudes of opened coordinates explicit.

\section{Quantum-Native Gate before Local Release}
\label{sec:quantum-native}

A global release gate asks whether the symmetry-breaking readout space contains label information before an implementable direction is chosen. The group twirl makes this an overlap-estimation problem: subtracting the twirled overlap removes the symmetric component, and character weights resolve the remaining signal by representation sector.

For normalized Haar or counting measure, inversion invariance implies that the Heisenberg and Schr\"odinger twirls in \eqref{eq:twirl} are the same Hilbert--Schmidt orthogonal projector when acting on matrices. We write this projector as \(\mathcal T\).

\subsection{Task-conditional breaking mass}

\begin{definition}[Breaking mass]
\label[definition]{def:breaking-mass}
For binary class-conditional states \(\rho_+\) and \(\rho_-\), let \(\Delta=\rho_+-\rho_-\) and define
\begin{equation}
D_H^{\mathrm{task}}
=\norm{(\mathrm{id}-\mathcal T)\Delta}_{\HS}^2.
\label{eq:task-breaking-mass}
\end{equation}
For a single state \(\rho\), define the auxiliary state-asymmetry mass
\begin{equation}
D_H(\rho)=\norm{(\mathrm{id}-\mathcal T)\rho}_{\HS}^2.
\label{eq:state-breaking-mass}
\end{equation}
\end{definition}

The single-state quantity is one half of the group-averaged squared commutator asymmetry used in Hilbert--Schmidt symmetry testing \cite{Bandyopadhyay2023}. The task quantity is different: it measures only symmetry breaking that separates the labels.

\begin{lemma}[Projection identity and extremal task signal]
\label{lem:breaking-mass-duality}
The twirl is self-adjoint and idempotent in the Hilbert--Schmidt inner product. Consequently,
\begin{align}
D_H^{\mathrm{task}}
&=\Tr(\Delta^2)-\Tr\bigl(\Delta\mathcal T(\Delta)\bigr),
\label{eq:task-breaking-overlaps}\\
D_H(\rho)
&=\Tr(\rho^2)-\Tr\bigl(\rho\mathcal T(\rho)\bigr).
\label{eq:state-breaking-overlaps}
\end{align}
Moreover,
\begin{equation}
\sqrt{D_H^{\mathrm{task}}}
=\max_{\substack{A=A^\dagger,\ \norm{A}_{\HS}\le1\\ \mathcal T(A)=0}}
\abs{\Tr(A\Delta)}.
\label{eq:breaking-mass-duality}
\end{equation}
If \(D_H^{\mathrm{task}}>0\), the maximum is attained by the normalized operator \((\mathrm{id}-\mathcal T)\Delta\), up to sign.
\end{lemma}

\begin{proof}
Haar invariance under \(g\mapsto g^{-1}\) gives \(\ip{A}{\mathcal T(B)}=\ip{\mathcal T(A)}{B}\), and averaging twice gives \(\mathcal T^2=\mathcal T\). Hence
\begin{equation}
\norm{(\mathrm{id}-\mathcal T)X}_{\HS}^2
=\Tr(X^\dagger X)-\Tr\bigl(X^\dagger\mathcal T(X)\bigr),
\end{equation}
which proves \eqref{eq:task-breaking-overlaps} and \eqref{eq:state-breaking-overlaps}. For every admissible \(A\), self-adjointness gives
\begin{equation}
\Tr(A\Delta)=\ip{A}{(\mathrm{id}-\mathcal T)\Delta}.
\end{equation}
Cauchy--Schwarz gives the upper bound in \eqref{eq:breaking-mass-duality}, and the normalized projected difference attains it.
\end{proof}

Thus \(D_H^{\mathrm{task}}=0\) if and only if every Hilbert--Schmidt-normalized breaking observable has zero label gap. If \(D_H^{\mathrm{task}}\ge r^2\), an observable with \(\norm{A}_{\op}\le\norm{A}_{\HS}\le1\) has gap at least \(r\), so \cref{cor:binary-release-separation} supplies the corresponding readout guarantee. This exact duality is the link between a global symmetry test and the release decision.

The classical--quantum discrepancy used for symmetry discovery measures symmetry of the joint input--label state \cite{CompanionDiscovery}. Its group average equals the breaking mass in that joint representation. With a trivial label action, exact joint invariance implies zero binary task breaking mass; approximate invariance gives a bound weighted by the class probabilities. A sign-equivariant binary task can carry odd-sector label signal even when its joint symmetry test passes.

The normalization determines which observable signals the gate controls. An observable with outcomes in $[-1,1]$ can have Hilbert--Schmidt norm as large as $\sqrt d$. Its strongest breaking signal is characterized by a different norm duality.

\begin{proposition}[Breaking signal for bounded-outcome readouts]
\label[proposition]{prop:operational-gap}
Let $\Delta_{\perp}=(I-\mathcal T)\Delta$ and define
\begin{equation}
\Gamma_H=\max_{\substack{A=A^\dagger,\ \|A\|_{\op}\le1\\\mathcal T(A)=0}}|\Tr(A\Delta)|.
\end{equation}
Then
\begin{equation}
\Gamma_H=\inf_{Z=Z^\dagger\in\Comm(H)}\|\Delta-Z\|_1,
\qquad
\sqrt{D_H^{\mathrm{task}}}\le\Gamma_H\le\|\Delta_{\perp}\|_1
\le\sqrt{\rank(\Delta_{\perp})D_H^{\mathrm{task}}}.
\label{eq:operational-gap}
\end{equation}
For an involutive conjugation action $\mathcal T(X)=(X+UXU^\dagger)/2$, one has $\Gamma_H=\|\Delta_{\perp}\|_1$. In particular, if $P$ is a traceless Hermitian Pauli and $UPU^\dagger=-P$, the states $\rho_\pm=(I\pm aP)/d$, $0<a\le1$, satisfy
\begin{equation}
D_H^{\mathrm{task}}=4a^2/d,\qquad \Gamma_H=2a.
\label{eq:mixed-gap-separation}
\end{equation}
\end{proposition}

\begin{proof}
On the real space of Hermitian matrices, the operator norm is dual to the trace norm. Restricting its unit ball to $\ker\mathcal T$ makes the support function the quotient trace norm modulo $(\ker\mathcal T)^\perp=\Comm(H)\cap\Herm(\cH)$. Equivalently, Lagrange duality for $\mathcal T(A)=0$ gives the infimum in \eqref{eq:operational-gap}; $A=0$ is a strictly feasible point of the norm ball, so finite-dimensional convex duality gives equality. The lower bound uses $A=\Delta_{\perp}/\|\Delta_{\perp}\|_{\HS}$ when $\Delta_{\perp}\ne0$. Taking $Z=\mathcal T\Delta$ and applying Cauchy--Schwarz to the singular values gives the upper bounds. For an involution, $U\Delta_{\perp}U^\dagger=-\Delta_{\perp}$, so its matrix sign, with value zero on its kernel, is also odd. This admissible observable attains $\Tr(\operatorname{sign}(\Delta_{\perp})\Delta)=\|\Delta_{\perp}\|_1$. The Pauli states give \eqref{eq:mixed-gap-separation} directly.
\end{proof}

An upper bound $u_H$ on the mass excludes Hilbert--Schmidt-normalized gaps above $\sqrt{u_H}$. For all operator-norm-bounded readouts, it gives the bound $\sqrt{r_Hu_H}$ whenever $r_H\ge\rank(\Delta_{\perp})$; $r_H=d$ always suffices. The mixed-state example attains this rank dependence. Retention is therefore certified relative to both the declared readout norm and the task-relevance scale.

\subsection{Two-copy twirl--swap gate}

The class-conditional states are population mean states when labeled inputs themselves vary. Preparing each register from a fresh independent example in its requested class gives the product-state expectations used below. If a finite data set is reused, the shot interval instead targets its empirical class means. To transfer that interval to the population, let $\bar\rho_\pm$ average $N_\pm$ independent labeled states and suppose $\|\bar\rho_y-\rho_y\|_{\HS}\le e_y$. Since the difference of two density matrices has Hilbert--Schmidt norm at most $\sqrt2$, contraction of the twirl complement gives
\begin{equation}
\left|\|(I-\mathcal T)(\bar\rho_+-\bar\rho_-)\|_{\HS}^2-D_H^{\mathrm{task}}\right|
\le 2\sqrt2\,(e_++e_-).
\label{eq:task-data-transfer}
\end{equation}
For example, $\bbE\|\bar\rho_y-\rho_y\|_{\HS}^2\le1/N_y$ and Markov's inequality give the simultaneous choice $e_y=\sqrt{2/(N_y\delta_{\mathrm{data}})}$ with failure probability at most $\delta_{\mathrm{data}}$. Adding \eqref{eq:task-data-transfer} to the shot radius gives a population certificate under data reuse; sharper Hilbert-space concentration can replace this simple sampling bound.

Assume that independent copies of the relevant states can be prepared, a Haar-random group element can be sampled and implemented, and the SWAP observable can be measured on a pair. For qubit registers the last operation is a destructive Bell-basis measurement with classical parity post-processing \cite{Bandyopadhyay2023}.

\begin{theorem}[Twirl--swap release gate]
\label{thm:twirl-swap-gate}
For a single state \(\rho\), let \(X_i\in\{-1,1\}\) be a SWAP outcome on \(\rho\otimes\rho\), and let \(Y_i\in\{-1,1\}\) be a SWAP outcome on \(\rho\otimes U_{g_i}\rho U_{g_i}^\dagger\), where \(g_i\) is Haar-random. With \(M\) pairs in each arm,
\begin{equation}
\widehat D_H=\frac1M\sum_{i=1}^M(X_i-Y_i)
\end{equation}
is unbiased for \(D_H(\rho)\). The threshold test \(\varphi=\mathbf1\{\widehat D_H>r^2/2\}\) satisfies
\begin{equation}
\bbP_{D_H=0}\{\varphi=1\}
+\sup_{D_H\ge r^2}\bbP\{\varphi=0\}
\le8\exp\left(-\frac{Mr^4}{32}\right).
\label{eq:twirl-swap-state-error}
\end{equation}

For the task gate, draw \(A_i,B_i\) independently and uniformly from \(\{+,-\}\), set \(c_+=1\), \(c_-=-1\), and multiply each SWAP outcome by \(4c_{A_i}c_{B_i}\). Use \(\rho_{A_i}\otimes\rho_{B_i}\) in the first arm and \(\rho_{A_i}\otimes U_{g_i}\rho_{B_i}U_{g_i}^\dagger\) in the second. The difference of the two empirical means, denoted \(\widehat D_H^{\mathrm{task}}\), is unbiased for \(D_H^{\mathrm{task}}\), and
\begin{equation}
\bbP_{D_H^{\mathrm{task}}=0}\{\widehat D_H^{\mathrm{task}}>r^2/2\}
+\sup_{D_H^{\mathrm{task}}\ge r^2}
\bbP\{\widehat D_H^{\mathrm{task}}\le r^2/2\}
\le8\exp\left(-\frac{Mr^4}{512}\right).
\label{eq:twirl-swap-task-error}
\end{equation}
Both gates use \(4M\) state copies and have \(M=O(r^{-4}\log(1/\delta))\) pair complexity independent of the Hilbert dimension, the breaking dimension, and the dictionary size. This dimension independence concerns copy-pair count conditional on state preparation, group-unitary access, and SWAP measurement. The access model separately determines implementation and gate costs.
\end{theorem}

\begin{proof}
The SWAP identity gives \(\bbE X_i=\Tr(\rho^2)\) and
\begin{equation}
\bbE Y_i=\int_H\Tr(\rho U_g\rho U_g^\dagger)\,\dd g
=\Tr\bigl(\rho\mathcal T(\rho)\bigr),
\end{equation}
so \eqref{eq:state-breaking-overlaps} proves unbiasedness. Each arm is an average of variables in \([-1,1]\). If their errors are both at most \(r^2/4\), then the error of their difference is at most \(r^2/2\). Hoeffding's inequality and a union bound over two arms and two testing errors give \eqref{eq:twirl-swap-state-error}.

For the task estimator, uniform label sampling gives
\begin{equation}
\bbE\bigl[4c_Ac_BX\bigr]
=\sum_{a,b\in\{+,-\}}c_ac_b\Tr(\rho_a\rho_b)
=\Tr(\Delta^2),
\end{equation}
and the twirled arm similarly has mean \(\Tr(\Delta\mathcal T(\Delta))\). The reweighted variables lie in \([-4,4]\), so the same argument with range eight gives \eqref{eq:twirl-swap-task-error}.
\end{proof}

The same proof yields confidence intervals useful for selection. With probability at least \(1-\delta\),
\begin{equation}
\abs{\widehat D_H-D_H}
\le\sqrt{\frac{8\log(4/\delta)}{M}},
\qquad
\abs{\widehat D_H^{\mathrm{task}}-D_H^{\mathrm{task}}}
\le\sqrt{\frac{128\log(4/\delta)}{M}}.
\label{eq:breaking-mass-confidence}
\end{equation}
An upper confidence bound below a predeclared relevance level \(r_{\min}^2\) certifies retention at that Hilbert--Schmidt-normalized signal scale. An interval containing the relevance threshold is reported as ambiguous; exact retention at zero signal requires an upper bound of zero.

\begin{proposition}[Gate noise floor]
\label[proposition]{prop:twirl-swap-noise}
Suppose the implemented operators $\widetilde U_g$ are unitary, \(\norm{\widetilde U_g-U_g}_{\op}\le\eta_U\) for every implemented group element and each single-state SWAP outcome is independently flipped with known probability \(\eta_m<1/2\). Let \(\widehat D_{\mathrm{deb}}=\widehat D_H/(1-2\eta_m)\). Then
\begin{equation}
\abs{\bbE\widehat D_{\mathrm{deb}}-D_H(\rho)}\le2\eta_U.
\end{equation}
For \(r^2\ge8\eta_U\), the unshifted test \(\mathbf1\{\widehat D_{\mathrm{deb}}>r^2/2\}\) obeys
\begin{equation}
\bbP_{D_H=0}\{\varphi=1\}
+\sup_{D_H\ge r^2}\bbP\{\varphi=0\}
\le8\exp\left(-\frac{M(1-2\eta_m)^2r^4}{128}\right).
\label{eq:twirl-swap-noise}
\end{equation}
\end{proposition}

\begin{proof}
Outcome flipping multiplies both arm means by \(1-2\eta_m\). For the group-unitary error,
\begin{equation}
\widetilde U\rho\widetilde U^\dagger-U\rho U^\dagger
=(\widetilde U-U)\rho\widetilde U^\dagger
+U\rho(\widetilde U-U)^\dagger,
\end{equation}
whose trace norm is at most \(2\eta_U\). The mean overlap therefore shifts by at most \(2\eta_U\). Under either hypothesis, the debiased mean remains at least \(r^2/4\) from the threshold when \(r^2\ge8\eta_U\). Hoeffding's inequality applied before debiasing gives \eqref{eq:twirl-swap-noise}.
\end{proof}

\subsection{Shared-sample sector spectroscopy}

Let \(\pi(g)X=U_gXU_g^\dagger\) denote the conjugation representation on \(\cB(\cH)\). Although a compact group may have infinitely many irreducible representations, only finitely many occur in this \(d^2\)-dimensional space. Define
\begin{equation}
\Lambda_{\mathrm{act}}
=\{\lambda:P_\lambda\neq0\text{ on }\cB(\cH)\},
\qquad |\Lambda_{\mathrm{act}}|\le d^2,
\end{equation}
where
\begin{equation}
P_\lambda
=d_\lambda\int_H\overline{\chi_\lambda(g)}\,\pi(g)\,\dd g
=d_\lambda\int_H\chi_\lambda(g^{-1})\,\pi(g)\,\dd g
\label{eq:isotypic-projector}
\end{equation}
is the central isotypic projector, $d_\lambda$ is the irreducible-representation dimension, and $\chi_\lambda$ is its character \cite{FultonHarris1991}. The conjugated character fixes the sector convention; the unconjugated integral selects the dual sector. Coherent symmetry-projector circuits give a complementary route \cite{LaBordeRethinasamyWilde2024}, while the estimator below resolves all sector masses from one destructive-measurement sample set.

\begin{theorem}[Sector-resolved breaking spectroscopy]
\label{thm:sector-spectroscopy}
For \(\lambda\in\Lambda_{\mathrm{act}}\), define
\begin{equation}
D_\lambda(\rho)=\ip{\rho}{P_\lambda\rho}
=\norm{P_\lambda\rho}_{\HS}^2.
\end{equation}
Then
\begin{equation}
\sum_{\lambda\neq\mathrm{triv}}D_\lambda(\rho)=D_H(\rho).
\label{eq:sector-sum-rule}
\end{equation}
Let \(\{(g_i,s_i)\}_{i=1}^M\) be one twirled-arm sample set, where \(s_i\in\{-1,1\}\) is the SWAP outcome for \(\rho\otimes U_{g_i}\rho U_{g_i}^\dagger\). Reuse it for every active sector through
\begin{equation}
\widehat D_\lambda
=\frac1M\sum_{i=1}^M
d_\lambda\operatorname{Re}\chi_\lambda(g_i)s_i.
\label{eq:sector-estimator}
\end{equation}
Each estimator is unbiased, and with probability at least \(1-\delta\), simultaneously for all active nontrivial sectors,
\begin{equation}
\abs{\widehat D_\lambda-D_\lambda}
\le
\sqrt{\frac{2d_\lambda^2\log(2|\Lambda_{\mathrm{act}}|/\delta)}{M}}
+\frac{4d_\lambda^2\log(2|\Lambda_{\mathrm{act}}|/\delta)}{3M}.
\label{eq:sector-confidence}
\end{equation}
Consequently, a single budget
\begin{equation}
M=O\left(d_{\max}^2r^{-4}
\log\frac{|\Lambda_{\mathrm{act}}|}{\delta}\right),
\qquad d_{\max}=\max_{\lambda\in\Lambda_{\mathrm{act}}}d_\lambda,
\end{equation}
estimates every sector mass to accuracy \(r^2\), for $0<r\le1$. The number of active sectors enters only logarithmically.
\end{theorem}

\begin{proof}
The operators \(P_\lambda\) are mutually orthogonal Hilbert--Schmidt projectors that sum to the identity on the conjugation representation, with \(P_{\mathrm{triv}}=\mathcal T\). This proves nonnegativity and \eqref{eq:sector-sum-rule}. Moreover,
\begin{equation}
\bbE\left[d_\lambda\operatorname{Re}\chi_\lambda(g)s\right]
=\operatorname{Re}\ip{\rho}{P_\lambda\rho}
=D_\lambda(\rho).
\end{equation}
The inner product is real because $P_\lambda$ is an orthogonal projector, and $\operatorname{Re}\overline{\chi_\lambda}=\operatorname{Re}\chi_\lambda$, which justifies the estimator's real character weight. The bound $|\chi_\lambda(g)|\le d_\lambda$ and character orthogonality imply $|X_\lambda|\le d_\lambda^2$ and $\Var(X_\lambda)\le\bbE X_\lambda^2\le d_\lambda^2$. Scalar Bernstein concentration and a union bound give \eqref{eq:sector-confidence}. Dependence between sector estimates is allowed because the underlying sample pairs are independent.
\end{proof}

For a non-real irrep, $P_\lambda\Delta$ need not be Hermitian; a physical Hermitian release dictionary combines the conjugate pair $\lambda,\bar\lambda$. Their reported masses are projection energies, not two independent Hermitian readouts. The task-sector masses \(D_\lambda^{\mathrm{task}}=\norm{P_\lambda\Delta}_{\HS}^2\) use the same samples and the signed class-pair reweighting from \cref{thm:twirl-swap-gate}. The summand becomes \(4c_Ac_Bd_\lambda\operatorname{Re}\chi_\lambda(g)s\). It has variance at most \(16d_\lambda^2\), so the same simultaneous rate holds with a larger absolute constant. Sector confidence intervals may prioritize the local dictionary, but a sector is excluded only when its upper confidence bound lies below the predeclared relevance scale.

\subsection{Access-model boundary and limits of the gate}

Suppose the local Gaussian experiment is calibrated so that \(\norm{\mu}^2=D_H^{\mathrm{task}}\) on a \(k\)-dimensional branch and its shot variance is \(\tau^2=\sigma^2/M\). At fixed confidence, \cref{thm:detection-boundary} then requires
\begin{equation}
D_H^{\mathrm{task}}\gtrsim\frac{\sigma^2\sqrt{k}}{M},
\end{equation}
whereas \cref{thm:twirl-swap-gate} requires
\begin{equation}
D_H^{\mathrm{task}}\gtrsim\sqrt{\frac{\log(1/\delta)}{M}}.
\end{equation}
The paired-copy gate is favorable in the high-dimensional, shot-limited regime
\begin{equation}
M\lesssim\frac{\sigma^4k}{\log(1/\delta)}.
\end{equation}
Direction-resolved measurements become sharper after enough shots. This crossover is why the algorithm treats the gate as an access-dependent screening step rather than a mandatory test.

The Pauli group makes the role of quantum memory explicit. If \(H_P\) is the \(n\)-qubit Pauli group modulo phases and \(d=2^n\), then
\begin{equation}
\mathcal T_{H_P}(\rho)=I/d,
\qquad
D_{H_P}(\rho)=\Tr(\rho^2)-1/d.
\label{eq:pauli-purity-reduction}
\end{equation}
Thus Pauli-symmetry testing contains purity testing. Distinguishing an unknown pure state from $I/d$ requires $\Omega(\sqrt d)$ copies under arbitrary adaptive single-copy measurements without quantum memory, whereas paired-copy SWAP measurements use a constant number at fixed confidence \cite{ChenCotlerHuangLi2021}.

A structured alternative can admit a sharper measurement strategy. The Pauli-spike construction below detects a breaking mass $r^2$ with a copy count of order $r^{-2}$ up to logarithms, showing that the twirl--swap rate depends on the access model and the alternative family.

\begin{proposition}[Point-null Pauli scan]
\label[proposition]{prop:pauli-point-null}
Let \(d=2^n\), \(\rho_0=I/d\), and let \(P_1,\ldots,P_{d^2-1}\) be the nonidentity Hermitian Pauli operators. For
\begin{equation}
\rho_{J,s}=\frac{I+s r\sqrt d\,P_J}{d},
\qquad s\in\{+1,-1\},
\qquad 0<r\le d^{-1/2},
\end{equation}
one has \(D_{H_P}(\rho_{J,s})=r^2\). There is a preselected nonadaptive measurement performed separately on each copy that distinguishes \(\rho_0\) from every \(\rho_{J,s}\), with false-positive plus worst-case false-negative probability at most \(\delta\), using
\begin{equation}
N=(d+1)\left\lceil
\frac{8\log(4d^2/\delta)}{dr^2}
\right\rceil
=O\left(\frac{\log(d/\delta)}{r^2}\right)
\label{eq:pauli-point-null-copies}
\end{equation}
single copies.
\end{proposition}

\begin{proof}
The nonidentity Paulis partition into \(d+1\) maximal commuting classes, whose joint eigenbases form a complete set of mutually unbiased stabilizer bases \cite{WoottersFields1989}. Allocate \(n_c=N/(d+1)\) copies to each basis and, for every Pauli \(P\) in its aligned class, form the empirical mean \(T_P\) of its \(\pm1\) eigenvalue. Under \(\rho_0\), every mean is zero. Under \(\rho_{J,s}\), the aligned statistic has mean \(sr\sqrt d\), while all other statistics have mean zero. With
\begin{equation}
\tau=\sqrt{\frac{2\log(4d^2/\delta)}{n_c}},
\end{equation}
Hoeffding's inequality and a union bound control \(\max_P\abs{T_P}\) under the null by \(\delta/2\). The sample size in \eqref{eq:pauli-point-null-copies} gives \(r\sqrt d\ge2\tau\), so the aligned statistic exceeds \(\tau\) with missed-detection probability at most \(\delta/2\).
\end{proof}

Collective spectrum testing distinguishes $I/d$ from arbitrary trace-distance-$\varepsilon$ alternatives with $\Theta(d/\varepsilon^2)$ copies \cite{ODonnellWright2021}. Substituting the trace distance $\varepsilon=r\sqrt d/2$ gives an $O(r^{-2})$ upper bound for the Pauli-spike family; the general spectrum-testing lower bound does not by itself apply to this restricted family. The twirl--swap gate offers dimension-independent copy complexity and exact removal of an unknown symmetric backbone. Indeed, for every $H$-symmetric $\rho_{\sym}$ and Hermitian $V$ with $\mathcal T(V)=0$ and $\norm{V}_{\HS}=1$,
\begin{equation}
D_H(\rho_{\sym}+rV)=r^2
\end{equation}
whenever \(\rho_{\sym}+rV\) is a state. Establishing the minimax complexity under an unknown composite symmetric null remains open.

\section{Multi-Direction Release: Curvature Matrices and Orthogonalization}
\label{sec:multidirection}

A useful release may combine several generators even when no individual diagonal curvature is positive. The symmetrized double commutator defines a quadratic form on the entire dictionary. Its spectrum identifies favorable combinations, and Gram normalization makes their relative scale explicit.

\subsection{Symmetrized double-commutator curvature matrix}

Let $\rho_0$ be the state produced by the retained backbone at $\beta=0$. For a Hermitian loss observable $C$, consider the linear probing loss
\begin{equation}
\cL(\beta)=\Tr\bigl(Ce^{-iB(\beta)}\rho_0e^{iB(\beta)}\bigr),\qquad C=C^\dagger.
\end{equation}
A nonlinear supervised loss weights these readout curvatures by its local first derivatives and adds a second-derivative correction. Both contributions are needed to compare release directions under the actual loss.

\begin{definition}[Release curvature matrix]
\label[definition]{def:release-curvature-matrix}
Given Hermitian breaking generators $B_1,\ldots,B_p$, define
\begin{equation}
K_{ab}
=-\left.\frac{\partial^2\cL(te_a+se_b)}{\partial t\partial s}\right|_{t=s=0}
=\frac12\Tr\left(\bigl([B_a,[B_b,C]]+[B_b,[B_a,C]]\bigr)\rho_0\right).
\label{eq:curvature-matrix}
\end{equation}
For $a=b$, $K_{aa}=\Tr([B_a,[B_a,C]]\rho_0)$, recovering the one-dimensional double commutator. For a minimization objective, $v^\top Kv>0$ means that the direction $B(v)=\sum_a v_aB_a$ gives a second-order local decrease.
\end{definition}

\begin{theorem}[Nonlinear supervised release curvature]
\label{thm:nonlinear-supervised-curvature}
Consider a release-probing sample $z_i=(\rho_i,y_i)$ and vector readouts
\begin{equation}
m_{i,j}(\beta)=\Tr\bigl(O_j e^{-iB(\beta)}\rho_i e^{iB(\beta)}\bigr),
\qquad j=1,\ldots,q,
\end{equation}
with empirical probing risk
\begin{equation}
\cL_n(\beta)=\frac1n\sum_{i=1}^n \ell(m_i(\beta),y_i),
\end{equation}
where $\ell(\cdot,y)$ is twice continuously differentiable near $m_i(0)$. Define
\begin{align}
d_{i,j}&=\partial_{m_j}\ell(m_i(0),y_i),
&Q_{i,jk}&=\partial_{m_j}\partial_{m_k}\ell(m_i(0),y_i),\\
\nu_{i,a,j}&=i\Tr([B_a,O_j]\rho_i),
&D_{i,ab}^{(j)}&=\frac12\Tr\left(\bigl([B_a,[B_b,O_j]]+[B_b,[B_a,O_j]]\bigr)\rho_i\right).
\end{align}
Then the release curvature matrix for the nonlinear supervised loss is
\begin{equation}
K^{\ell}_{ab}
=-\left.\partial_a\partial_b\cL_n(\beta)\right|_{0}
=\frac1n\sum_{i=1}^n
\left[
\sum_{j=1}^q d_{i,j}D_{i,ab}^{(j)}
-
\sum_{j,k=1}^q Q_{i,jk}\nu_{i,a,j}\nu_{i,b,k}
\right].
\label{eq:nonlinear-curvature}
\end{equation}
For a linear validation observable, or whenever all first-order readout responses $\nu_{i,a,j}$ vanish in the probed branch, the second term disappears and \eqref{eq:nonlinear-curvature} reduces to the double-commutator curvature with local observable $C_i=\sum_j d_{i,j}O_j$.
\end{theorem}

\begin{proof}
The first readout derivative is
\begin{equation}
\partial_a m_{i,j}(0)=i\Tr([B_a,O_j]\rho_i)=\nu_{i,a,j}.
\end{equation}
The mixed second derivative is
\begin{equation}
\partial_a\partial_b m_{i,j}(0)
=-\frac12\Tr\left(\bigl([B_a,[B_b,O_j]]+[B_b,[B_a,O_j]]\bigr)\rho_i\right)
=-D_{i,ab}^{(j)}.
\end{equation}
Applying the chain rule to $\ell(m_i(\beta),y_i)$ gives
\begin{equation}
\partial_a\partial_b\ell_i(0)
=\sum_j d_{i,j}\partial_a\partial_bm_{i,j}(0)
+\sum_{j,k}Q_{i,jk}\partial_am_{i,j}(0)\partial_bm_{i,k}(0).
\end{equation}
Negating and averaging over $i$ proves \eqref{eq:nonlinear-curvature}.
\end{proof}

\begin{remark}[Measurement consequence]
\label[remark]{rem:nonlinear-measurement}
For nonlinear losses, estimating only the double-commutator observable can overstate release benefit when $Q_i\succeq0$, because the Gauss--Newton correction in \eqref{eq:nonlinear-curvature} is negative semidefinite in the release direction. In experiments, the products $\nu_{i,a,j}\nu_{i,b,k}$ should be computed from exact derivatives or from independently split shot records to avoid same-record product bias; noisy parameter-shift estimates require the same independence or bias correction. Estimation error in the loss derivatives $d_{i,j}$ and $Q_{i,jk}$ must also be propagated when these are evaluated at noisy readout estimates.
\end{remark}

\begin{lemma}[Multi-direction release quadratic form]
\label[lemma]{thm:quadratic-release}
Assume the linear loss above is stationary in the full release block, $\nabla_\beta\cL(0)=0$. Then
\begin{equation}
\cL(\beta)=\cL(0)-\frac12\beta^\top K\beta+O\bigl(\norm{B(\beta)}_{\op}^3\norm{C}_{\op}\bigr).
\label{eq:quadratic-release}
\end{equation}
Ignoring finite-sample costs, a strictly negative quadratic loss variation exists if and only if $K$ has a positive eigenvalue. A zero quadratic mode requires higher-order analysis.
\end{lemma}

\begin{proof}
Move the unitary action to the observable and use Baker--Campbell--Hausdorff:
\begin{equation}
 e^{iB(\beta)}Ce^{-iB(\beta)}=C+i[B(\beta),C]-\frac12[B(\beta),[B(\beta),C]]+O(\norm{B(\beta)}^3\norm{C}).
\end{equation}
The mixed second derivatives are symmetrized because the Hessian is symmetric. If the first-order term is zero, the sign of the quadratic term determines local descent.
\end{proof}

\subsection{Local release gain with first-order terms}
\label{subsec:first-order-release}

First-order stationarity must hold in the full probed release block. It follows, for example, when the state and loss are $H$-invariant and the generators lie entirely in nontrivial sectors, or when the full release gradient has been verified to vanish. Training only the hard backbone establishes stationarity within the retained class and leaves breaking gradients unconstrained. The following trust-region certificate includes those gradients and their measurement uncertainty.

Let $S_M\succeq0$ and $G_N\succeq0$ be declared quadratic activation and structural costs, and define the penalized objective $\cR_{N,M}(\beta)=\cL(\beta)+\tfrac12\beta^\top S_M\beta+\beta^\top G_N\beta$. Their coefficients may be calibrated by a local risk model such as \cref{prop:ridge-scale}; they are distinct from confidence radii. Define
\begin{equation}
\ell_a=-\left.\partial_{\beta_a}\cL(\beta)\right|_{0},
\qquad
A=K-S_M-2G_N,
\end{equation}
where \(\ell_a\) is the first-order descent benefit and \(A\) is the curvature benefit after those penalties. For a Gram matrix \(G\succ0\) and radius \(\rho>0\), define the population quadratic release gain
\begin{equation}
\mathfrak G_\rho(\ell,A;G)
=\sup_{\beta^\top G\beta\le \rho^2}
\left\{\ell^\top\beta+\frac12\beta^\top A\beta\right\}.
\label{eq:trust-gain}
\end{equation}
If \(\ell=0\), this reduces to \(\rho^2\lambda_{\max}(G^{-1/2}AG^{-1/2})_+/2\). If \(A\preceq0\) but \(\ell\neq0\), a first-order release benefit can still be certified.

\begin{theorem}[Safe release certificate with first-order terms]
\label{thm:first-order-safe-release}
Assume that on the Gram ball \(\{\beta:\beta^\top G\beta\le\rho^2\}\), the third-order remainder obeys
\begin{equation}
\left|\cL(\beta)-\cL(0)+\ell^\top\beta+\frac12\beta^\top K\beta\right|
\le \frac{L_3}{6}\rho^3.
\end{equation}
Assume also that the simultaneous confidence event
\begin{equation}
\norm{G^{-1/2}(\widehat A-A)G^{-1/2}}_{\op}\le r_K,
\qquad
\norm{G^{-1/2}(\widehat\ell-\ell)}_2\le r_\ell
\label{eq:first-order-ci}
\end{equation}
holds with probability at least \(1-\delta\). If
\begin{equation}
\mathfrak G_\rho(\widehat\ell,\widehat A;G)
>
\rho r_\ell+\frac12\rho^2r_K+\frac{L_3}{6}\rho^3,
\label{eq:first-order-certificate}
\end{equation}
then there exists a local release \(\beta\) whose penalized validation risk is strictly below the hard-symmetry point. The quadratic candidate is obtained from the trust-region Karush--Kuhn--Tucker optimality conditions \cite[Thm.~4.1]{NocedalWright2006}
\begin{equation}
(\nu G-\widehat A)\beta=\widehat\ell,\quad
\nu\ge\max\{0,\lambda_{\max}(G^{-1/2}\widehat A G^{-1/2})\},\quad
\beta^\top G\beta\le\rho^2,\quad
\nu(\beta^\top G\beta-\rho^2)=0.
\label{eq:trust-kkt}
\end{equation}
\end{theorem}

\begin{proof}
For any $\beta$ in the Gram ball, Cauchy--Schwarz bounds the linear estimation error by $\rho r_\ell$, and the operator-norm bound controls the quadratic error by $\rho^2r_K/2$. Thus the empirical maximizer decreases the true penalized objective whenever \eqref{eq:first-order-certificate} holds. After whitening by $G^{1/2}$, \eqref{eq:trust-kkt} is the necessary and sufficient global optimality characterization for the quadratic trust-region problem. The positive-semidefinite condition $\nu G-\widehat A\succeq0$ is essential. If this matrix is singular, the linear equation is solved on its range and a kernel component is chosen to meet the radius condition; this includes the pure-eigenvector case $\widehat\ell=0$.
\end{proof}

\subsection{Hilbert--Schmidt and Fisher orthogonalization}

Breaking dictionaries are often redundant and physically nonuniform. We use a Gram matrix before interpreting the spectrum:
\begin{equation}
G^{\HS}_{ab}=\Tr(B_aB_b),
\qquad
G^F_{ab}=\operatorname{Re}\Tr(\rho_0L_aL_b),
\label{eq:gram}
\end{equation}
where $L_a$ is the symmetric logarithmic derivative determined by
\begin{equation}
-i[B_a,\rho_0]=\frac12(\rho_0L_a+L_a\rho_0).
\end{equation}
Thus $G^F$ is the quantum Fisher Gram of the unitary tangent directions. For a pure probing state it reduces to $G^F_{ab}=4\operatorname{Re}\Cov_{\rho_0}(B_a,B_b)$. At an invariant backbone, the representation-theoretic release theorem of \cite{CompanionDiscovery} makes retained and nontrivial breaking Fisher blocks orthogonal, under its stated tangent assumptions. Statistical selection within a block still uses its full curvature matrix. The matrix $G^{\HS}$ is an offline operator normalization, whereas $G^F$ measures local physical distinguishability in the current state. After small Gram eigenvalues are removed or merged, release directions are obtained from
\begin{equation}
Kv=\lambda Gv,
\qquad v^\top Gv=1.
\label{eq:generalized-eig}
\end{equation}
If $G^F$ has null directions, those directions are not locally distinguishable at $\rho_0$ and should not be released solely from a point estimate. They require a different probe state, a different validation point, or quotient treatment.

\section{Finite-Shot Certified Release}
\label{sec:finite-shot}

The curvature matrix of the previous section is known only through a finite-shot estimate, and the release decision must stay valid against that estimation error. Penalizing the spectrum by a confidence radius is what converts an empirical curvature into a certified subspace and a provable descent step.

\subsection{Penalized spectral release}

At a release-stationary point, assume additionally that the loss is four times continuously differentiable and its cubic derivative vanishes on the release block. This holds for a locally even release loss. The penalized objective then has the expansion
\begin{equation}
\cR_{N,M}(\beta)=\cR_{N,M}(0)
-\frac12\beta^\top K\beta
+\frac12\beta^\top S_M\beta
+\beta^\top G_N\beta
+\frac1{24}\mathcal Q(\beta)+o(\norm\beta^4),
\label{eq:local-risk-release}
\end{equation}
where $\mathcal Q(\beta)=D^4\cL(0)[\beta,\beta,\beta,\beta]$. A positive stabilizing quartic is an additional assumption used only for the amplitude surrogate below. The quadratic spectral criterion also holds for a general stationary loss with an $O(\|\beta\|^3)$ remainder. The coefficient convention makes the final quadratic matrix $K-S_M-2G_N$.

\begin{proposition}[Penalized spectral release criterion]
\label[proposition]{thm:penalized-spectrum}
In the local model \eqref{eq:local-risk-release}, the hard-symmetry point $\beta=0$ is strictly unstable at quadratic order in the release block if and only if
\begin{equation}
\lambda_{\max}(K-S_M-2G_N)>0.
\label{eq:release-eig-criterion}
\end{equation}
With a Gram normalization $G\succ0$, the equivalent criterion is
\begin{equation}
\lambda_{\max}\left(G^{-1/2}(K-S_M-2G_N)G^{-1/2}\right)>0.
\end{equation}
\end{proposition}

\begin{proof}
The second-order term in \eqref{eq:local-risk-release} is $-\frac12\beta^\top(K-S_M-2G_N)\beta$. A positive eigenvalue gives a direction of negative second-order change; if the matrix is negative semidefinite, no second-order release decrease exists. This does not decide higher-order descent on its kernel. The Gram form is the change of variables $u=G^{1/2}\beta$.
\end{proof}

\subsection{Simultaneous spectral confidence and subspace recovery}

Only an estimator $\widehat K$ is available. Assume a high-probability operator-norm radius
\begin{equation}
\bbP\{\norm{\widehat K-K}_{\op}\le r_\delta\}\ge1-\delta.
\label{eq:op-confidence}
\end{equation}
With a certified matrix variance proxy and a deterministic sample bound, matrix Bernstein \cite{Tropp2012} gives the radius in \cref{thm:shadow-curvature}. Without those inputs, simultaneous scalar confidence intervals supply the entrywise envelope below. Scalar finite differences additionally require their truncation error. The radius accounts for the selection bias introduced by searching over many directions.

\begin{corollary}[Simultaneous release certificate]
\label[corollary]{cor:spectral-certificate}
Let $A=K-S_M-2G_N$ and $\widehat A=\widehat K-S_M-2G_N$. On the event \eqref{eq:op-confidence}:
\begin{enumerate}
\item If $\lambda_j(\widehat A)>r_\delta$, then $A$ has at least $j$ positive eigenvalues.
\item If $\lambda_1(\widehat A)\le -r_\delta$, then $A\preceq0$ and no strictly quadratic release gain exists in the search space.
\item If $\abs{\lambda_j(\widehat A)}\le r_\delta$, the sign of the corresponding ordered population eigenvalue is statistically ambiguous and should not be forced open.
\end{enumerate}
\end{corollary}

\begin{proof}
Weyl's inequality gives $\abs{\lambda_j(\widehat A)-\lambda_j(A)}\le\norm{\widehat K-K}_{\op}\le r_\delta$.
\end{proof}

\begin{theorem}[Certified descent along an empirical release direction]
\label{thm:certified-release-descent}
Work in the Gram-whitened release coordinates. Assume that $\nabla_\beta\cL(0)=0$ and that, for every unit vector $v$ in the release dictionary, the population validation loss satisfies
\begin{equation}
\cL(\eta v)\le \cL(0)-\frac{\eta^2}{2}v^\top Kv+\frac{L_3}{6}|\eta|^3,
\qquad 0\le \eta\le \eta_0.
\label{eq:descent-smoothness}
\end{equation}
Let $\widehat v_1$ be a unit leading eigenvector of $\widehat A=\widehat K-S_M-2G_N$ with eigenvalue $\widehat\lambda_1$. On the event \eqref{eq:op-confidence}, define
\begin{equation}
\Delta_{\mathrm{desc}}
=\widehat\lambda_1-r_\delta+
\widehat v_1^\top(S_M+2G_N)\widehat v_1.
\label{eq:certified-descent-gap}
\end{equation}
If $L_3>0$, $\Delta_{\mathrm{desc}}>0$ and $2\Delta_{\mathrm{desc}}/L_3\le\eta_0$, then the step
\begin{equation}
\beta_{\mathrm{desc}}=\eta_{\mathrm{desc}}\widehat v_1,
\qquad
\eta_{\mathrm{desc}}=\frac{2\Delta_{\mathrm{desc}}}{L_3},
\end{equation}
guarantees the population loss decrease
\begin{equation}
\cL(0)-\cL(\beta_{\mathrm{desc}})
\ge
\frac{2}{3L_3^2}\Delta_{\mathrm{desc}}^3.
\label{eq:certified-descent-drop}
\end{equation}
In particular, $\widehat\lambda_1>r_\delta$ ensures $\Delta_{\mathrm{desc}}>0$ when $S_M\succeq0$ and $G_N\succeq0$.
\end{theorem}

\begin{proof}
Since $\widehat v_1$ is an eigenvector of $\widehat A$,
\begin{equation}
\widehat v_1^\top\widehat K\widehat v_1
=\widehat\lambda_1+
\widehat v_1^\top(S_M+2G_N)\widehat v_1.
\end{equation}
The confidence event gives
\begin{equation}
\widehat v_1^\top K\widehat v_1
\ge
\widehat v_1^\top\widehat K\widehat v_1-r_\delta
=\Delta_{\mathrm{desc}}.
\end{equation}
Substituting $v=\widehat v_1$ into \eqref{eq:descent-smoothness} yields
\begin{equation}
\cL(\eta\widehat v_1)\le
\cL(0)-\frac{\eta^2}{2}\Delta_{\mathrm{desc}}+\frac{L_3}{6}\eta^3.
\end{equation}
The right-hand side is minimized at $\eta=2\Delta_{\mathrm{desc}}/L_3$. Evaluating the resulting lower bound on $\cL(0)-\cL(\eta\widehat v_1)$ gives \eqref{eq:certified-descent-drop}.
\end{proof}

\begin{remark}[Computable third-order constant for circuit release]
\label[remark]{rem:analytic-L3}
The smoothness constant in \cref{thm:certified-release-descent} is directly computable for a one-parameter quantum release. Let $B(v)=\sum_av_aB_a$ be Hermitian and define
\begin{equation}
\cL_v(\beta)=\Tr\!\left(Ce^{-i\beta B(v)}\rho_0e^{i\beta B(v)}\right),
\qquad
\rho_v(\beta)=e^{-i\beta B(v)}\rho_0e^{i\beta B(v)}.
\end{equation}
With $\operatorname{ad}_B(X)=[B,X]$, differentiating in the Heisenberg picture gives
\begin{equation}
\frac{\dd^3}{\dd\beta^3}\cL_v(\beta)
=-i\Tr\!\left(\rho_v(\beta)\operatorname{ad}_{B(v)}^3(C)\right).
\label{eq:L3-commutator}
\end{equation}
Since $\norm{\rho_v(\beta)}_1=1$ and $\norm{[X,Y]}_{\op}\le2\norm{X}_{\op}\norm{Y}_{\op}$,
\begin{equation}
\sup_{\beta}\left|\frac{\dd^3}{\dd\beta^3}\cL_v(\beta)\right|
\le
8\norm{B(v)}_{\op}^3\norm{C}_{\op}.
\label{eq:L3-operator-bound}
\end{equation}
Thus the descent certificate may use the uniform analytic bound
\begin{equation}
L_3=8\norm{C}_{\op}\sup_{\norm{v}_2=1}\norm{B(v)}_{\op}^3.
\end{equation}
Here $v$ denotes a unit direction in the same certified coordinate metric used in \cref{thm:certified-release-descent}; if a Gram or Fisher whitening is applied before diagonalizing the release matrix, $B(v)$ is the corresponding physical generator after that whitening map. A less conservative implementation may use the selected-direction constant $L_3(\widehat v_1)=8\norm{B(\widehat v_1)}_{\op}^3\norm{C}_{\op}$ in the step-size test, while the displayed supremum gives a precomputed uniform certificate over the whole dictionary. For fixed-weight Pauli release generators with a dictionary normalization giving $\sup_{\norm{v}_2=1}\norm{B(v)}_{\op}=O(1)$, this bound is independent of the global Hilbert-space dimension.
\end{remark}

\begin{theorem}[Release subspace stability]
\label{thm:subspace-stability}
Let $P_+$ be the spectral projector of $A=K-S_M-2G_N$ onto its positive eigenspace of dimension $r$. Assume $1\le r<p$ and
\begin{equation}
\lambda_r(A)>0,\qquad \lambda_{r+1}(A)<0,
\qquad
\gamma_+=\min\{\lambda_r(A),-\lambda_{r+1}(A)\}>0.
\end{equation}
Let $\widehat P_+$ be the projector selected from $\widehat A$ by thresholding positive eigenvalues after subtracting a radius $\epsilon$ with $\norm{\widehat A-A}_{\op}\le\epsilon<\gamma_+/2$. Then the selected dimension is correct and
\begin{equation}
\norm{\widehat P_+-P_+}_{\op}\le \frac{2\epsilon}{\gamma_+}.
\label{eq:dk-release}
\end{equation}
\end{theorem}

\begin{proof}
The dimension statement follows from Weyl's inequality. The projector bound is the Davis--Kahan sin-theta theorem \cite{DavisKahan1970} applied to the spectral split between the positive and nonpositive parts of $A$.

\end{proof}

\subsection{Finite-shot executable estimation: shadow matrices and scalar probes}
\label{subsec:shot-estimation}

For each pair \((a,b)\), write
\begin{equation}
O_{ab}=\frac12\bigl([B_a,[B_b,C]]+[B_b,[B_a,C]]\bigr),
\qquad
K_{ab}=\Tr(O_{ab}\rho_0).
\end{equation}
Low-weight Pauli expansions permit grouped measurements. Classical shadows reuse each measurement record across all curvature entries \cite{HuangKuengPreskill2020}: a random Hermitian snapshot $\widehat\rho$ satisfies $\bbE_\rho\widehat\rho=\rho$, so $X_{ab}=\Tr(O_{ab}\widehat\rho)$ estimates $K_{ab}$. For the chosen ensemble, use the uncentered shadow norm
\begin{equation}
\|O\|_{\sh}^2=\sup_\rho\bbE_\rho\bigl[\Tr(O\widehat\rho)^2\bigr],
\end{equation}
which bounds the single-record second moment uniformly over input states.

\begin{theorem}[Simultaneous shadow estimation of the release curvature matrix]
\label{thm:shadow-curvature}
Assume \(\|O_{ab}\|_{\sh}\le \nu\) for all \(a,b\), and that a median-of-means shadow estimator satisfies the single-observable tail bound
\begin{equation}
\bbP\{|\widehat K_{ab}-K_{ab}|>\varepsilon\}
\le 2\exp\left(-cM\varepsilon^2/\nu^2\right).
\end{equation}
Then for a numerical constant \(C\),
\begin{equation}
M\ge C\nu^2\varepsilon^{-2}\log\frac{2p^2}{\delta}
\quad\Longrightarrow\quad
\max_{a,b}|\widehat K_{ab}-K_{ab}|\le\varepsilon
\end{equation}
with probability at least \(1-\delta\), and hence \(\|\widehat K-K\|_{\op}\le p\varepsilon\). More sharply, if \(\widehat K=M^{-1}\sum_{t=1}^MX_t\) with independent symmetric matrix samples satisfying \(\bbE X_t=K\), \(\|X_t-K\|_{\op}\le L\) and variance proxy \(\|\sum_t\bbE(X_t-K)^2\|_{\op}\le Mv\), then with probability at least \(1-\delta\),
\begin{equation}
\|\widehat K-K\|_{\op}
\le
C\left(\sqrt{\frac{v\log(2p/\delta)}{M}}+\frac{L\log(2p/\delta)}{M}\right).
\label{eq:shadow-entry-op}
\end{equation}
\end{theorem}

\begin{proof}
The entrywise statement follows by applying the single-observable shadow concentration bound to the \(p^2\) entries and taking a union bound. The crude operator bound follows from \(\|E\|_{\op}\le p\|E\|_{\max}\). The second part is the matrix Bernstein inequality applied to the independent centered matrix samples.

\end{proof}

\begin{proposition}[Computable simultaneous curvature certificate]
\label[proposition]{prop:empirical-curvature-certificate}
Let $X_1,\ldots,X_M$ be independent and identically distributed symmetric curvature snapshot matrices for a fixed probe family, with $\bbE X_t=K$ and known deterministic bounds $|(X_t)_{ab}|\le R_{ab}$. Write $q=p(p+1)/2$, $\widehat K=M^{-1}\sum_tX_t$, and $\widehat v_{ab}=(M-1)^{-1}\sum_t((X_t)_{ab}-\widehat K_{ab})^2$. For $M\ge2$, define the symmetric nonnegative envelope
\begin{equation}
E_{ab}=\sqrt{\frac{2\widehat v_{ab}\log(4q/\delta)}{M}}
+\frac{14R_{ab}\log(4q/\delta)}{3(M-1)}.
\label{eq:empirical-entry-envelope}
\end{equation}
Then $\|\widehat K-K\|_{\op}\le\|E\|_{\op}$ with probability at least $1-\delta$, including after selection of an empirical eigendirection. For local Pauli snapshots of $O_{ab}=\sum_Pc_{ab,P}P$, a valid state-independent bound is
\begin{equation}
R_{ab}=\max_{\mathbf b\in\{X,Y,Z\}^{n}}
\sum_{P\text{ compatible with }\mathbf b}3^{\operatorname{wt}(P)}|c_{ab,P}|.
\label{eq:deterministic-shadow-range}
\end{equation}
The simpler sum over all strings is also valid.
\end{proposition}

\begin{proof}
Apply the two-sided empirical Bernstein inequality of \cite[Thm.~4]{MaurerPontil2009} to each upper-triangular entry, whose interval has length $2R_{ab}$, with failure budget $\delta/q$. The two-sided form uses $\log(4q/\delta)$ and the unbiased sample variance. A union bound gives $|\widehat K_{ab}-K_{ab}|\le E_{ab}$ simultaneously. For any real unit vector $v$, $|v^\top(\widehat K-K)v|\le |v|^\top E|v|\le\|E\|_{\op}$, proving the operator bound. A Pauli snapshot contributes zero for an incompatible string and a signed value $3^{\operatorname{wt}(P)}$ otherwise, which proves \eqref{eq:deterministic-shadow-range}. No independence among entries is needed.
\end{proof}

The empirical variance controls the leading statistical term, while the deterministic range protects against outcomes absent from a finite sample. For a fixed empirical supervised risk, identically distributed snapshots can be obtained by sampling its examples uniformly with replacement. The resulting interval is conditional on those examples. A stratified design requires concentration for its actual sampling law.

For a population curvature matrix, shot and example-sampling errors are added. If the $n$ independent per-example entries satisfy $|(K_i)_{ab}|\le B_K$, Hoeffding concentration and a union bound give
\begin{equation}
r_{\mathrm{data}}=pB_K\sqrt{\frac{2\log(2p^2/\delta_{\mathrm{data}})}{n}}.
\end{equation}
This is an operator-norm sampling radius in the fixed probing coordinates. Loss derivatives in \eqref{eq:nonlinear-curvature} require their own simultaneous bounds when estimated.

\begin{proposition}[Pauli locality and grouped-measurement shot cost]
\label[proposition]{prop:pauli-resource}
Suppose
\begin{equation}
B_a=\sum_{r}b_{a,r}P_{a,r},
\qquad
C=\sum_s c_sQ_s,
\end{equation}
where $P_{a,r}$ and $Q_s$ are Pauli strings. Let $\|A\|_{P,1}$ denote the $\ell_1$ norm of the Pauli coefficients of $A$. Then
\begin{equation}
\|O_{ab}\|_{P,1}
\le 4\|B_a\|_{P,1}\|B_b\|_{P,1}\|C\|_{P,1}.
\label{eq:pauli-l1}
\end{equation}
If the Pauli strings in $B_a,B_b,C$ have weights at most $w_B,w_B,w_C$, respectively, every Pauli string appearing in $O_{ab}$ has weight at most $2w_B+w_C$ before cancellations. If $O_{ab}=\sum_{g=1}^{G_{ab}}\sum_{t\in g}c_tP_t$ is partitioned into jointly measurable commuting groups and group $g$ receives $m_g$ shots in independent batches, then the unbiased grouped estimator has variance at most
\begin{equation}
\Var(\widehat K_{ab})
\le \sum_{g=1}^{G_{ab}}\frac{(\sum_{t\in g}|c_t|)^2}{m_g}.
\label{eq:group-var}
\end{equation}
The continuous allocation optimum over $M_{ab}$ shots is
\begin{equation}
\Var(\widehat K_{ab})
\le \frac{\left(\sum_g\sum_{t\in g}|c_t|\right)^2}{M_{ab}}
=\frac{\|O_{ab}\|_{P,1}^2}{M_{ab}}.
\label{eq:group-var-opt}
\end{equation}
For integer counts and $M_{ab}>G_{ab}$, allocate $m_g=\lceil(M_{ab}-G_{ab})L_g/\sum_h L_h\rceil$ to each nonzero group, where $L_g=\sum_{t\in g}|c_t|$. This uses at most $M_{ab}$ shots and bounds the variance by $\|O_{ab}\|_{P,1}^2/(M_{ab}-G_{ab})$, at most twice the continuous bound when $M_{ab}\ge2G_{ab}$. Consequently, low-weight local release dictionaries have polynomial measurement overhead whenever the grouped Pauli $\ell_1$ norms and the number of candidate pairs are polynomial in system size.
\end{proposition}

\begin{proof}
For Pauli strings $P,Q$, the commutator is either zero or has Pauli-coefficient magnitude $2$. Applying this twice gives
\begin{equation}
\|[B_a,[B_b,C]]\|_{P,1}
\le4\|B_a\|_{P,1}\|B_b\|_{P,1}\|C\|_{P,1},
\end{equation}
and symmetrization cannot increase the bound. Pauli multiplication can enlarge support only by the union of the supports of the multiplied strings, giving the weight bound. For a commuting group, a single shot returns simultaneous $\pm1$ outcomes for every $P_t$ in the group, so the absolute value of the grouped random variable is at most $\sum_{t\in g}|c_t|$. This proves \eqref{eq:group-var}. Minimizing $\sum_g L_g^2/m_g$ subject to $\sum_gm_g=M_{ab}$ gives $m_g\propto L_g$ and \eqref{eq:group-var-opt}.
\end{proof}

\begin{corollary}[Shadow norms of local release commutators]
\label[corollary]{cor:shadow-commutator}
Use the single-qubit random Pauli shadow ensemble and write
\begin{equation}
O_{ab}=\sum_P c_{ab,P}P.
\end{equation}
If every nonzero Pauli string in $O_{ab}$ has weight at most $k_{ab}$, then its Pauli-shadow norm obeys
\begin{equation}
\norm{O_{ab}}_{\sh}^2
\le\left(\sum_P3^{\operatorname{wt}(P)/2}|c_{ab,P}|\right)^2
\le3^{k_{ab}}\norm{O_{ab}}_{P,1}^2.
\label{eq:shadow-local-weight}
\end{equation}
In the setting of \cref{prop:pauli-resource}, $k_{ab}\le2w_B+w_C$. Hence the simultaneous entrywise shadow estimate in \cref{thm:shadow-curvature} is guaranteed whenever
\begin{equation}
M\ge
C\,3^{2w_B+w_C}\epsilon^{-2}
\max_{a,b}\norm{O_{ab}}_{P,1}^2
\log\frac{2p^2}{\delta}.
\label{eq:shadow-local-shot}
\end{equation} It makes explicit that high-weight global release dictionaries can be measurement-expensive before optimization begins, while local or hardware-native dictionaries keep the shadow factor polynomial when the Pauli coefficient norms are controlled.
\end{corollary}

\begin{proof}
The single-Pauli snapshot has second moment $3^{\operatorname{wt}(P)}$. Minkowski's inequality in $L^2$, followed by the supremum over states, gives the weighted coefficient-sum bound. Correlations between compatible strings prevent replacing it by the sum of squared coefficients: $O=Z_1+Z_2$ on $|00\rangle$ has snapshot second moment $8$, exceeding $3+3=6$.
\end{proof}

When \([B_a,[B_b,C]]\) has high Pauli weight, the shadow norm \(\nu\) may be exponentially large. In that regime one should use scalar central differences: measure \(\cL(0)\) and \(\cL(\pm\delta_\beta)\) for a chosen direction \(B(v)\), and estimate \(-\cL''(0)\). \Cref{app:fd} gives the bias--variance decomposition, the resulting optimal step schedule \(\delta_\beta\asymp M^{-1/8}\), and the precise sense in which the sampling cost then decouples from the Pauli weight of the nested commutator: the measured observable remains the original loss observable, so the statistical scale is set by \(\sigma_C^2/(\delta_\beta^4M)\) rather than by the shadow norm of \([B,[B,C]]\). \Cref{subsec:executable-certification} compares the two modes quantitatively on one model.

\subsection{How far to release}

In the positive spectral coordinates $Aq_j=\lambda_jq_j$, $\lambda_j>0$, suppose the quartic term is separable in those coordinates, $\mathcal Q(\beta)=\sum_j q_j^{(4)}z_j^4$ with $z_j=q_j^\top\beta$. With $q_j^{(4)}>0$, the quartic surrogate has minimizers
\begin{equation}
 z_{j,\star}^2=\frac{6\lambda_j}{q_j^{(4)}}.
\label{eq:release-amplitude}
\end{equation}
For empirical eigenvalues $\widehat\lambda_j$ of $\widehat A$ in the same coordinate metric, a candidate shrinkage rule is
\begin{equation}
 s_j=\min\left\{1,\;c_s\sqrt{\frac{(\widehat\lambda_j-r_\delta)_+}{(\widehat q_j^{(4)})_++\varepsilon_q}}\right\},
\qquad
\alpha_j=s_j^{-1}-1.
\label{eq:shrinkage-release}
\end{equation}
Here $c_s>0$ fixes the amplitude convention, $\varepsilon_q>0$ regularizes the estimated quartic coefficient, and $s_j=0$ means $\alpha_j=\infty$. This surrogate suggests a continuous onset of amplitude. Its finite step still requires the smoothness certificate or independent validation; an approximate quartic coefficient alone does not certify it.

\begin{remark}[Release as a continuous transition in hypothesis space]
\label[remark]{rem:landau}
The positive spectral coordinates give \eqref{eq:local-risk-release} the form of a mean-field free energy. Writing $m_{\eff}^2$ for the eigenvalues of $S_M+2G_N-K$, the quadratic coefficient of $z_j$ is $\tfrac12 m_{\eff,j}^2$ and the quartic coefficient is $q_j^{(4)}/24$, so the minimizer is $z_{j,\star}=0$ while $m_{\eff,j}^2>0$ and turns on continuously as $\sqrt{-m_{\eff,j}^2}$ once the sign changes, exactly as in \eqref{eq:release-amplitude}. In the local model, release is therefore a continuous transition in hypothesis space whose control parameter is the shot budget: $S_M\propto1/M$ shifts the effective mass, so increasing $M$ at fixed curvature drives $m_{\eff,j}^2$ through zero and opens the direction. The analogy is a statement about the local risk surface, not about the physical state: the transition happens in the space of models the data support, and the ordered phase is a released hypothesis class rather than a symmetry-broken state.
\end{remark}

\section{Local Statistical Decision Boundary for Releasing Symmetry}
\label{sec:decision-boundary}

The available measurement protocol determines the cost of finding a release direction. A Gaussian sequence model isolates how this cost depends on branch dimension and dictionary size. Its quantum interpretation requires calibration of the observation noise and total copy budget; Gram orthogonalization alone specifies the coordinate metric.

\begin{assumption}[Canonical local release experiment]
\label[assumption]{ass:local-gaussian}
After Gram orthogonalization of a $k$-dimensional candidate breaking branch, the statistic used for release has the form
\begin{equation}
Z=\mu+\tau_{N,M}\xi,
\qquad
\xi\sim\mathcal N(0,I_k),
\qquad
\tau_{N,M}^2=\frac{\sigma_{\mathrm{data}}^2}{N}+\frac{\sigma_{\mathrm{shot}}^2}{M}.
\label{eq:local-gaussian}
\end{equation}
The hard-symmetry null is $H_0:\mu=0$ and the released alternative is $H_1(r):\norm{\mu}_2\geq r$.
\end{assumption}

The mean can represent a calibrated gradient, curvature statistic, or validation improvement when the stated Gaussian model is justified. Here $M$ indexes the variance of each observed coordinate; the total state-copy cost depends on whether those coordinates share measurements.

\begin{theorem}[Finite-sample release detection boundary]
\label{thm:detection-boundary}
Let $Z$ follow \cref{ass:local-gaussian}, with $k\ge1$, $\tau_{N,M}>0$, $0<\delta<1$, and $t=\log(1/\delta)$.
\begin{enumerate}
\item The test
\begin{equation}
\varphi(Z)=\mathbf 1\left\{\frac{\norm{Z}_2^2}{\tau_{N,M}^2}>k+2\sqrt{kt}+2t\right\}
\label{eq:norm-test}
\end{equation}
has type-I error at most $\delta$. Moreover, with $C_1=16$, if
\begin{equation}
r^2\geq C_1\tau_{N,M}^2\bigl(\sqrt{kt}+t\bigr),
\label{eq:upper-detection}
\end{equation}
then $\inf_{\norm\mu\ge r}\bbP_\mu\{\varphi(Z)=1\}\ge1-\delta$.
\item Conversely, for any test $\psi$ and any $r>0$,
\begin{equation}
\bbP_0\{\psi=1\}+\sup_{\norm\mu=r}\bbP_\mu\{\psi=0\}
\geq
1-\frac12\sqrt{\exp\left(\frac{r^4}{2k\tau_{N,M}^4}\right)-1}.
\label{eq:lower-detection}
\end{equation}
Consequently, if $r^2\le c\tau_{N,M}^2\sqrt{k}$ for a sufficiently small numerical constant $c$, no procedure can have both small false-release probability and high release power uniformly over unknown breaking directions.
\end{enumerate}
\end{theorem}

\begin{proof}
Under $H_0$, $\norm{Z}^2/\tau_{N,M}^2$ is central $\chi_k^2$, and Laurent--Massart concentration \cite{LaurentMassart2000} gives \eqref{eq:norm-test}. Under the alternative, write $W=\norm{Z}^2/\tau_{N,M}^2$ and $\lambda=\norm{\mu}^2/\tau_{N,M}^2$. For $u\ge0$, its moment-generating function gives $\log\bbE e^{-u(W-k-\lambda)}\le(k+2\lambda)u^2$. Chernoff optimization therefore yields
\begin{equation}
\bbP\{W\le k+\lambda-2\sqrt{(k+2\lambda)t}\}\le e^{-t}.
\end{equation}
Since $2\sqrt{(k+2\lambda)t}\le2\sqrt{kt}+\lambda/2+4t$, the condition $\lambda\ge16(\sqrt{kt}+t)$ places this lower-tail threshold above the rejection threshold and proves the claimed power.

For the lower bound, following the standard Gaussian testing method \cite{IngsterSuslina2003}, put a prior on alternatives by taking $\mu=a\varepsilon$ with $\varepsilon\in\{\pm1\}^k$ uniform and $ka^2=r^2$. If $P_\pi$ is the resulting mixture distribution and $P_0$ is the null, then
\begin{equation}
\chi^2(P_\pi\|P_0)+1
=\bbE_{\varepsilon,\varepsilon'}\exp\left(\frac{a^2}{\tau_{N,M}^2}\varepsilon^\top\varepsilon'\right)
=\cosh\left(\frac{a^2}{\tau_{N,M}^2}\right)^k
\leq \exp\left(\frac{r^4}{2k\tau_{N,M}^4}\right).
\end{equation}
Total variation is bounded by $\TV(P_\pi,P_0)\leq\frac12\sqrt{\chi^2(P_\pi\|P_0)}$. Since the worst-case type-II error is at least the Bayes average type-II error under $\pi$, the claimed inequality follows.
\end{proof}

The squared mean-signal threshold in this experiment is of order $\tau_{N,M}^2\sqrt{k}$ at fixed confidence. A curvature-search rule inherits this boundary only through an observation model that identifies its statistic and noise scale with $\mu$ and $\tau_{N,M}$. The next calculation treats the distinct cost of choosing one signal-bearing coordinate from a dictionary.

\begin{theorem}[Sparse dictionary search boundary]
\label{thm:sparse-dictionary-boundary}
Let a release search over $p$ candidate directions produce
\begin{equation}
Y_j=\mu_j+\tau\xi_j,
\qquad \xi_j\sim\mathcal N(0,1)\ \text{independently},
\qquad \tau=\sigma/\sqrt M.
\end{equation}
The null is $H_0:\mu=0$. For $p\ge2$, $\tau>0$, and $\mu_0>0$, the one-sparse alternative is
\begin{equation}
H_1(\mu_0):\quad \mu=\mu_0 e_j\ \text{for an unknown }j\in\{1,\ldots,p\}.
\end{equation}
The threshold test
\begin{equation}
\varphi_{\max}(Y)=\mathbf 1\left\{\max_jY_j>\tau\sqrt{2\log(p/\delta)}\right\}
\label{eq:max-test}
\end{equation}
has type-I error at most $\delta$. If
\begin{equation}
\mu_0\ge \tau\left(\sqrt{2\log(p/\delta)}+\sqrt{2\log(1/\delta)}\right),
\label{eq:sparse-upper}
\end{equation}
then its power is at least $1-\delta$. Conversely, for any release test $\psi$,
\begin{equation}
\bbP_0\{\psi=1\}+\sup_{1\le j\le p}\bbP_{\mu_0e_j}\{\psi=0\}
\ge
1-\frac12\sqrt{\frac{\exp(\mu_0^2/\tau^2)-1}{p}}.
\label{eq:sparse-lower}
\end{equation}
Thus if $\mu_0^2\le (1-\varepsilon)\tau^2\log p$ along a sequence with $p\to\infty$, the sum of false-release and worst-case missed-release probabilities tends to one. Searching over an unknown breaking dictionary therefore has the unavoidable scale
\begin{equation}
\mu_0\asymp \frac{\sigma}{\sqrt M}\sqrt{\log p}.
\label{eq:sparse-scale}
\end{equation}
\end{theorem}

\begin{proof}
The type-I bound follows from
\begin{equation}
\bbP_0\{\max_jY_j>t\}\le p\exp(-t^2/(2\tau^2))
\end{equation}
with $t=\tau\sqrt{2\log(p/\delta)}$. Under $H_1$ at coordinate $j$, the event $Y_j>t$ has probability at least $1-\delta$ whenever \eqref{eq:sparse-upper} holds. For the lower bound, mix uniformly over the $p$ alternatives. If $P_\pi$ denotes the mixture and $P_0$ the null, then
\begin{equation}
\chi^2(P_\pi\|P_0)
=\frac1p\left(\exp(\mu_0^2/\tau^2)-1\right).
\end{equation}
The standard chi-square total-variation bound $\TV(P_\pi,P_0)\le\frac12\sqrt{\chi^2(P_\pi\|P_0)}$ and the testing inequality give \eqref{eq:sparse-lower}.

\end{proof}

\begin{corollary}[Correlated dictionaries and orthogonal packings]
\label[corollary]{cor:correlated-dictionary}
Let $X\sim\mathcal N(\theta,\tau^2I_d)$ and let a dictionary produce standardized statistics $Y_j=u_j^\top X$ with $\norm{u_j}_2=1$, $j=1,\ldots,p$. The max test
\begin{equation}
\mathbf 1\left\{\max_jY_j>\tau\sqrt{2\log(p/\delta)}\right\}
\end{equation}
has type-I error at most $\delta$ without any independence assumption among the $Y_j$. Conversely, if the dictionary contains an orthonormal subdictionary of size $q$, then for the alternatives $\theta=\mu_0u_j$ on that subdictionary, the lower bound in \cref{thm:sparse-dictionary-boundary} holds with $p$ replaced by $q$. This brackets the search cost between an orthogonal-subdictionary lower bound and a full-dictionary upper bound. A matching metric-entropy characterization for general correlated dictionaries requires additional analysis.
\end{corollary}

\begin{proof}
The upper bound is the union bound applied to the marginal Gaussian tails. The lower bound restricts the alternatives to the orthonormal subdictionary and repeats the mixture calculation in \cref{thm:sparse-dictionary-boundary}. Restricting the alternative set can only make the testing problem easier, so impossibility on that subset is a valid lower bound for the full dictionary.
\end{proof}

The dense and sparse experiments quantify different search costs. Their confidence margins distinguish a detected release signal from an unresolved one; an unresolved branch can remain closed at the chosen scale or enter a separately validated soft-release candidate.

\section{Structural Risk Theory for Release Paths}
\label{sec:srm}

The spectral criterion is local, and a local gain can still be paid for by a global complexity cost. The comparison that settles this is a finite-sample risk decomposition. When separate approximation, estimation, shot, optimization, and discovery bounds are available, their organizing decomposition is
\begin{equation}
\cE_{\gen}(H,S,\alpha)
\lesssim
\cE^{\sym}_{\bias}(H,S,\alpha)
+\cE^{\comm}_{\var}(H,S,\alpha)
+\cE_{\shot}(S,\alpha)
+\cE_{\opt}
+\cE_{\disc}.
\label{eq:risk-decomp}
\end{equation}

\subsection{Projection bias from excessive symmetry}

For a released class, write $\cE^{\sym}_{\bias}(a)=\inf_{f\in\cF_a}\|f-f_\star\|_{L^2(P_X)}^2$ for its approximation error; this reduces to \eqref{eq:projection-bias} for the closed equivariant subspace. Along a genuinely nested release path, enlarging $S$ or decreasing $\alpha$ cannot increase this infimum. The gated linear classes in \cref{thm:soft-effective-dimension} have this nesting under a common coefficient-radius constraint. Their estimation costs then determine which reduction in bias is worth realizing.

\subsection{Effective commutant complexity}

For a structure $a=(H,S,\alpha)$, let $\cA_a$ be the allowed Hermitian readout or local tangent-readout space, with orthonormal basis and feature vector $\Phi_a(\rho)$. Let $\widehat\Sigma_a=N^{-1}\sum_i\Phi_a(\rho_i)\Phi_a(\rho_i)^\top$. Write $\Rad_N(\cF)=N^{-1}\bbE_\epsilon\sup_{f\in\cF}\sum_i\epsilon_i f(\rho_i)$ for the empirical Rademacher complexity, where the independent signs $\epsilon_i$ are uniform on $\{-1,1\}$.

\begin{theorem}[Data-dependent effective-commutant bound]
\label{thm:effective-commutant-bound}
Let the class $\cF_a=\{\rho\mapsto w^\top\Phi_a(\rho):\|w\|_2\le R_a\}$ be fixed independently of the $N$ examples. For a loss in $[0,B_\ell]$ that is $L_\ell$-Lipschitz in the prediction, with probability at least $1-\delta$,
\begin{equation}
\sup_{f\in\cF_a}|R(f)-\widehat R_N(f)|
\le4L_\ell R_a\sqrt{\frac{\Tr\widehat\Sigma_a}{N}}
+3B_\ell\sqrt{\frac{\log(4/\delta)}{2N}}.
\label{eq:effective-bound}
\end{equation}
For a Hilbert--Schmidt orthonormal operator basis and density-matrix inputs,
\[
\Tr\widehat\Sigma_a\le\frac1N\sum_i\Tr(\rho_i^2)\le1.
\]
\end{theorem}

\begin{proof}
Conditionally on the data, Jensen's inequality and independence of Rademacher signs give
\begin{equation}
\Rad_N(\cF_a)\le\frac{R_a}{N}\bbE_\epsilon\left\|\sum_i\epsilon_i\Phi_a(\rho_i)\right\|_2
\le R_a\sqrt{\frac{\Tr\widehat\Sigma_a}{N}}.
\end{equation}
Apply the empirical Rademacher generalization inequality and Lipschitz contraction to the two one-sided loss deviations \cite{BartlettMendelson2002}; the displayed constants allow the two-sided union bound. For an orthonormal operator family, Bessel's inequality gives $\sum_j|\Tr(A_j\rho)|^2\le\Tr(\rho^2)$ pointwise.
\end{proof}

The radius and normalization are essential to a capacity comparison. At fixed Hilbert--Schmidt coefficient radius, enlarging a commutant cannot by itself turn this bound exponential: the feature trace is bounded by the purity. An operator-norm-bounded readout can instead require Hilbert--Schmidt radius of order $\sqrt d$, and a basis with only coordinatewise bounds $|\Phi_j|\le1$ permits the looser trace bound $m_a$. The Schur--Weyl dimensions $4^n$ and $\binom{n+3}{3}$ for unrestricted and permutation-invariant readouts \cite{CompanionDiscovery} therefore enter only after the same readout normalization is fixed on both sides. The complementary sparse-coefficient bound in \cite{CompanionDiscovery} controls an $\ell_1$ ball and pays $\sqrt{\log m_H}$; it prices a different class.

\subsection{Soft release effective dimension}

The previous bound treats a released space as either open or closed. Soft release requires a continuous dimension formula.

\begin{theorem}[Soft-release effective commutant dimension]
\label{thm:soft-effective-dimension}
Let the feature vector decompose as $\Phi=(\Phi_H,\Phi_S)$, where $\Phi_H$ is retained and $\Phi_S\in\R^q$ contains Gram-orthonormal breaking features. Let $D_\alpha=\diag(s_1(\alpha_1),\ldots,s_q(\alpha_q))$ and define
\begin{equation}
\Phi_\alpha(\rho)=(\Phi_H(\rho),D_\alpha\Phi_S(\rho)).
\end{equation}
For the linear class $\{w^\top\Phi_\alpha:\norm{w}_2\le R\}$,
\begin{equation}
\Rad_N(\cF_{H,S,\alpha})
\le R\sqrt{\frac{\Tr\widehat\Sigma_H+\Tr(D_\alpha\widehat\Sigma_S D_\alpha)}{N}}.
\label{eq:soft-rad}
\end{equation}
If the breaking coordinates diagonalize their empirical covariance, this becomes
\begin{equation}
\Rad_N(\cF_{H,S,\alpha})
\le R\sqrt{\frac{\Tr\widehat\Sigma_H+\sum_{j=1}^q s_j(\alpha_j)^2\widehat\nu_j}{N}},
\label{eq:soft-trace}
\end{equation}
where $\widehat\nu_j$ are the breaking covariance eigenvalues. For a different parameter norm, first transform both the feature map and the feasible coefficient set into its whitened coordinates. The displayed Rademacher bound is conditional on the realized feature map. A population-risk bound for a data-chosen rotation or gate additionally requires independent validation or uniform control over the candidate maps. Under the isotropic noise--curvature model of \cref{prop:ridge-scale}, the corresponding shot dimension is
\begin{equation}
D_{\br}(S,\alpha)=\sum_{j=1}^q s_j(\alpha_j)^2.
\label{eq:soft-shot-dim}
\end{equation}
\end{theorem}

\begin{proof}
Apply the Rademacher argument in \cref{thm:effective-commutant-bound} to the gated feature map $\Phi_\alpha$. Since $\norm{(\Phi_H,D_\alpha\Phi_S)}_2^2=\norm{\Phi_H}_2^2+\norm{D_\alpha\Phi_S}_2^2$ in whitened coordinates, the trace of the empirical covariance is $\Tr\widehat\Sigma_H+\Tr(D_\alpha\widehat\Sigma_SD_\alpha)$. Diagonalizing the breaking covariance gives \eqref{eq:soft-trace}. The shot dimension follows because a gate of amplitude $s_j$ scales the variance contribution of that released coordinate by $s_j^2$ in the local quadratic model.
\end{proof}

\begin{remark}
Equation \eqref{eq:soft-trace} quantifies the statistical cost of partially opening a direction: its feature variance and, under the local ridge model, its shot cost are multiplied by $s_j^2$. The gate remains a candidate choice until its loss benefit is certified or independently validated.
\end{remark}

\subsection{Oracle inequality for candidate release paths}

For a finite candidate family, the selection score combines empirical risk with simultaneous uncertainty and nonnegative structure costs. Independence of final validation from candidate generation makes this comparison valid for trained predictors.

\begin{assumption}[Selection and validation discipline]
\label[assumption]{ass:sample-splitting}
The candidate family $\mathscr A$ is fixed before final validation, or is treated conditionally as in \cref{thm:post-selection-oracle}. The final risk estimate is evaluated on data independent of the random choices that generated $\mathscr A$, or by a cross-fitting procedure with separately justified foldwise confidence events.
\end{assumption}

Define
\begin{equation}
\widehat{\srmscore}(a)=\widehat R_N(\widehat f_a)
+\Lambda_a(N,\delta)+\widehat\Xi_a(M,\delta)+\widehat\Omega_a(N,\delta),
\label{eq:srm-score}
\end{equation}
where $\widehat f_a$ is the trainer output in structure $a$, $\Lambda_a$ is the commutant complexity, $\widehat\Xi_a$ estimates the shot/release uncertainty, and $\widehat\Omega_a$ estimates the discovery or task-validation uncertainty. These penalty estimates are taken to be nonnegative, as they are used only as upper-confidence costs. Let
\begin{equation}
\widehat a\in\argmin_{a\in\mathscr A}\widehat{\srmscore}(a).
\end{equation}

\begin{theorem}[Finite-family structural-risk oracle inequality]
\label{thm:oracle}
Under \cref{ass:sample-splitting}, assume that for every $a\in\mathscr A$, with probability at least $1-\delta/|\mathscr A|$,
\begin{equation}
\abs{R(\widehat f_a)-\widehat R_N(\widehat f_a)}\le\Lambda_a,
\qquad
\abs{\widehat\Xi_a-\Xi_a}\le\Xi_a,
\qquad
\abs{\widehat\Omega_a-\Omega_a}\le\Omega_a.
\end{equation}
Then with probability at least $1-\delta$,
\begin{equation}
R(\widehat f_{\widehat a})
\le
\inf_{a\in\mathscr A}\left\{R(\widehat f_a)+2\Lambda_a+2\Xi_a+2\Omega_a\right\}.
\label{eq:oracle}
\end{equation}
If the candidate trainers additionally admit the component-wise risk bounds on the right-hand side, this comparison yields
\begin{equation}
\cE_{\gen}(\widehat a)
\lesssim
\inf_{a=(H,S,\alpha)}\left[
\cE^{\sym}_{\bias}(a)+\cE^{\comm}_{\var}(a)+\cE_{\shot}(a)+\cE_{\opt}(a)+\cE_{\disc}(a)
\right].
\label{eq:oracle-decomp}
\end{equation}
\end{theorem}

\begin{proof}
On the union-bound event, $R(\widehat f_a)\le\widehat R_N(\widehat f_a)+\Lambda_a$ and $\widehat R_N(\widehat f_a)\le R(\widehat f_a)+\Lambda_a$ for all $a$. The minimality of $\widehat a$ in \eqref{eq:srm-score} gives, for arbitrary $a$,
\begin{align}
R(\widehat f_{\widehat a})
&\le \widehat R_N(\widehat f_{\widehat a})+\Lambda_{\widehat a} \\
&\le \widehat R_N(\widehat f_a)+\Lambda_a+\widehat\Xi_a+\widehat\Omega_a-\widehat\Xi_{\widehat a}-\widehat\Omega_{\widehat a} \\
&\le R(\widehat f_a)+2\Lambda_a+2\Xi_a+2\Omega_a,
\end{align}
where the estimates are replaced by their confidence bounds and negative terms are discarded. Taking the infimum over $a$ proves \eqref{eq:oracle}. The final decomposition is conditional on bounds for those trainer-specific terms; the selection inequality itself compares the risks of the realized trained predictors.

\end{proof}

\subsection{Post-selection validity: independent validation oracle}
\label{subsec:postselection-oracle}

Conditioning on all fitting and release information freezes an adaptively generated family before it encounters validation data. The validation comparison then depends on the number and risks of these realized candidates, regardless of the preceding search.

\begin{theorem}[Post-selection oracle inequality for adaptive candidate families]
\label{thm:post-selection-oracle}
Let \(\widehat{\mathscr A}\) be any random finite candidate family measurable with respect to \(\mathcal D_{\fit}\cup\mathcal D_{\rel}\), with \(|\widehat{\mathscr A}|\le A_{\max}\). For every \(a\in\widehat{\mathscr A}\), train \(\widehat f_a\) without using \(\mathcal D_{\valsplit}\). Suppose that, conditionally on \(\mathcal D_{\fit}\cup\mathcal D_{\rel}\), the validation estimate satisfies
\begin{equation}
\bbP\left\{
|R(\widehat f_a)-\widehat R_{\valsplit}(\widehat f_a)|
\le c_a(\delta/A_{\max})
\ \text{for all }a\in\widehat{\mathscr A}
\right\}
\ge 1-\delta.
\end{equation}
Define
\begin{equation}
\widehat a\in\argmin_{a\in\widehat{\mathscr A}}
\left\{\widehat R_{\valsplit}(\widehat f_a)+c_a(\delta/A_{\max})+\Xi_a+\Omega_a\right\},
\end{equation}
where \(\Xi_a,\Omega_a\ge0\) are costs fixed conditionally on the fitting and release streams. Then, with probability at least \(1-\delta\),
\begin{equation}
R(\widehat f_{\widehat a})
\le
\inf_{a\in\widehat{\mathscr A}}
\left\{R(\widehat f_a)+2c_a(\delta/A_{\max})+\Xi_a+\Omega_a\right\}.
\label{eq:post-selection-oracle}
\end{equation}
If the retained-subgroup family is produced by the dictionary-recovery procedure of \cite{CompanionDiscovery} with failure probability \(\delta_{\disc}\), the risk comparison remains valid for the realized family with or without successful discovery. Intersecting with a valid subgroup-recovery event of failure probability \(\delta_{\disc}\) also certifies its discovery interpretation, with probability at least \(1-\delta-\delta_{\disc}\). One may set \(\Omega_a=0\) for exact recovered candidates; a tolerance or encoding bias remains an approximation cost. Recovery of dictionary elements alone does not ensure that a target subgroup or predictor was included.
\end{theorem}

\begin{proof}
Condition on the fitting and release streams. The candidate family is now fixed with respect to the validation stream. A union bound over at most \(A_{\max}\) candidates gives the simultaneous validation event. On that event, the same two-sided comparison used in \cref{thm:oracle} proves \eqref{eq:post-selection-oracle}. No assumption is needed on how aggressively \(\widehat{\mathscr A}\) was searched before validation, because all such adaptivity has been conditioned away. For the last claim, intersect this event with the dictionary-recovery event of \cite{CompanionDiscovery} and apply a union bound.
\end{proof}

\begin{corollary}[Fast validation along nested release paths]
\label[corollary]{cor:nested-adaptivity}
Let $\widehat f_1,\ldots,\widehat f_K$ be trained without the $n_{\val}$ validation examples, and assume $|Y|,|\widehat f_k(X)|\le b$. For squared loss, write $f_\star(X)=\bbE[Y\mid X]$ and $e_k=R(\widehat f_k)-R(f_\star)$. Suppose the measured validation risks are within a simultaneous additive shot error $\epsilon_{\shot}$ of their exact empirical values. Let $\widehat k$ minimize measured validation risk plus a nonnegative prevalidation cost $q_k$. On an event of validation probability at least $1-\delta$, intersected with the shot event,
\begin{equation}
e_{\widehat k}\le\min_{k\le K}\{3e_k+2q_k\}
+\frac{80b^2\log(2K/\delta)}{n_{\val}}+4\epsilon_{\shot}.
\label{eq:fast-validation}
\end{equation}
If a nested path satisfies $e_k\le Ck^{-2s}+C'k/N_{\eff}$, $q_k\le C''k/N_{\eff}$, $s>0$, $N_{\eff}^{-1}=N^{-1}+M^{-1}$, and contains $k\asymp N_{\eff}^{1/(2s+1)}$, then
\begin{equation}
e_{\widehat k}\lesssim N_{\eff}^{-2s/(2s+1)}
+\frac{b^2\log(2K/\delta)}{n_{\val}}+\epsilon_{\shot}.
\end{equation}
The path rate is preserved when the final two terms are of that order or smaller. If the simultaneous shot event fails with probability at most $\delta_{\shot}$, the joint guarantee has probability at least $1-\delta-\delta_{\shot}$.
\end{corollary}

\begin{proof}
Condition on the trained candidates. For $U_k=(\widehat f_k(X)-Y)^2-(f_\star(X)-Y)^2$, conditional expectation gives $\bbE U_k=e_k=\bbE(\widehat f_k-f_\star)^2$. Also $|U_k|\le4b^2$ and $\Var(U_k)\le16b^2e_k$. Bernstein's inequality and a union bound yield, with $t=\log(2K/\delta)$,
\begin{equation}
|(P_{n_{\val}}-P)U_k|
\le\sqrt{32b^2e_kt/n_{\val}}+8b^2t/(3n_{\val})
\le e_k/2+20b^2t/n_{\val}
\end{equation}
for all $k$. The common empirical loss of $f_\star$ cancels from candidate comparisons. Minimality of $\widehat k$, the two shot errors, and $q_{\widehat k}\ge0$ give $e_{\widehat k}/2\le3e_k/2+q_k+40b^2t/n_{\val}+2\epsilon_{\shot}$. This is \eqref{eq:fast-validation}. Balancing $k^{-2s}$ against $k/N_{\eff}$ proves the last assertion. The fast rate uses the squared-loss variance relation; bounded loss alone gives the slower validation term in \cref{thm:post-selection-oracle}.
\end{proof}

\subsection{Dynamical Lie algebra growth and barren-plateau margins}

The same confidence margin filters release directions whose curvature is too small to resolve. Unitary-design and dynamical Lie-algebra analyses identify concentration mechanisms for expressive quantum circuits \cite{McClean2018,Cerezo2021Cost,Ragone2024LieBP}. Under an operator-norm concentration bound for the probed curvature block, the probability of certifying such a branch follows directly from the estimator radius.

\begin{corollary}[Finite-shot rejection of exponentially flat release branches]
\label[corollary]{cor:dla-bp-margin}
Let a candidate release branch enlarge the backbone dynamical Lie algebra from $\mathfrak g_H$ to
\begin{equation}
\mathfrak g_{H,S}=\operatorname{Lie}(\mathfrak g_H,iS),
\qquad D_S=\dim\mathfrak g_{H,S}.
\end{equation}
Let $\theta$ be drawn from a specified probing distribution, and assume the true release-curvature block satisfies the concentration bound
\begin{equation}
\bbE_{\theta}\norm{K_S(\theta)}_{\op}^2\le \frac{C_{\mathrm{BP}}}{D_S}.
\label{eq:dla-bp-concentration}
\end{equation}
Suppose the empirical curvature estimator satisfies
\begin{equation}
\bbP\{\norm{\widehat K_S-K_S}_{\op}>r_\delta/2\}\le\delta_{\mathrm{est}},
\qquad
r_\delta=c\sigma\sqrt{\frac{\log(p/\delta)}{M}}.
\end{equation}
Then, for any positive semidefinite penalties $S_M$ and $G_N$,
\begin{equation}
\bbP\{\lambda_{\max}(\widehat K_S-S_M-2G_N)>r_\delta\}
\le
\delta_{\mathrm{est}}
+\frac{4C_{\mathrm{BP}}}{D_Sr_\delta^2}
=
\delta_{\mathrm{est}}
+O\!\left(\frac{M}{D_S\sigma^2\log(p/\delta)}\right).
\label{eq:dla-bp-filter}
\end{equation}
Consequently, if $M$ and $C_{\mathrm{BP}}/\sigma^2$ grow at most polynomially while $D_S$ grows exponentially, $\delta_{\mathrm{est}}$ is small, and the concentration model \eqref{eq:dla-bp-concentration} holds, this branch is not certified with high probability.
\end{corollary}

\begin{proof}
Since $S_M+2G_N\succeq0$,
\begin{equation}
\lambda_{\max}(\widehat K_S-S_M-2G_N)\le\lambda_{\max}(\widehat K_S)\le\norm{\widehat K_S}_{\op}.
\end{equation}
If $\norm{\widehat K_S}_{\op}>r_\delta$, then either $\norm{\widehat K_S-K_S}_{\op}>r_\delta/2$ or $\norm{K_S}_{\op}>r_\delta/2$. The first event has probability at most $\delta_{\mathrm{est}}$. The second is bounded by Markov's inequality and \eqref{eq:dla-bp-concentration} as $4C_{\mathrm{BP}}/(D_Sr_\delta^2)$. Substituting the displayed radius gives \eqref{eq:dla-bp-filter}.
\end{proof}

\section{Quotient QNG after Symmetry Release}
\label{sec:qng}

Release changes model capacity and parameter redundancy at the same time. A newly opened branch may share gauge coordinates with the retained backbone, or may be locally Fisher-indistinguishable at the current state. A full-space damped inverse then turns vertical finite-shot noise into representative drift. The quotient viewpoint supplies the correct update, in line with standard optimization on quotient and embedded manifolds \cite{AbsilMahonySepulchre2008,Boumal2023}.

\subsection{Release-dependent vertical and horizontal spaces}

Work on a regular parameter stratum $\Theta_S$ with a chosen Euclidean coordinate metric. Let $\Gamma_S$ be a smoothly acting group of state-preserving parameter redundancies, with constant orbit dimension. The objective is also invariant under this action; predictor invariance alone does not imply a null direction of the state quantum Fisher matrix. Its vertical tangent space is
\begin{equation}
\cV_{\theta,S}=\{X_\xi(\theta):\xi\in\operatorname{Lie}(\Gamma_S)\}\subseteq T_\theta\Theta_S.
\end{equation}
The Fisher geometry is \emph{faithful} at $\theta$ if
\begin{equation}
\ker F_S(\theta)=\cV_{\theta,S},
\label{eq:faithful-release}
\end{equation}
which is the $\Gamma$-basic faithfulness condition of \cite{ChenQNGQuotient2026}; it holds at generic points under the conditions of \cite{ChenQNGQuotient2026} and finite-shot rank or gap certification requires the independent structural or separation assumptions specified in \cite{ChenQNGQuotient2026}. The estimated spectrum alone does not identify an exact state-preserving gauge action.
Then the horizontal space is $\mathsf H_{\theta,S}=\cV_{\theta,S}^\perp$ and $P_{\mathsf H,S}$ denotes the horizontal projector. The projector must be recomputed after any release or pruning event. The release-point Fisher orthogonality of \cite{CompanionDiscovery} organizes physical blocks, while \eqref{eq:faithful-release} separately identifies their parameter redundancy.

\subsection{Gauge-projected Tikhonov QNG}

Write QNG for quantum natural gradient and $\widehat F_S$ for an estimate of the quantum Fisher information matrix. A singular Fisher matrix whose kernel is exactly the redundancy directions should be inverted on the horizontal space rather than damped in the ambient one: the Moore--Penrose QNG step is then the minimum-norm horizontal lift of the quotient natural gradient \cite{ChenQNGQuotient2026}, and the circuit realizes the orbit natural gradient up to an explicit metric projection \cite{ChenQNGQuotient2026}. With noisy Fisher and gradient estimates the damping order matters, and the quotient-stable update is the horizontal Tikhonov step of \cite{ChenQNGQuotient2026},
\begin{equation}
\Delta\theta_{\mathsf H}
=-\eta P_{\mathsf H,S}
\bigl(P_{\mathsf H,S}\widehat F_SP_{\mathsf H,S}+\gamma I_{\mathsf H}\bigr)^{-1}
P_{\mathsf H,S}\widehat g,
\label{eq:horizontal-qng}
\end{equation}
in which damping acts only after the exact gauge kernel has been projected out.

Write $P_H=P_{\mathsf H,S}$, $P_V=I-P_H$, $\widehat g=g+\xi$ with $g=P_Hg$, $\bbE\xi=0$, $\Cov(\xi)=\Sigma$, assume the gradient shots are independent of the Fisher shots, and let $\widehat F_H=P_H\widehat F_SP_H$ satisfy $\lambda_{\min}(\widehat F_H+\gamma I_{\mathsf H})\ge\underline m+\gamma>0$. The leading projector makes $P_V\Delta\theta_{\mathsf H}=0$, and for a fixed exact projector the cumulative vertical displacement over a whole trajectory vanishes identically \cite{ChenQNGQuotient2026}. The full-space benchmark $\Delta\theta_{\mathrm{full}}=-\eta(F_S+\gamma I)^{-1}\widehat g$ has no such property: if $F_SP_V=0$ then $(F_S+\gamma I)^{-1}=\gamma^{-1}I$ on the vertical space, so
\begin{equation}
P_V\Delta\theta_{\mathrm{full}}=-\frac{\eta}{\gamma}P_V\xi,
\qquad
\bbE\norm{P_V\Delta\theta_{\mathrm{full}}}^2=\frac{\eta^2}{\gamma^2}\Tr(P_V\Sigma P_V),
\label{eq:vertical-drift}
\end{equation}
which is \cite{ChenQNGQuotient2026}: ambient damping converts vertical shot noise into representative drift at rate $\gamma^{-2}$. And the horizontal stochastic error of \eqref{eq:horizontal-qng} obeys
\begin{equation}
\bbE\!\left(\norm{\Delta\theta_{\mathsf H}-\bbE(\Delta\theta_{\mathsf H}\mid\widehat F_H)}^2\,\middle|\,\widehat F_H\right)
\le
\eta^2\frac{\Tr(P_H\Sigma P_H)}{(\underline m+\gamma)^2},
\label{eq:horizontal-noise}
\end{equation}
since the stochastic part is $-\eta P_H(\widehat F_H+\gamma I_{\mathsf H})^{-1}P_H\xi$; this is the damped analogue of the pseudoinverse bound $\eta^2\norm{\Sigma}_{\op}\rank(F_S)/(f^+_{\min})^2$ of \cite{ChenQNGQuotient2026}, with $\underline m+\gamma$ as the lower spectral scale controlling noise amplification. The trade-off it exposes is the one the certificate must price: damping suppresses horizontal noise but biases the step.

Soft release damps the newly opened horizontal directions. Coordinates with zero gate amplitude are excluded from the trained block; the damping sums below involve only positive gates. Let $R_\alpha\succeq0$ denote a structural damping operator on $\mathsf H_{\theta,S}$, with larger eigenvalues corresponding to harder soft-release penalties. To connect this metric object with the soft gates in \eqref{eq:soft-scale} and \eqref{eq:shrinkage-release}, choose an orthonormal basis $\{u_j\}$ for the released horizontal block, for example after diagonalizing $P_{\mathsf H,S}\widehat F_SP_{\mathsf H,S}$ or the whitened breaking covariance, and set
\begin{equation}
R_\alpha
=\gamma_\alpha\sum_j\alpha_j u_ju_j^\top
=\gamma_\alpha\sum_j\bigl(s_j(\alpha_j)^{-1}-1\bigr)u_ju_j^\top,
\qquad \gamma_\alpha>0.
\label{eq:R-alpha-mapping}
\end{equation}
One may take $\gamma_\alpha=\gamma$ to tie the structural friction to the Tikhonov scale. Directions near the statistical confidence boundary have small $s_j$ and therefore large metric friction, whereas fully released directions have $s_j=1$ and no extra structural damping. The metric eigenvalues used below are $r_j=\gamma_\alpha\alpha_j$ in this chosen horizontal chart; this keeps the statistical soft-release penalties $\alpha_j$ distinct from the QNG damping eigenvalues that control update variance.

\begin{equation}
\Delta\theta_{\mathrm{soft}}
=-\eta P_{\mathsf H,S}
\bigl(P_{\mathsf H,S}\widehat F_SP_{\mathsf H,S}+\gamma I_{\mathsf H}+R_\alpha\bigr)^{-1}
P_{\mathsf H,S}\widehat g.
\label{eq:soft-horizontal-qng}
\end{equation}

\begin{proposition}[Soft-release metric damping]
\label[proposition]{thm:soft-qng}
Assume \eqref{eq:faithful-release}. Condition on pilot information $\mathcal P$ that fixes the Fisher estimate, $R_\alpha$, and the horizontal chart before drawing gradient shots. Let $P_H=P_{\mathsf H,S}$ and $P_V=I-P_H$. The update \eqref{eq:soft-horizontal-qng} satisfies $P_V\Delta\theta_{\mathrm{soft}}=0$. Moreover, suppose that in an orthonormal horizontal chart a direction $u_j$ is a common eigendirection of the horizontal Fisher block used in the update, $P_H\widehat F_SP_H$, and $R_\alpha$, with eigenvalues $\widehat f_j\ge0$ and $r_j\ge0$. Then
\begin{equation}
\abs{\ip{u_j}{\Delta\theta_{\mathrm{soft}}}}
\le
\frac{\eta\abs{\ip{u_j}{\widehat g}}}{\widehat f_j+\gamma+r_j}
\le
\frac{\eta\abs{\ip{u_j}{\widehat g}}}{\gamma+r_j},
\label{eq:soft-qng-coordinate-bound}
\end{equation}
and, if $\widehat g=g+\xi$ with $\bbE(\xi\mid\mathcal P)=0$, then conditionally on $\mathcal P$
\begin{equation}
\bbE\Bigl(\abs{\ip{u_j}{\Delta\theta_{\mathrm{soft}}-\bbE(\Delta\theta_{\mathrm{soft}}\mid\mathcal P)}}^2\,\Big|\,\mathcal P\Bigr)
\le
\eta^2\frac{\Var(\ip{u_j}{\xi}\mid\mathcal P)}{(\widehat f_j+\gamma+r_j)^2}.
\label{eq:soft-qng-noise-bound}
\end{equation}
Thus a large metric penalty suppresses both the update amplitude and the shot-noise variance in that horizontal direction at rates $O(r_j^{-1})$ and $O(r_j^{-2})$. Under \eqref{eq:R-alpha-mapping}, this is the same monotone suppression as increasing the statistical soft-release penalty $\alpha_j$.
\end{proposition}

\begin{proof}
The leading projector in \eqref{eq:soft-horizontal-qng} gives $P_V\Delta\theta_{\mathrm{soft}}=0$. In the stated horizontal normal form, the inverse damped Fisher operator acts on $u_j$ by multiplication with $(\widehat f_j+\gamma+r_j)^{-1}$. Taking the $u_j$ component gives \eqref{eq:soft-qng-coordinate-bound}. Conditionally on $\mathcal P$ the eigenvalue $\widehat f_j$ is fixed, so applying the same scalar inverse to the noise component $\ip{u_j}{\xi}$ and squaring gives \eqref{eq:soft-qng-noise-bound}.
\end{proof}

The common-eigendirection hypothesis is the zero-coupling case of a general law. When $R_\alpha$ and the horizontal Fisher block do not commute, the exact coupling between a damped soft block and the remaining horizontal directions is the Schur-complement rule of \cite{ChenQNGQuotient2026}, and \cref{thm:soft-qng} is its restriction to an invariant subspace. The scale $\gamma_\alpha$ in \eqref{eq:R-alpha-mapping} may equally be set from the confidence rule of \cite{ChenQNGQuotient2026}---a scalar ridge calibrated on an independent pilot budget---in place of the shrinkage rule \eqref{eq:shrinkage-release}; \cref{prop:ridge-scale} records the dictionary between the two conventions.

An estimated horizontal projector introduces leakage proportional to its subspace error. The following inverse-norm formulation follows the argument of \cite{ChenQNGQuotient2026}; it states the estimated-block lower bound directly.

\begin{proposition}[Leakage under an estimated projector]
\label[proposition]{thm:estimated-projector}
Let $P_H$ be the true horizontal projector and $\widehat P_H$ an estimate with $\norm{\widehat P_H-P_H}_{\op}\le\varepsilon_P$, and suppose the true horizontal Fisher block has minimum eigenvalue $m>0$ and the estimated damped block satisfies $\lambda_{\min}(\widehat P_H\widehat F\widehat P_H+\gamma I_{\widehat H})\ge m+\gamma-\varepsilon_F>0$. Then the estimated gauge-projected update
\begin{equation}
\widehat\Delta
=-\eta\widehat P_H(\widehat P_H\widehat F\widehat P_H+\gamma I_{\widehat H})^{-1}\widehat P_H\widehat g
\label{eq:estimated-horizontal-update}
\end{equation}
has true vertical component
\begin{equation}
\norm{P_V\widehat\Delta}
\le
\eta\varepsilon_P\frac{\norm{\widehat P_H\widehat g}}{m+\gamma-\varepsilon_F},
\qquad
\bbE\norm{P_V\widehat\Delta}^2
\le
\eta^2\varepsilon_P^2\frac{G^2+\Tr\Sigma}{(m+\gamma-\varepsilon_F)^2}
\label{eq:estimated-vertical-leakage}
\end{equation}
The expectation bound assumes the projector and inverse-norm hypotheses hold almost surely, and $\bbE\norm{\widehat g}^2\le G^2+\Tr\Sigma$; the pathwise bound holds on their joint confidence event.
\end{proposition}

\begin{proof}
Since $P_VP_H=0$, $P_V\widehat P_H=P_V(\widehat P_H-P_H)$ and hence $\norm{P_V\widehat P_H}\le\varepsilon_P$. Apply this to \eqref{eq:estimated-horizontal-update}, use the assumed inverse-norm bound, then square and take expectations.
\end{proof}

An exact projector removes instantaneous vertical shot noise identically; a fixed exact projector also gives zero cumulative vertical displacement \cite{ChenQNGQuotient2026}; \cref{thm:estimated-projector} says imperfect quotient identification degrades that guarantee continuously in the subspace angle. The perturbation bound in \cite{ChenQNGQuotient2026} replaces the hypothesis on $\lambda_{\min}$ by an explicit bound in the primitive radii, with denominator $m+\gamma-\delta_A$ for $\delta_A=\varepsilon_F+(2f_{\max}+f_{\min})\varepsilon_P+f_{\max}\varepsilon_P^2$ and numerator $\norm{g}+\varepsilon_g$. Since $\delta_A\ge\varepsilon_F$ that form is more conservative than \eqref{eq:estimated-vertical-leakage}, and it is the one to use when only $\norm{\widehat F-F}_{\op}$ and $\varepsilon_P$ are controlled, the hypothesis above being an assumption on the estimated block rather than a consequence of those radii.

\subsection{Identifying physical release directions}

A physical release direction must have positive Fisher length at the probing state. For pure states, $F(v,v)=4\Var(A_v)$ and a fidelity central difference provide the directional test with its bias bound \cite{ChenQNGQuotient2026}; for mixed states, use the symmetric-logarithmic-derivative Fisher Gram in \eqref{eq:gram}. The released Fisher kernel must then be recomputed. A known state-preserving action, or independent rank-and-gap information together with a valid spectral confidence event, determines which directions can be removed \cite{ChenQNGQuotient2026}. Damping acts on the resulting horizontal block through \eqref{eq:horizontal-qng}.

\subsection{End-to-end release certificate}
\label{subsec:end-to-end}

The release, validation and optimization statements proved separately above hold simultaneously, on one high-probability event and under one confidence budget. Stronger concentration bounds, cross-fitting schemes, or optimizer-specific descent estimates replace the displayed terms without changing the argument.

\begin{corollary}[End-to-end auditable release guarantee]
\label[corollary]{cor:end-to-end}
Run the release workflow with three independent streams $\mathcal D_{\fit},\mathcal D_{\rel},\mathcal D_{\valsplit}$. When the paired-copy gate is used, suppose its allocated release samples yield, simultaneously over every screened subgroup, an event
\begin{equation}
\abs{\widehat D_H^{\mathrm{task}}-D_H^{\mathrm{task}}}\le r_{D,H}
\label{eq:end-to-end-D-event}
\end{equation}
with probability at least $1-\delta_D$. Let $A=K-S_M-2G_N$ denote the penalized release matrix in the coordinate metric used for spectral selection, including Gram whitening when it is applied. Suppose the remaining release stream yields, with probability at least $1-\delta_{\rel}$, a simultaneous confidence event
\begin{equation}
\|\widehat A-A\|_{\op}\le r_A.
\label{eq:end-to-end-A-event}
\end{equation}
If paired-copy gating is bypassed, set \(\delta_D=0\) and omit the gate clause for that subgroup.
If the Gram matrix, shot penalty, or structural penalty is estimated, its perturbation is included in $r_A$ after the same whitening convention. All group and direction families tested by these events are fixed before their measurement batches, or covered uniformly. Suppose also that the validation stream satisfies the conditional simultaneous risk event in \cref{thm:post-selection-oracle} with probability at least $1-\delta_{\valsplit}$, and that the projector-estimation and inverse-norm events used in \cref{thm:estimated-projector} hold jointly over all reported training steps with probability at least $1-\delta_P$. Select release directions only when the corresponding eigenvalues of $\widehat A$ exceed $r_A$ and mark $[-r_A,r_A]$ as ambiguous. Train every selected structure by the gauge-projected QNG update \eqref{eq:horizontal-qng}, by the soft-deformed update \eqref{eq:soft-horizontal-qng}, or by the estimated-projector version \eqref{eq:estimated-horizontal-update}. Then, with probability at least $1-\delta_D-\delta_{\rel}-\delta_{\valsplit}-\delta_P-\delta_{\disc}$, where $\delta_{\disc}$ is the discovery failure probability charged in \cref{thm:post-selection-oracle} and may be dropped when the retained family is specified rather than discovered:
\begin{enumerate}
\item every gate-certified subgroup with \(\widehat D_H^{\mathrm{task}}-r_{D,H}>0\) has positive global task breaking mass, while \(\widehat D_H^{\mathrm{task}}+r_{D,H}<r_{\min}^2\) excludes Hilbert--Schmidt-normalized breaking gaps of size \(r_{\min}\); for operator-norm-bounded readouts use the rank-adjusted criterion of \cref{prop:operational-gap};
\item every opened spectral direction has positive true penalized quadratic gain in $A=K-S_M-2G_N$;
\item whenever all stationarity, smoothness, radius, and step-range conditions in \cref{thm:certified-release-descent} hold for a certified leading direction, the corresponding local step satisfies the descent bound \eqref{eq:certified-descent-drop};
\item if both the positive and negative spectra are separated from zero by more than $2r_A$, its release dimension is recovered and its subspace error obeys the Davis--Kahan bound \eqref{eq:dk-release};
\item the final selected structure satisfies the post-selection oracle inequality \eqref{eq:post-selection-oracle};
\item for each estimated-projector step the pathwise leakage bound holds; on the joint projector event its accumulated squared local leakage satisfies
\begin{equation}
\sum_{t=1}^T\|P_{V,t}\widehat\Delta_t\|^2
\le\sum_{t=1}^T\eta_t^2\varepsilon_{P,t}^2
\frac{\|\widehat P_{H,t}\widehat g_t\|^2}{(m_t+\gamma_t-\varepsilon_{F,t})^2}.
\label{eq:cumulative-leakage}
\end{equation}
This is a sum of local leakage energies. An endpoint-displacement statement requires a fixed projector or a specified transport between changing tangent spaces.
\end{enumerate}
\end{corollary}

\begin{proof}
The gate claim follows from \eqref{eq:end-to-end-D-event} and the duality \eqref{eq:breaking-mass-duality}. For each selected unit eigenvector $v$, $v^\top Av\ge v^\top\widehat Av-r_A>0$ on \eqref{eq:end-to-end-A-event}; Weyl's inequality additionally controls the positive eigenvalue count. The descent statement is \cref{thm:certified-release-descent}. The subspace claim is \cref{thm:subspace-stability}. The oracle claim is \cref{thm:post-selection-oracle} after conditioning on $\mathcal D_{\fit}\cup\mathcal D_{\rel}$. The leakage statement is \cref{thm:estimated-projector} applied at each training step and summed. A union bound over the gate, release, validation, and projector-estimation events gives the displayed probability.
\end{proof}

\section{Finite-Sample Releasable Equivariant Learning Algorithm}
\label{sec:algorithm}

The algorithm fixes the order in which the gate, the spectral test and the validation comparison are allowed to see the data, since that order is what keeps the final margin valid.

\begingroup
\captionsetup{hypcap=false}
\par\medskip\hrule\smallskip
\captionof{algorithm}{Statistical symmetry release with independent validation}
\label{alg:rsrmqng}
\begin{algorithmic}[1]
\Require Independent fitting, release, and validation streams; a breaking dictionary; measurement and confidence budgets; a declared readout norm and relevance level $r_{\min}$.
\State Specify the finite subgroup family $\mathscr H$, or discover it before release probing using \cite{CompanionDiscovery}, with kernel task validation \cite{Gretton2012} when classical samples are available. Allocate $\delta_{\disc}$ when discovery is used.
\For{$H\in\mathscr H$}
  \State Fit the hard backbone on $\mathcal D_{\fit}$ and store it as a candidate.
  \If{paired-copy gating is predeclared}
    \State Estimate $D_H^{\mathrm{task}}$ with radius $r_{D,H}$ by \cref{thm:twirl-swap-gate}.
    \State Set $c_H=1$ for unit Hilbert--Schmidt readouts; for unit operator-norm readouts use a certified rank bound as in \cref{prop:operational-gap}, or $c_H=\dim\cH$.
    \If{$c_H(\widehat D_H^{\mathrm{task}}+r_{D,H})<r_{\min}^2$}
      \State Record retention at the declared signal scale and continue to the next subgroup.
    \EndIf
    \State Use sector intervals to prioritize the dictionary; follow the predeclared local fallback when the gate is ambiguous.
  \EndIf
  \State Remove commutant components and normalize the breaking dictionary by its chosen Gram matrix. Freeze the probing backbone.
  \State Estimate the loss-corrected curvature $\widehat K$ and, when needed, the first-order benefit, with simultaneous radii that include sampling and truncation errors.
  \If{first-order stationarity is established}
    \State Form $\widehat A=\widehat K-S_M-2G_N$ in the certified coordinate metric. Set $S=\operatorname{span}\{B(v_j):\widehat\lambda_j>r_A\}$ for its eigenpairs $(\widehat\lambda_j,v_j)$. If $S=\{0\}$, continue to the next subgroup; treat $[-r_A,r_A]$ as ambiguous.
    \State Initialize one certified direction by \cref{thm:certified-release-descent} when its step-range conditions hold.
  \Else
    \State Compute the maximizing step $\widehat\beta$ in \cref{thm:first-order-safe-release}. If its certificate fails, continue to the next subgroup; otherwise set $S=\operatorname{span}\{B(\widehat\beta)\}$ and initialize at $\widehat\beta$.
  \EndIf
  \State Construct the predeclared soft-amplitude candidates. Recompute each released Fisher kernel and horizontal projector, then train using \eqref{eq:horizontal-qng} or \eqref{eq:soft-horizontal-qng}.
  \State Store the trained predictors and their prevalidation costs. A later curvature kick requires a certificate at its new backbone.
\EndFor
\State Evaluate all stored predictors on $\mathcal D_{\valsplit}$ and select by \cref{thm:post-selection-oracle}. Validation does not feed back into candidate generation.
\State Return the selected $(H,S,\alpha)$, its signal and release intervals, risk bound, and geometric diagnostics.
\end{algorithmic}
\smallskip\hrule\medskip
\endgroup

At a release-stationary point, the spectral rule compares estimated quadratic benefit with the declared costs and a simultaneous uncertainty margin. For a unit direction $v$ in the certified Gram-whitened coordinates, the sufficient condition is
\begin{equation}
\underbrace{v^\top \widehat K v}_{\text{release benefit}}
>
\underbrace{v^\top S_M v}_{\text{activation cost}}
+
\underbrace{2v^\top G_N v}_{\text{structural penalty}}
+
\underbrace{r_\delta}_{\text{multiple-testing margin}}.
\label{eq:decision-one-line}
\end{equation}
Equivalently, the spectral implementation opens the positive eigenspace of \(\widehat K-S_M-2G_N-r_\delta I\), or of its Gram-whitened version. This form avoids treating noncommuting penalty matrices as if their eigenvalues could be compared coordinatewise.

\section{Numerical Experiments}
\label{sec:numerics}

The experiments examine global signal detection, finite-shot direction selection, and geometric training. Quantum-state simulations evaluate measurement costs directly; Gaussian and regression models isolate the statistical mechanisms. The numerical models, seeds, and measurement budgets are specified in Appendix~\ref{app:repro}.
\FloatBarrier

\subsection{What a hard symmetry hides, and the quantum-native gate}
\label{subsec:exp-hidden}

The spin-flip example isolates information erased by a hard readout constraint. The paired-copy and sector experiments then measure that information under different access models.

\subsubsection{Exact twirling blind-spot classification}
The first experiment tests \cref{cor:binary-release-separation}. Let $H=\{I,X^{\otimes n}\}$ and
\begin{equation}
\rho_y=\left(\frac{I+ymZ}{2}\right)^{\otimes n},\qquad y\in\{+1,-1\}.
\end{equation}
The two class-conditional states have the same $H$-twirl, so every hard spin-flip-invariant measurement is at chance for the sign label. The released observable is the normalized magnetization $A=M_z/n$. We simulate $n=6$ qubits, $M$ independent state copies per test item, and direct $Z$ measurements of all qubits. The horizontal axis in \cref{fig:twirling-blindspot-classification} is the released signal $m\sqrt{nM}$, the scale controlling the empirical magnetization sign test. Hard invariant readouts remain at accuracy $1/2$, while the released sector rapidly approaches perfect classification once the order parameter is resolved above shot noise.

\begin{figure}[!htbp]
\centering
\includegraphics[width=0.74\linewidth]{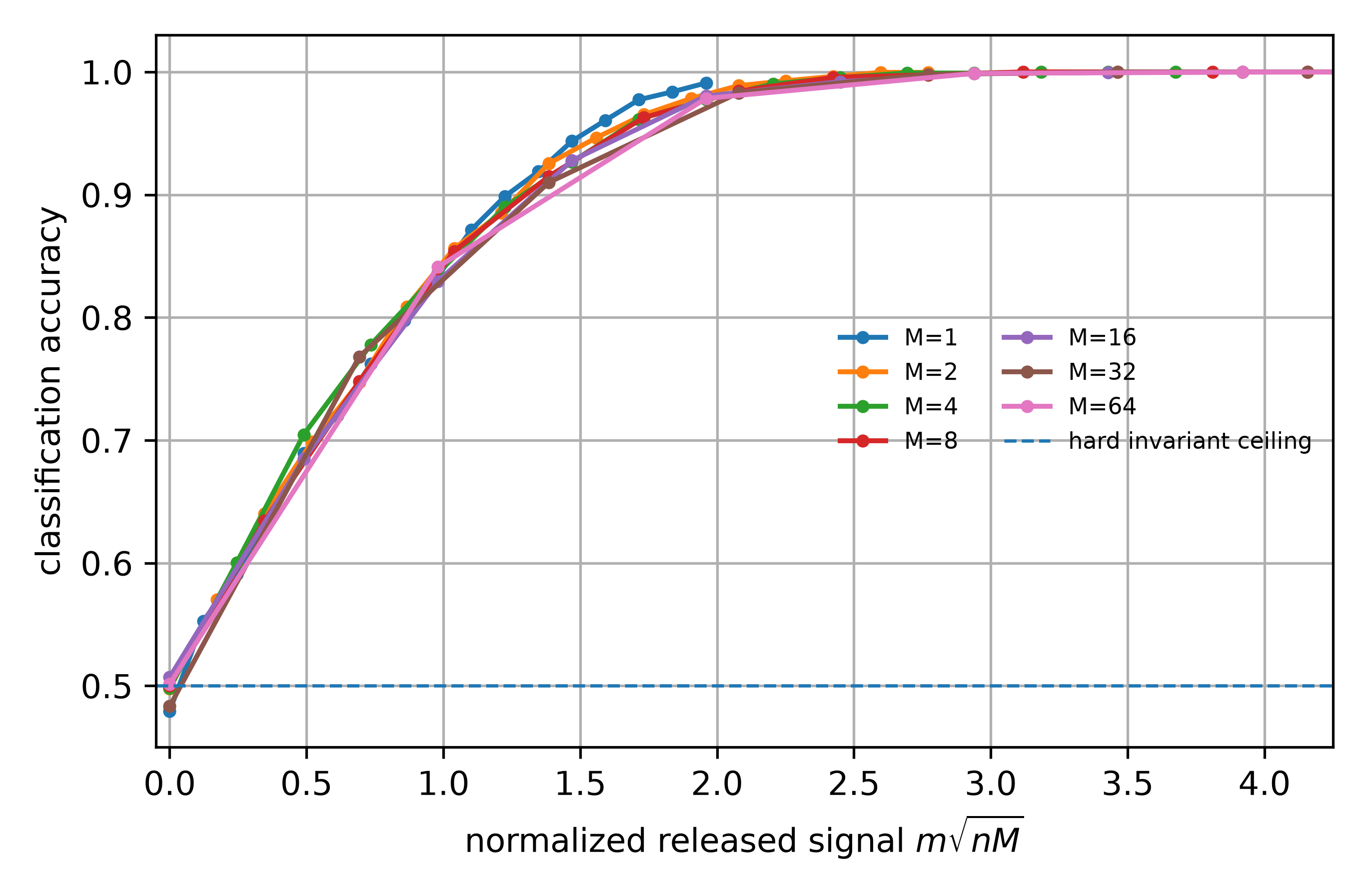}
\caption{Exact twirling blind-spot classification. The two states have identical spin-flip twirls, so hard invariant measurements have chance-level accuracy. Releasing the odd order parameter $M_z/n$ recovers the label once the normalized shot signal $m\sqrt{nM}$ is of constant order.}
\label{fig:twirling-blindspot-classification}
\end{figure}

\subsubsection{Quantum-native gate, sector spectroscopy, and the point-null boundary}
We now compare measurement-access models at the same total budget of \(B=2000\) state copies. For \(H=\{I,X^{\otimes n}\}\), both estimators target the same state breaking mass \(D_H\). The twirl--swap estimator uses \(B/4\) pairs in each arm. Its mean absolute error remains between \(0.045\) and \(0.059\) from \(d=4\) to \(d=128\). A single-copy local-randomized-measurement estimator uses fifty random product bases and twenty shots per state in each basis. Its mean error rises from \(0.013\) to \(0.127\) and crosses the paired-copy error near \(d=32\). This comparison exhibits a dimension crossover for these two state-overlap estimators at equal total copy count; the calibrated Gaussian crossover in \cref{sec:quantum-native} concerns a different local observation model.

For \(H=S_3\) acting by qubit permutations on three qubits, one shared set of \(200{,}000\) twirled SWAP outcomes estimates both nontrivial sectors. The exact and estimated masses are \(0.01011\) and \(0.00987\) in the sign sector, and \(0.06614\) and \(0.06342\) in the standard sector. The empirical per-sample variances are \(0.9999\) and \(3.9955\), below the respective \(d_\lambda^2\) bounds \(1\) and \(4\).

The final panel tests the full Pauli scan in \cref{prop:pauli-point-null}, including all \(d^2-1=255\) statistics and an independently simulated null. At \(d=16\), \(D_H=r^2=1/64\), and \(1989\) used copies from the nominal budget of \(2000\), the empirical false-release probability is \(0.003\) and power is \(0.8695\). Under the same total-copy budget, the twirl--swap signal-to-noise ratio is only \(0.247\); the conservative bound \eqref{eq:twirl-swap-state-error} requires \(665{,}214\) pairs per arm. The twirl--swap gate buys robustness to an unknown symmetric backbone, and a structured point null is where a measurement acting on one copy at a time can beat it.

\begin{figure}[!htbp]
\centering
\includegraphics[width=0.98\linewidth]{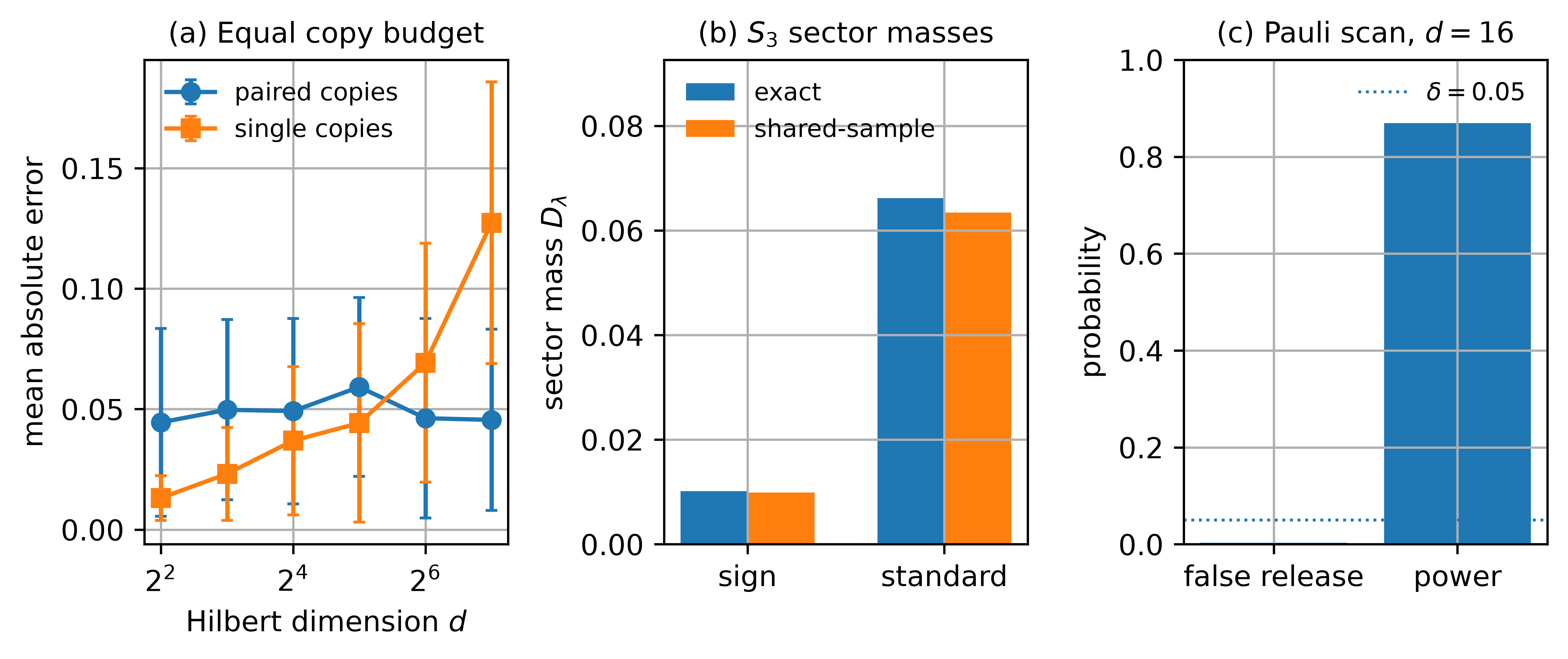}
\caption{Quantum-native detection mechanisms. (a) At equal total copy budget, the two-copy twirl--swap error is stable with Hilbert dimension, while this single-copy randomized estimator degrades. (b) Character reweighting of one shared sample set recovers both nontrivial \(S_3\) sector masses. (c) The complete nonadaptive Pauli scan has controlled false release and high power at a point null where the dimension-free twirl--swap estimator is not instance-optimal.}
\label{fig:quantum-native-detection}
\end{figure}

\FloatBarrier

\subsection{Detection boundaries}
\label{subsec:exp-boundaries}

Gaussian sequence simulations compare the dense norm threshold with the sparse dictionary threshold at controlled false-release levels.

\subsubsection{Local release detection boundary}
The dense simulation uses $Z=\mu+\mathcal N(0,I_k)$, $k\in\{2,4,8,16,32\}$, and the norm test \eqref{eq:norm-test} at $\delta=0.05$. Plotting power against $\|\mu\|^2/\sqrt{k}$ compares the dimensions on the scale in \cref{thm:detection-boundary}. Power rises as this normalized signal exceeds a constant-order threshold. The simulation illustrates the calibrated test; the mixture argument establishes the lower bound for all tests.

\begin{figure}[!htbp]
\centering
\includegraphics[width=0.74\linewidth]{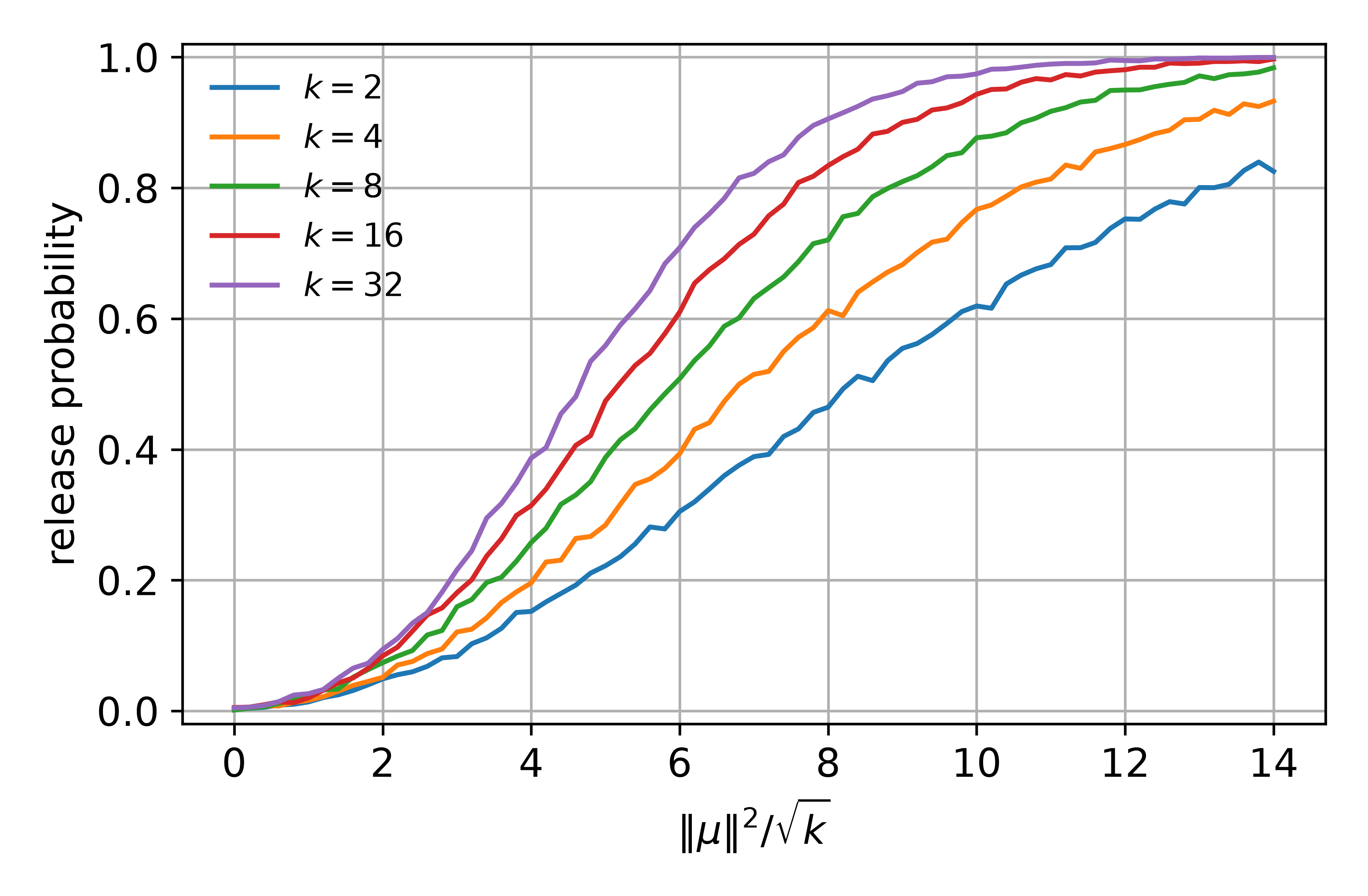}
\caption{Gaussian local release detection boundary. The release test is calibrated at false-release level $\delta=0.05$. Power is controlled by the normalized signal $\|\mu\|^2/(\tau^2\sqrt{k})$, consistent with \cref{thm:detection-boundary}.}
\label{fig:detection-boundary}
\end{figure}

\subsubsection{Sparse dictionary search and direction recovery}
The previous boundary assumes the branch is known; we now test the unknown-direction boundary in \cref{thm:sparse-dictionary-boundary} with a sparse spiked model. There are \(p=64\) candidate breaking directions and a single true index \(j_\star\). The observed statistic is
\begin{equation}
Y_j=\mu\mathbf 1\{j=j_\star\}+\sigma M^{-1/2}Z_j,
\qquad Z_j\sim\mathcal N(0,1).
\end{equation}
The test releases if \(\max_jY_j>\sigma\sqrt{2\log(p/0.05)/M}\). The horizontal axis is the normalized signal \(a=\mu\sqrt M/(\sigma\sqrt{\log p})\). \Cref{fig:spiked-boundary} shows that both release and correct-direction recovery become reliable only once this normalized signal is of constant order, confirming the \(\sqrt{\log p/M}\) price of searching over a dictionary.

\begin{figure}[!htbp]
\centering
\begin{subfigure}{0.48\linewidth}
\centering
\includegraphics[width=\linewidth]{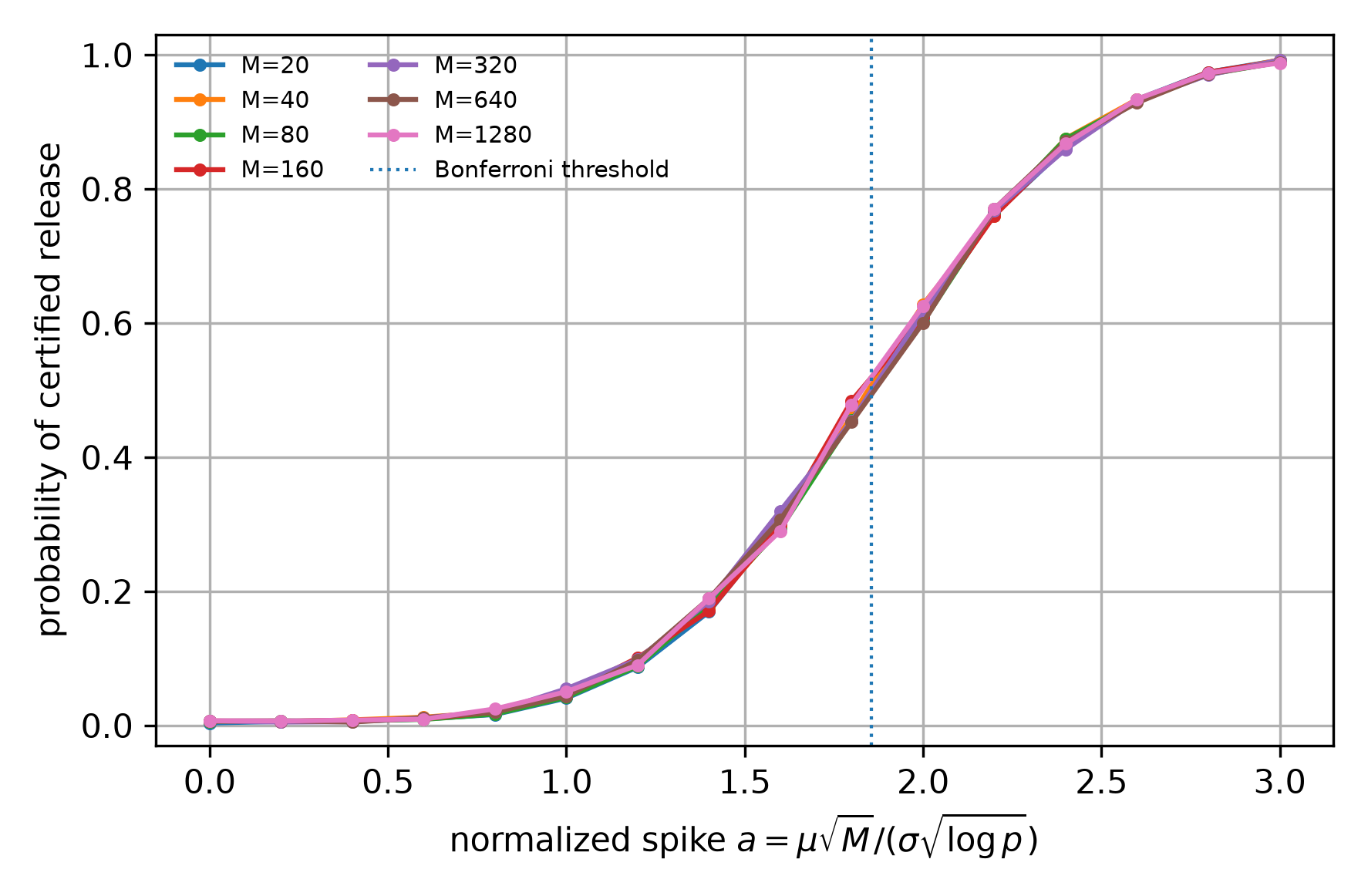}
\caption{Release probability.}
\end{subfigure}\hfill
\begin{subfigure}{0.48\linewidth}
\centering
\includegraphics[width=\linewidth]{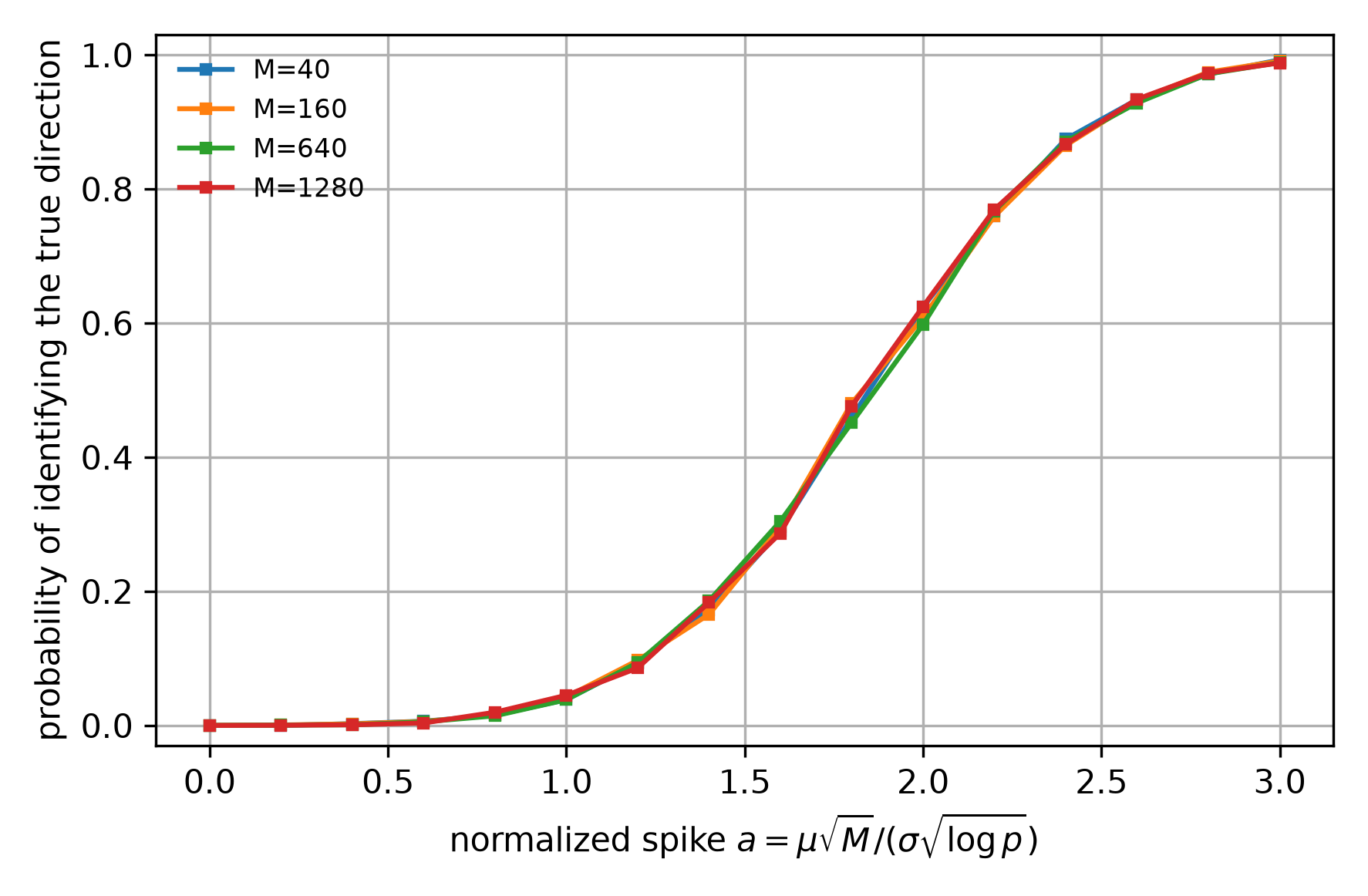}
\caption{Correct direction recovery.}
\end{subfigure}
\caption{Spiked detection boundary for an unknown release direction. Below the \(\sigma\sqrt{\log p/M}\) scale, release and direction recovery are both unreliable; above it, both rise sharply.}
\label{fig:spiked-boundary}
\end{figure}

\FloatBarrier

\subsection{Ising release curvature and executable certification}
\label{subsec:exp-ising}

The Ising model distinguishes physical symmetry-breaking fluctuations from loss-dependent release curvature and supports a direct comparison of the two measurement modes.

\subsubsection{Task-relative release curvature across the Ising transition}
We exactly diagonalize the periodic transverse-field Ising Hamiltonian \cite{Pfeuty1970,Sachdev2011}
\begin{equation}
H(h)=-\sum_{i=1}^L Z_iZ_{i+1}-h\sum_{i=1}^L X_i,
\label{eq:tfim}
\end{equation}
for $L\in\{6,8,10,12,14\}$ and $h\in\{0.5,1.0,1.5\}$. The model has a global $\mathbb Z_2$ spin-flip symmetry $P=\prod_iX_i$. In finite symmetric ground states, $\langle M_z\rangle=0$, but the order-parameter quantum Fisher information (QFI) signal \cite{CamposVenutiZanardi2007}
\begin{equation}
F_{M_z}=4\Var(M_z),\qquad M_z=\sum_iZ_i,
\end{equation}
measures order-parameter fluctuations in a pure ground state. The corresponding finite-size diagnostic is reported in \cite{CompanionDiscovery}; the calculations here also include $L=6$. Release is determined by the chosen loss and generator, even when this fluctuation scale is large.

For a task-relative release curvature, take $C=-(M_z/L)^2$ and $B=M_y/2$. A $y$-rotation gives
\begin{equation}
\kappa_B=-\cL''(0)=\frac{2}{L^2}\left(\langle M_x^2\rangle-\langle M_z^2\rangle\right).
\label{eq:ising-kappa}
\end{equation}
\Cref{fig:tfim-release} shows that the curvature is positive on the disordered side, negative on the ordered side, and crosses near criticality. Thus release curvature is a local task-improvement test, not a stand-alone phase label.

\begin{figure}[!htbp]
\centering
\includegraphics[width=0.74\linewidth]{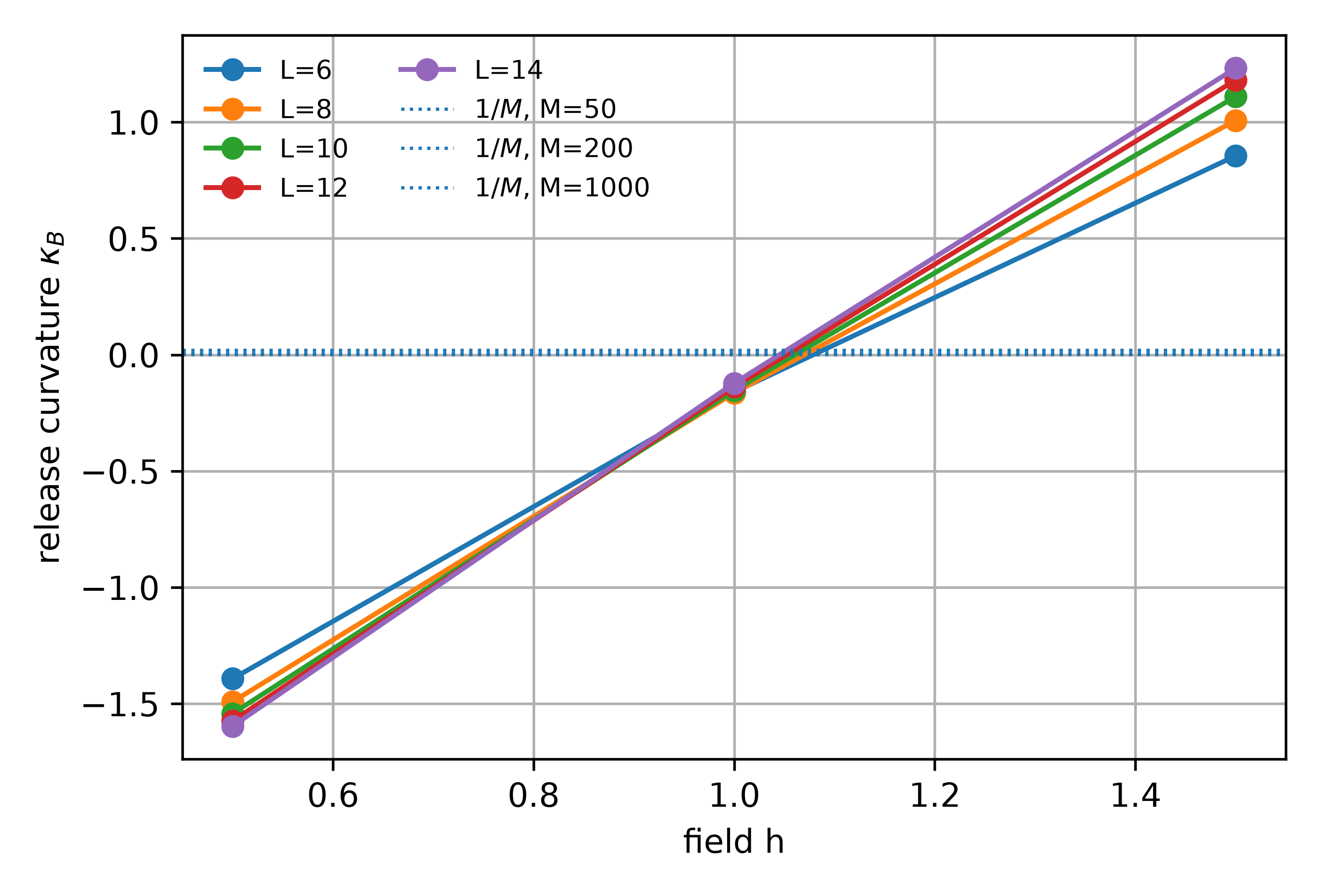}
\caption{Task-relative release curvature for $C=-(M_z/L)^2$ and $B=M_y/2$. Positive curvature certifies a local release improvement only after finite-shot and structural penalties are subtracted.}
\label{fig:tfim-release}
\end{figure}

\subsubsection{Executable certification on an Ising state}
\label{subsec:executable-certification}

The same Ising state provides a direct test of the estimators in \cref{subsec:shot-estimation,app:fd}. All decision radii use measured records and known operators; exact state information is used afterward to evaluate coverage and descent.

The setting is the $L=8$ periodic transverse-field Ising ground state $\rho_0(h)$ of \eqref{eq:tfim} with retained group $H=\mathbb Z_2$. The validation observable is $C=-(M_z/L)^2$, which is $\mathbb Z_2$-invariant and read out from a single $Z$-basis shot, so one loss evaluation costs one shot. The breaking dictionary has $p=6$ translation-invariant, spin-flip-odd, generators orthonormal for the normalized trace, $\Tr(B_aB_b)/2^L=\delta_{ab}$, of Pauli weight at most two: $M_y$, $M_z$, and the four combinations $\sum_i(X_iY_{i+1}\pm Y_iX_{i+1})$ and $\sum_i(X_iZ_{i+1}\pm Z_iX_{i+1})$. Exact diagonalization gives the population curvature matrix $K$; its spectrum is $(1.948,0.097,0,0,-0.16,-0.22)$ at $h=1.5$, $(0.060,0.052,0,0,-0.67,-0.78)$ at $h=1.0$, and $(0.0137,0,0,-0.70,-1.55,-2.79)$ at $h=0.5$. The two near-zero eigenvalues belong to the antisymmetric combinations $\sum_i(X_iY_{i+1}-Y_iX_{i+1})$ and $\sum_i(X_iZ_{i+1}-Z_iX_{i+1})$, whose curvature is below $10^{-4}$ at every field: these are unresolved directions at the measured budgets of \cref{cor:spectral-certificate}(iii), and they occur here without being engineered. The commuting generator $M_z$ behaves differently. Since $[M_z,C]=0$ its diagonal entry vanishes, but it stays coupled off-diagonally to $\sum_i(X_iZ_{i+1}+Z_iX_{i+1})$ at strength $0.15$ to $0.20$, and at $h=0.5$ it is in fact the leading positive mode: the unit coordinate eigenvector of $\lambda_1=0.014$ is approximately $-0.996\,e_{M_z}-0.093\,e_{XZ+ZX}$, where the coordinates refer to the normalized dictionary generators. A zero diagonal entry is therefore not a zero of the release matrix, which is why the certificate is applied to the spectrum and not entrywise.

For scalar probes, each of the $q=21$ polarization directions uses $M$ shots at $0$ and at $\pm\delta_\beta$. The statistical part is a simultaneous two-sided empirical Bernstein interval over all $3q$ loss estimates. The truncation part is computed from the known operators through \eqref{eq:fd-operator-bias}, without the ground state or the population curvature. Their polarization envelope supplies one radius for both release selection and descent. Here $S_M=M^{-1}I$ and $G_N=0$ are declared surrogate costs, and $L_3=8\|B(\widehat v_1)\|_{\op}^3\|C\|_{\op}$ supplies the step bound. The state-independent matrix bias bounds are $0.1252$, $0.7588$, and $1.8375$ for $\delta_\beta=0.1,0.25,0.4$, respectively; the smaller realized biases $0.0290,0.1748,0.4200$ at $h=1.5$ are used only for retrospective comparison.

Across $200$ repetitions at each of $45$ scalar settings, every operator error lies within its radius and no selected direction has nonpositive true penalized gain. At $h=1.5$, all $200$ repetitions certify the leading direction at $M=10^7$ for $\delta_\beta=0.1$, and at $M=10^6,10^7$ for $\delta_\beta=0.25$. The descent bound holds in all $600$ asserted steps. The $\delta_\beta=0.4$ grid gives no certificate: its smaller sampling variance is offset by the larger state-independent bias bound. At $h=1.0$ and $0.5$, the tested budgets certify no release. These are finite experimental checks of coverage and descent, alongside the distributional guarantees proved above.

For matrix-element estimation, the $21$ curvature observables contain $1248$ distinct Pauli strings of weight at most four. Their maximum coefficient $\ell_1$ norm is $11.0$, their weighted-sum shadow-norm upper bound is $78.0$, and the maximum deterministic snapshot bound in \eqref{eq:deterministic-shadow-range} is $103.5$. One shared bank estimates every entry, using the empirical variances and deterministic ranges in \cref{prop:empirical-curvature-certificate}. At $h=1.5$, the mean radius is $5.842$ at $10^3$ snapshots, which certifies no direction; it falls to $0.883$ at $10^4$ snapshots and $0.183$ at $10^5$, and every one of the $20$ repetitions at each of these latter budgets certifies a leading direction. The corresponding mean operator errors are $0.377$, $0.106$, and $0.026$.

On the tested grids, shared shadows therefore certify release at $10^4$ total shots, while the scalar mode at $\delta_\beta=0.25$ first does so at $6.3\times10^7$ total shots, a factor of $6300$. All $280$ repetitions in the mode comparison satisfy their confidence radii and produce zero false releases. The comparison prices these specified estimators and bias bounds; it does not optimize all measurement strategies. Low-weight commutators explain the advantage of shared records here. A nonlocal dictionary can change that ordering, and the scalar method remains available when direct commutator measurement is expensive.

\begin{figure}[!htbp]
\centering
\includegraphics[width=0.98\linewidth]{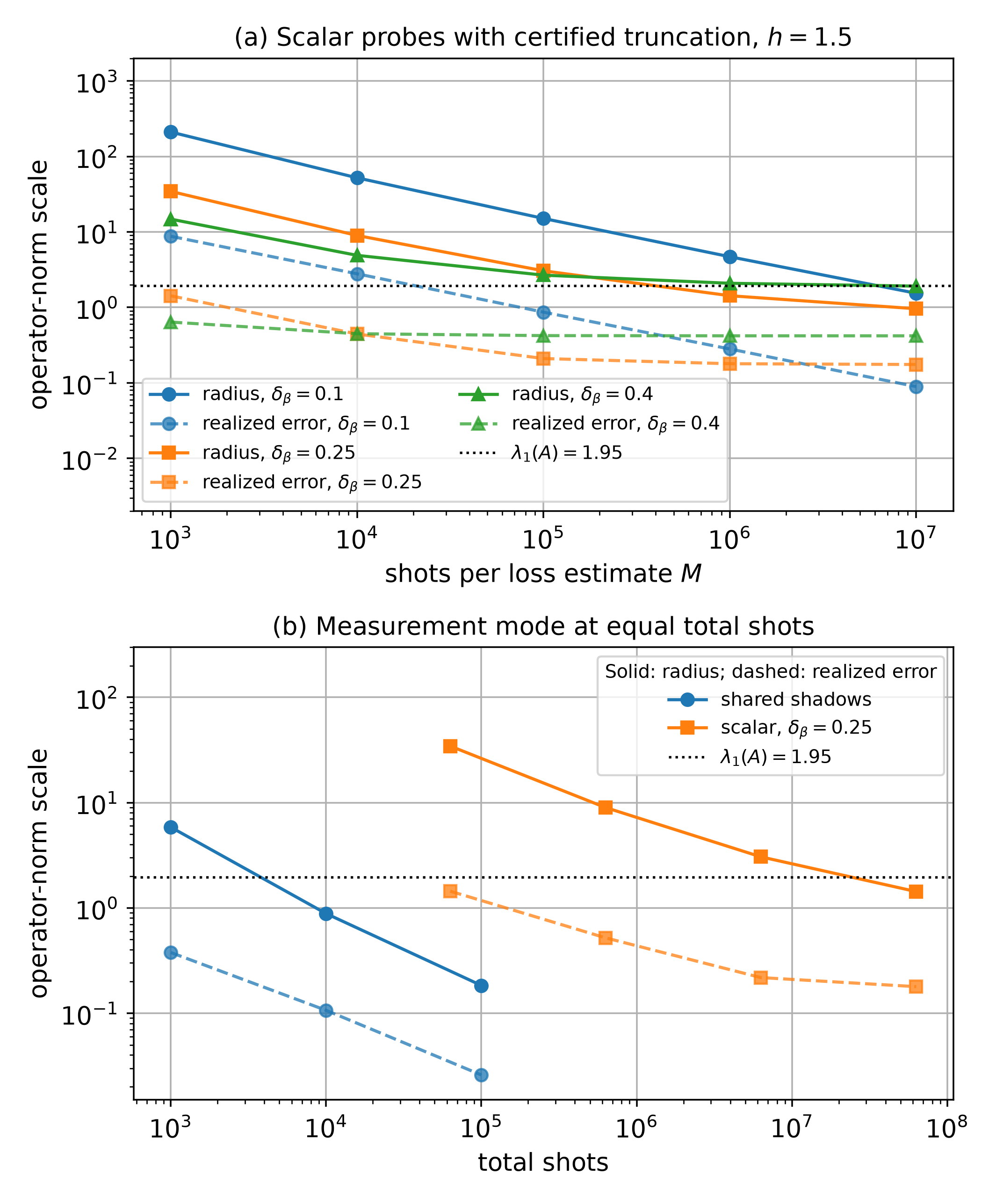}
\caption{Certified release on the $L=8$ Ising ground state at $h=1.5$. (a) Scalar-probe radii include state-independent truncation bounds; dashed curves show realized errors against the exact curvature. (b) Shared-shadow and scalar-probe radii versus total shots. Every decision uses the empirical eigenvalue and its radius; the dotted population eigenvalue is only a reference scale. On the tested grids the first fully successful budgets differ by a factor of $6300$.}
\label{fig:executable-certification}
\end{figure}

\FloatBarrier

\subsection{Release selection with finite data and finite shots}
\label{subsec:exp-selection}

Synthetic curvature matrices isolate estimation and subspace selection. Regression examples examine the corresponding capacity cost and the protection supplied by independent validation.

\subsubsection{Multi-branch release matrix and spectral selection}
We construct an $8\times8$ population curvature matrix with two positive eigenvalues after the declared $0.08I$ penalty and six negative penalized modes. The observed matrix is $\widehat K=K+W$, where the symmetric noise standard deviation scales as $M^{-1/2}$. The selection rule releases the positive spectrum of $\widehat K-0.08I-r_MI$, with the aggregate deterministic quadratic penalty $0.08I$ and the illustrative margin $r_M=0.9M^{-1/2}$. \Cref{fig:multi} shows that the selected dimension approaches the true dimension and the false-release count goes to zero as shot budget increases; the estimated subspace error decreases consistently with \cref{thm:subspace-stability}.

\begin{figure}[!htbp]
\centering
\begin{subfigure}{0.48\linewidth}
\centering
\includegraphics[width=\linewidth]{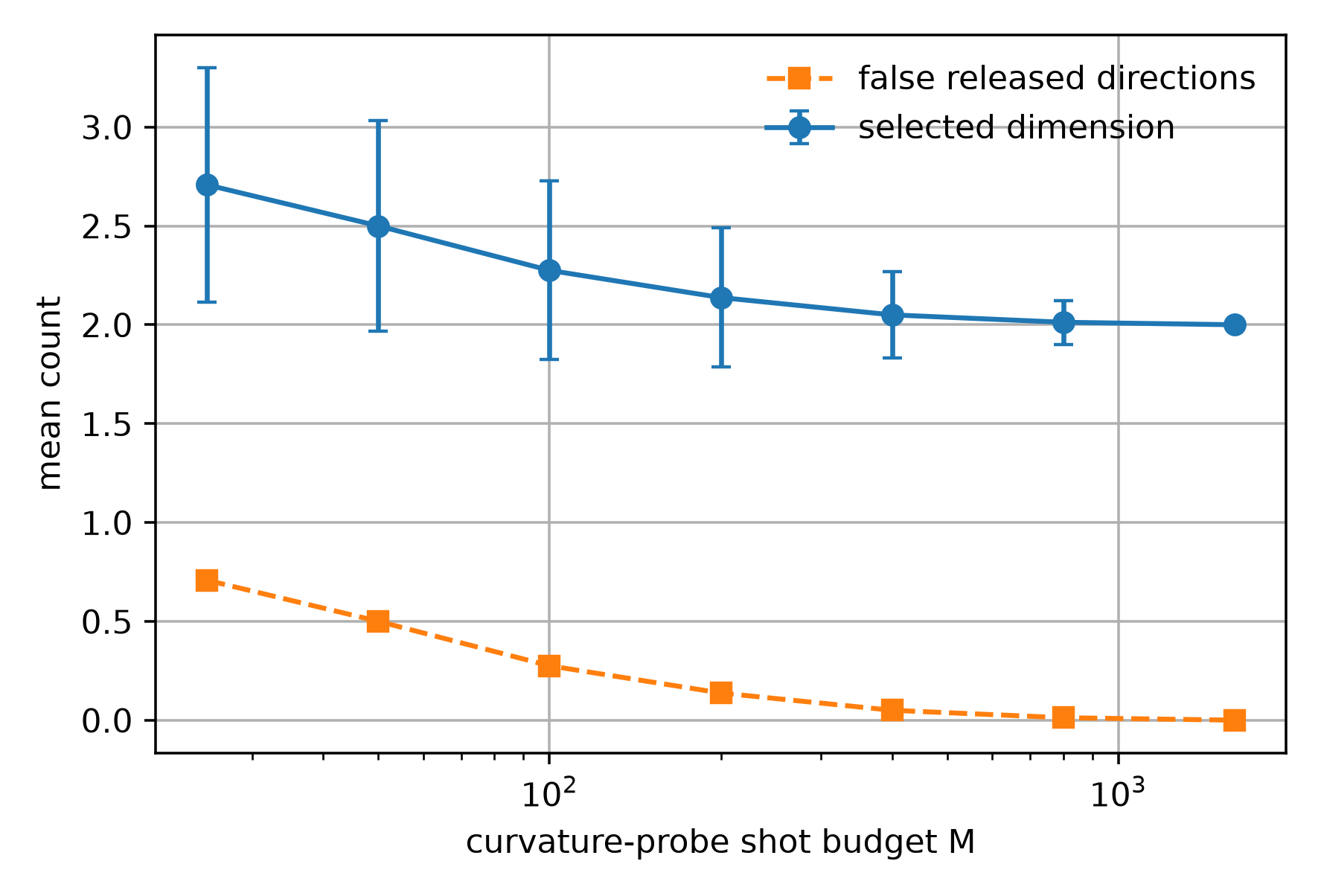}
\caption{Selected and false release counts.}
\label{fig:multi-select}
\end{subfigure}\hfill
\begin{subfigure}{0.48\linewidth}
\centering
\includegraphics[width=\linewidth]{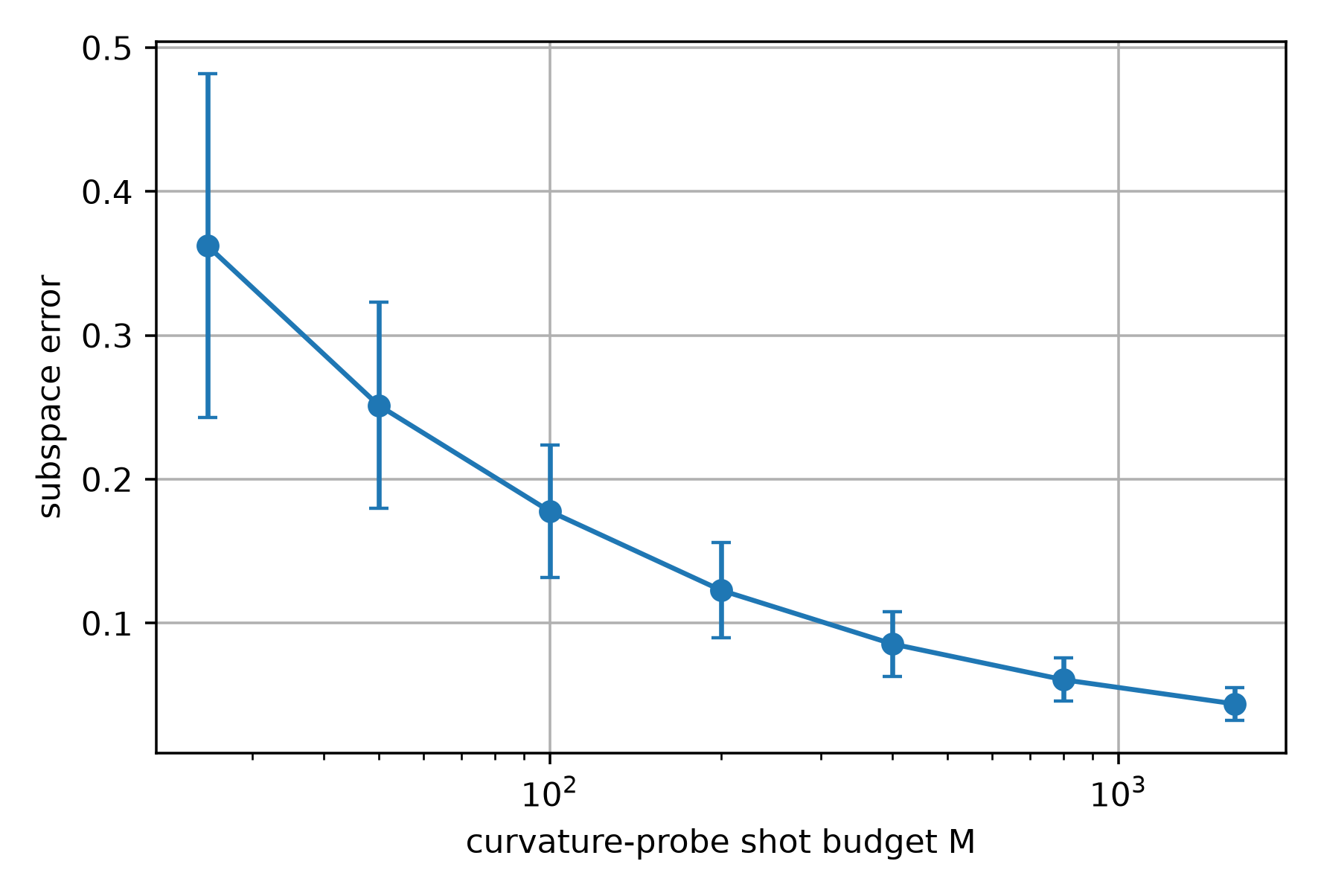}
\caption{Release-subspace error.}
\label{fig:multi-subspace}
\end{subfigure}
\caption{Multi-direction curvature-matrix selection. The prescribed spectral margin reduces over-release in this synthetic noise model; the Ising experiment separately evaluates a rigorous finite-shot certificate.}
\label{fig:multi}
\end{figure}

\subsubsection{Finite-data--finite-shot release selection}
Sample noise and shot noise enter the same margin, so we vary both training sample size $N$ and curvature-probe shot budget $M$ in a Gaussian-matrix surrogate with a fully specified population spectrum. The population matrix has two strong release modes. Empirical task noise scales as $N^{-1/2}$ and shot noise scales as $M^{-1/2}$. \Cref{fig:finite-data-shot} shows the quadratic matrix-estimation excess and selected dimension. The selector is conservative when either $N$ or $M$ is small, and approaches the correct release dimension when both increase. This surrogate isolates the additive data and shot variances; its matrix-estimation loss is distinct from a supervised prediction risk.

\begin{figure}[!htbp]
\centering
\begin{subfigure}{0.48\linewidth}
\centering
\includegraphics[width=\linewidth]{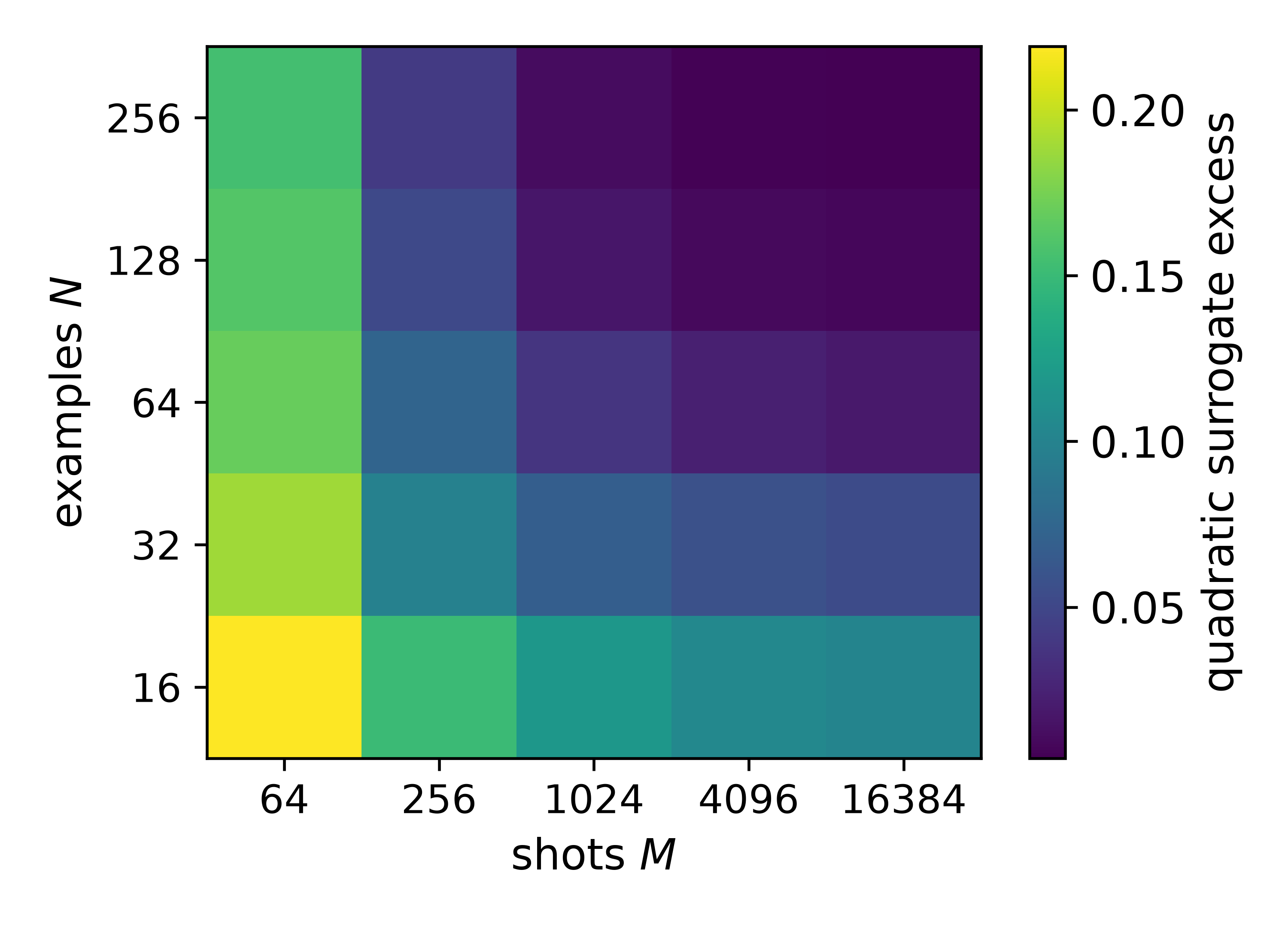}
\caption{Quadratic surrogate excess.}
\end{subfigure}\hfill
\begin{subfigure}{0.48\linewidth}
\centering
\includegraphics[width=\linewidth]{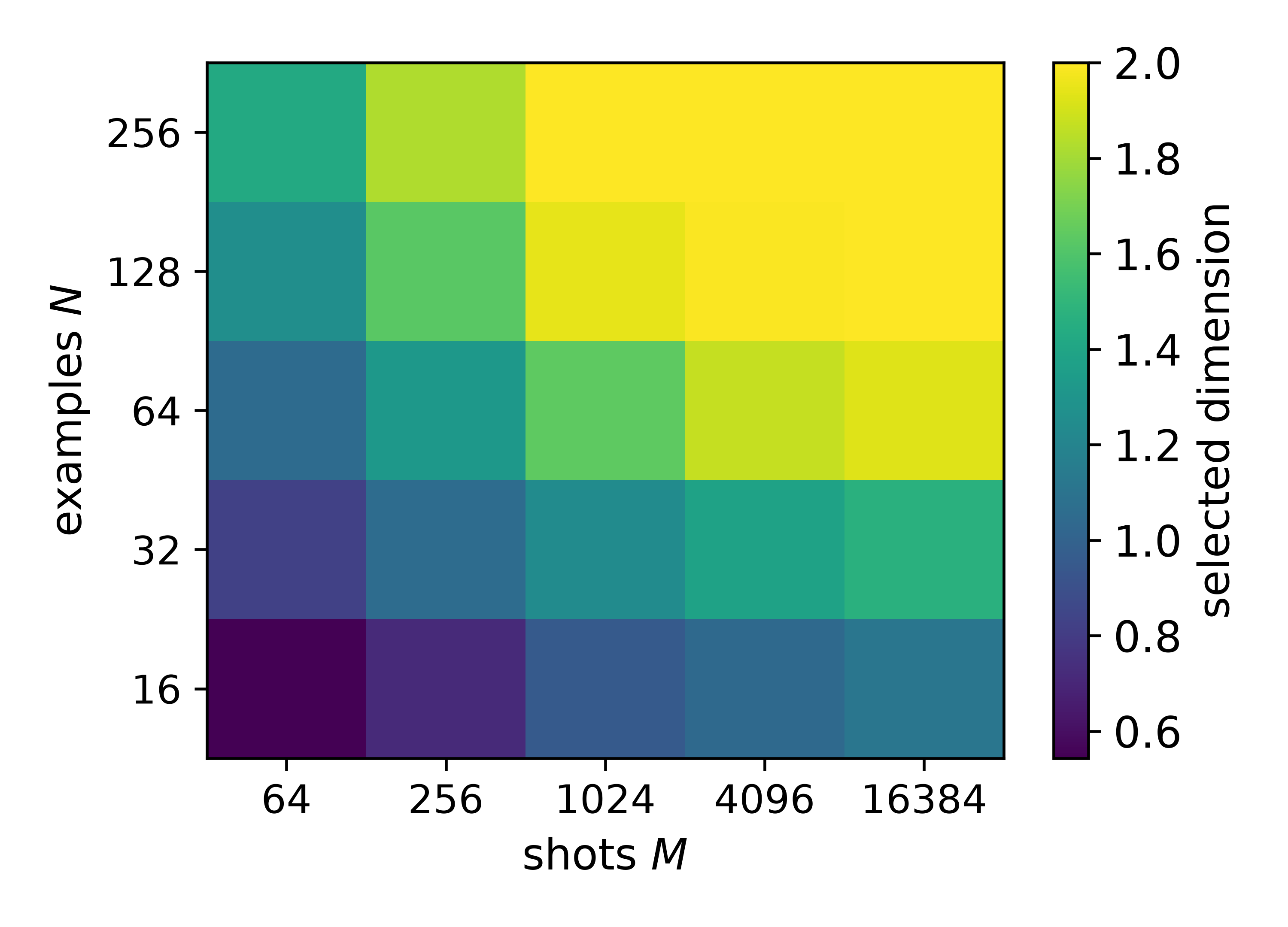}
\caption{Selected release dimension.}
\end{subfigure}
\caption{Finite-data--finite-shot matrix selection. The error is measured against the positive part of the population penalized matrix; both examples and shots control the selected dimension.}
\label{fig:finite-data-shot}
\end{figure}

\subsubsection{Soft-release effective dimension}
To illustrate the effective-dimension cost in \cref{thm:soft-effective-dimension}, we simulate a breaking feature spectrum with a decaying task signal and compare the proxy risk ``bias + variance + shot'' over the soft penalty $\alpha$. \Cref{fig:soft-effective} shows that the optimal soft penalty $\alpha$ decreases as $N$ grows, so the selected release becomes stronger: with more data, the variance cost of an opened direction is cheaper to pay. The gate continuously controls the effective dimension through the squared amplitudes in \eqref{eq:soft-trace}.

\begin{figure}[!htbp]
\centering
\includegraphics[width=0.74\linewidth]{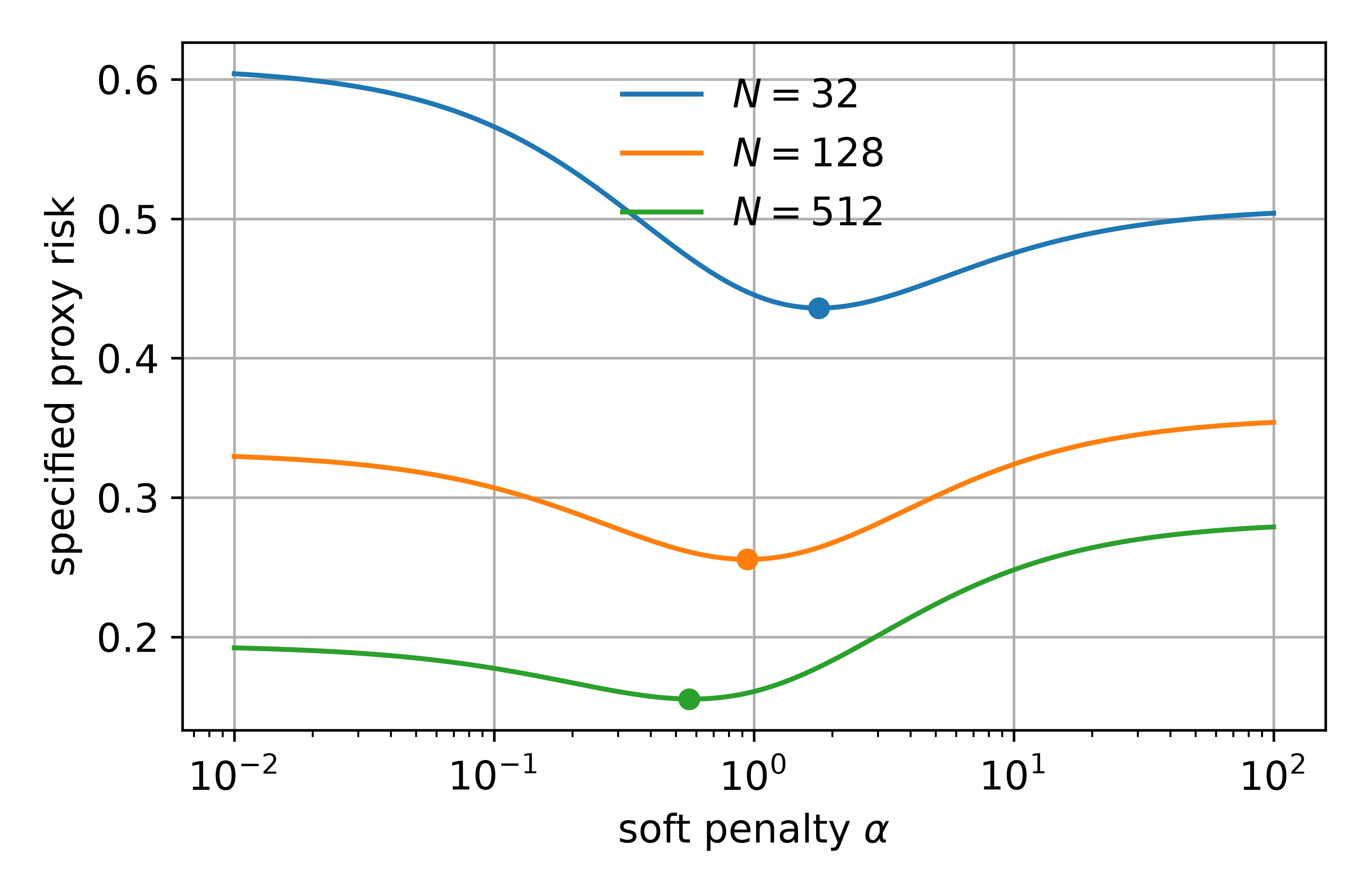}
\caption{Soft-release effective dimension. Larger $\alpha$ means harder symmetry. The optimal soft penalty decreases as sample size increases, because the variance cost of released directions becomes cheaper.}
\label{fig:soft-effective}
\end{figure}

\subsubsection{Structural-risk selection: when to retain and when to release}
We use a transparent even/odd regression task. Inputs are $x\sim\operatorname{Unif}[-1,1]$ and labels are
\begin{equation}
 y=0.75x^2+bx+\epsilon,
 \qquad \epsilon\sim\mathcal N(0,0.25^2).
\end{equation}
The hard model uses even features $(1,x^2)$ and the released model adds the odd feature $x$. The illustrative score is training mean-squared error plus $\lambda\sqrt{d/N}$ and a release shot penalty. Its Gaussian-noise risk curves examine the bias--complexity tradeoff; they do not implement the independent-validation score of \cref{thm:post-selection-oracle}. \Cref{fig:srm-selection} shows that the selector retains symmetry for $b=0$ or weak breaking at small $N$, and releases increasingly often as either $b$ or $N$ grows.

\begin{figure}[!htbp]
\centering
\includegraphics[width=0.74\linewidth]{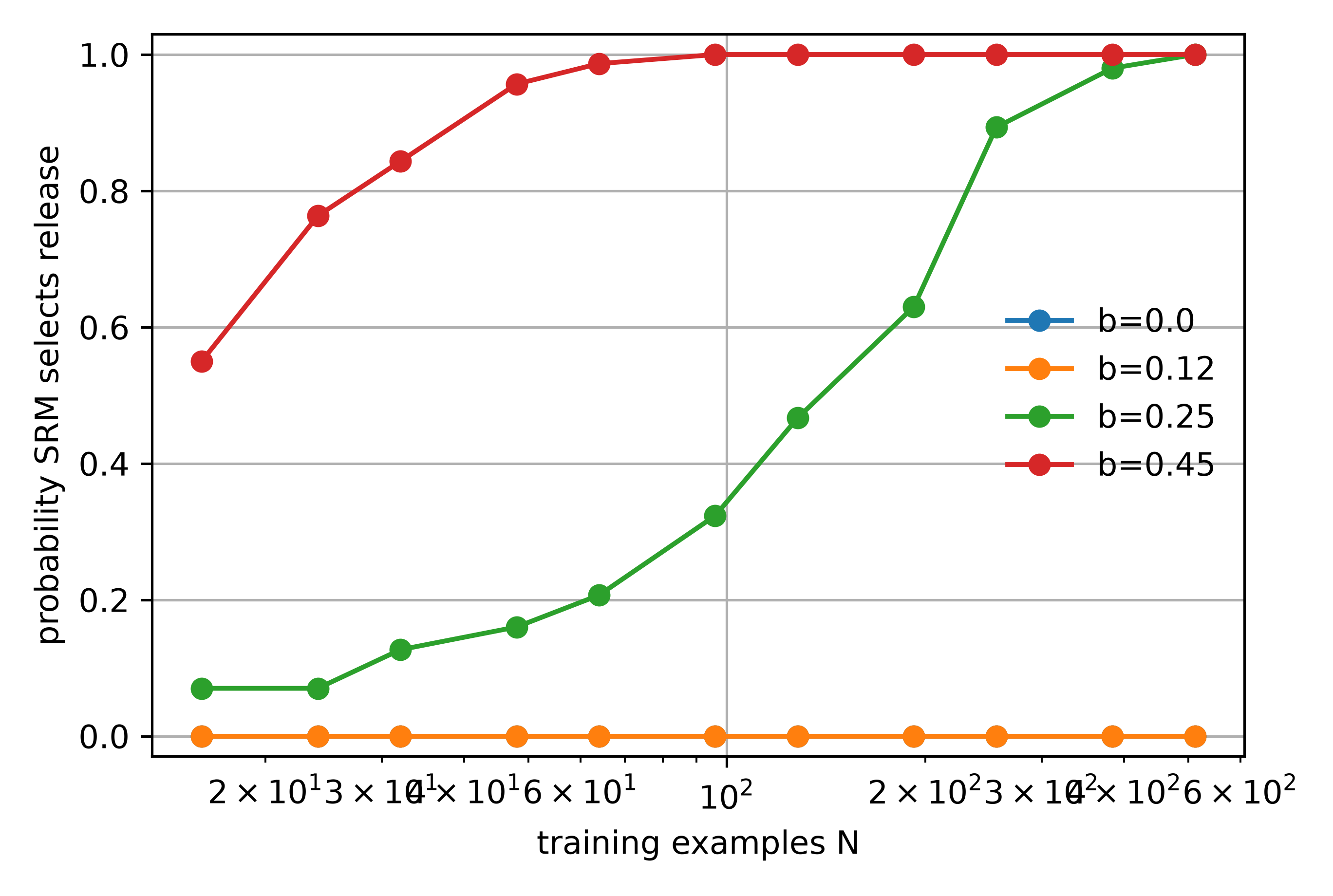}
\caption{Structural-risk release probability in the even/odd toy task. Weak breaking and small samples favor retention; stronger breaking or larger samples favor release.}
\label{fig:srm-selection}
\end{figure}

\subsubsection{Post-selection validity: training search versus independent validation gate}
Search and validation on the same data inflate the reported gain; this test isolates that bias. The hard model again uses even features \((1,x^2)\). During training, the release search considers one true odd feature \(x\) and twelve additional odd polynomial branches, then chooses the branch with the largest training mean-squared error improvement. If release is decided directly from the training improvement, the zero-threshold procedure almost always releases even under the null \(b=0\). Adding a regressor to a nested least-squares model already gives a nonnegative training gain; searching thirteen branches further amplifies this effect. The independent validation gate evaluates only the training-selected branch on fresh data and applies the studentized improvement threshold specified in \cref{app:repro}. This simulation evaluates its false-release behavior; the rigorous finite-sample bounds in \cref{thm:post-selection-oracle,cor:nested-adaptivity} require their stated loss and confidence assumptions. \Cref{fig:postselection-gate} shows that validation suppresses false release at \(b=0\) while recovering power for \(b=0.25\) as \(N\) grows.

\begin{figure}[!htbp]
\centering
\includegraphics[width=0.74\linewidth]{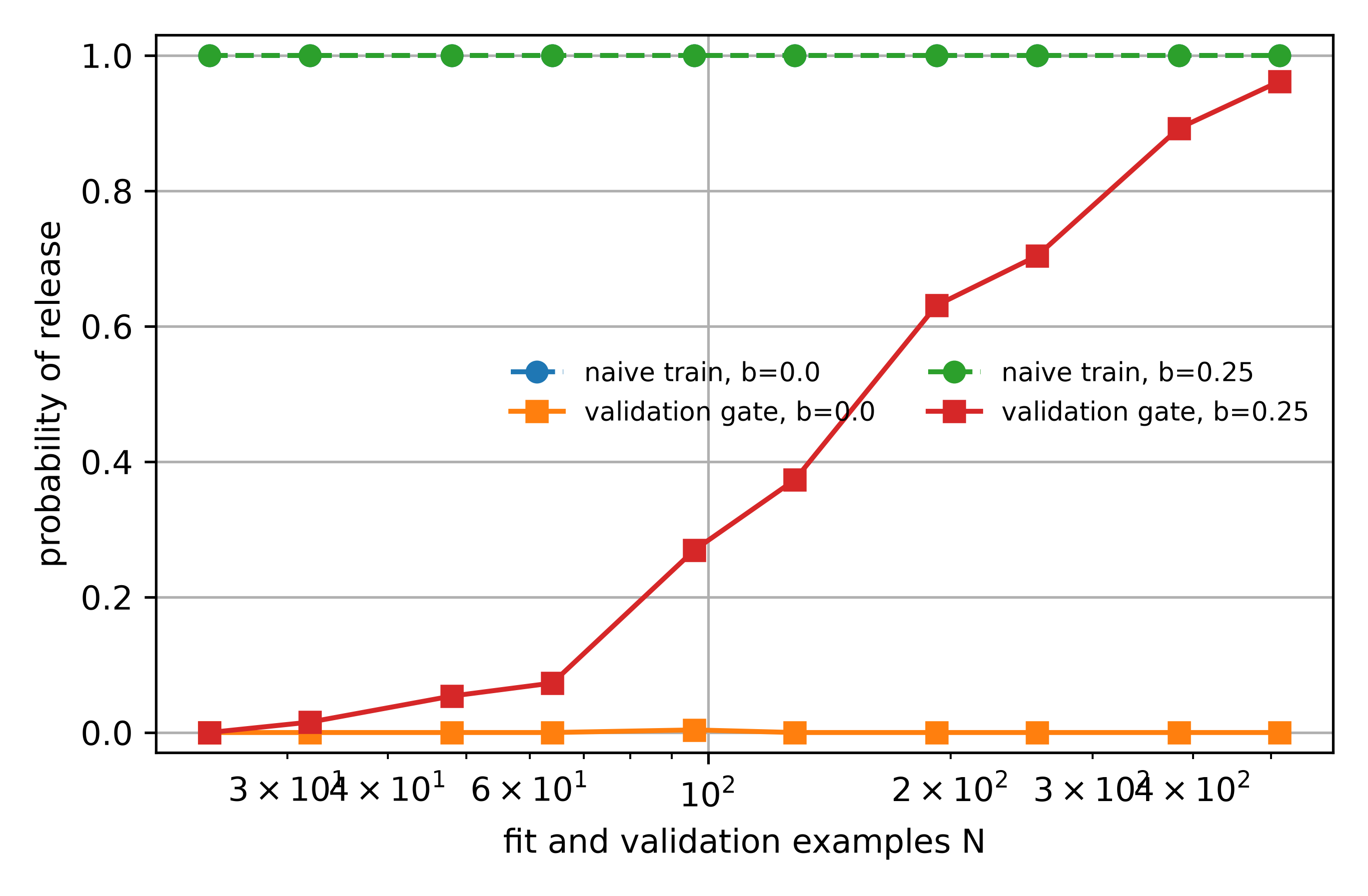}
\caption{Post-selection validation gate. Training-set search finds spurious release directions under no breaking, while an independent validation gate suppresses the false release and retains power when a genuine breaking term is present.}
\label{fig:postselection-gate}
\end{figure}

\FloatBarrier

\subsection{Quotient QNG and mechanism checks}
\label{subsec:exp-qng}

Normal-form dynamics separate the effect of gauge projection from loss optimization and test how projector error changes vertical drift.

\subsubsection{Quotient QNG: exact and estimated projector stress tests}
The normal-form loss separates the physical symmetric coordinate, the release coordinate, and an exact gauge:
\begin{equation}
\cL(\theta,b,v)=L_0+\frac12\theta^2-\frac{\mu}{2}b^2+\frac16b^4+c\theta b^2,
\label{eq:normal-form}
\end{equation}
where $\theta$ is a symmetric coordinate, $b$ is a released coordinate, and $v$ is a pure gauge coordinate. The Fisher matrix is $F=\diag(1,1/2,0)$. Full-space Tikhonov and gauge-projected QNG have similar loss reduction, but full-space damping performs a random walk in $v$, while gauge-projected QNG keeps $v=0$ by construction; see \cref{fig:qng-normal}. The multi-qubit counterpart, on Slater and Givens circuits with a continuous-release comparison of soft modes against exact gauge directions, is \cite{ChenQNGQuotient2026}; the normal form here isolates the same mechanism in three coordinates.

\begin{figure}[!htbp]
\centering
\begin{subfigure}{0.48\linewidth}
\centering
\includegraphics[width=\linewidth]{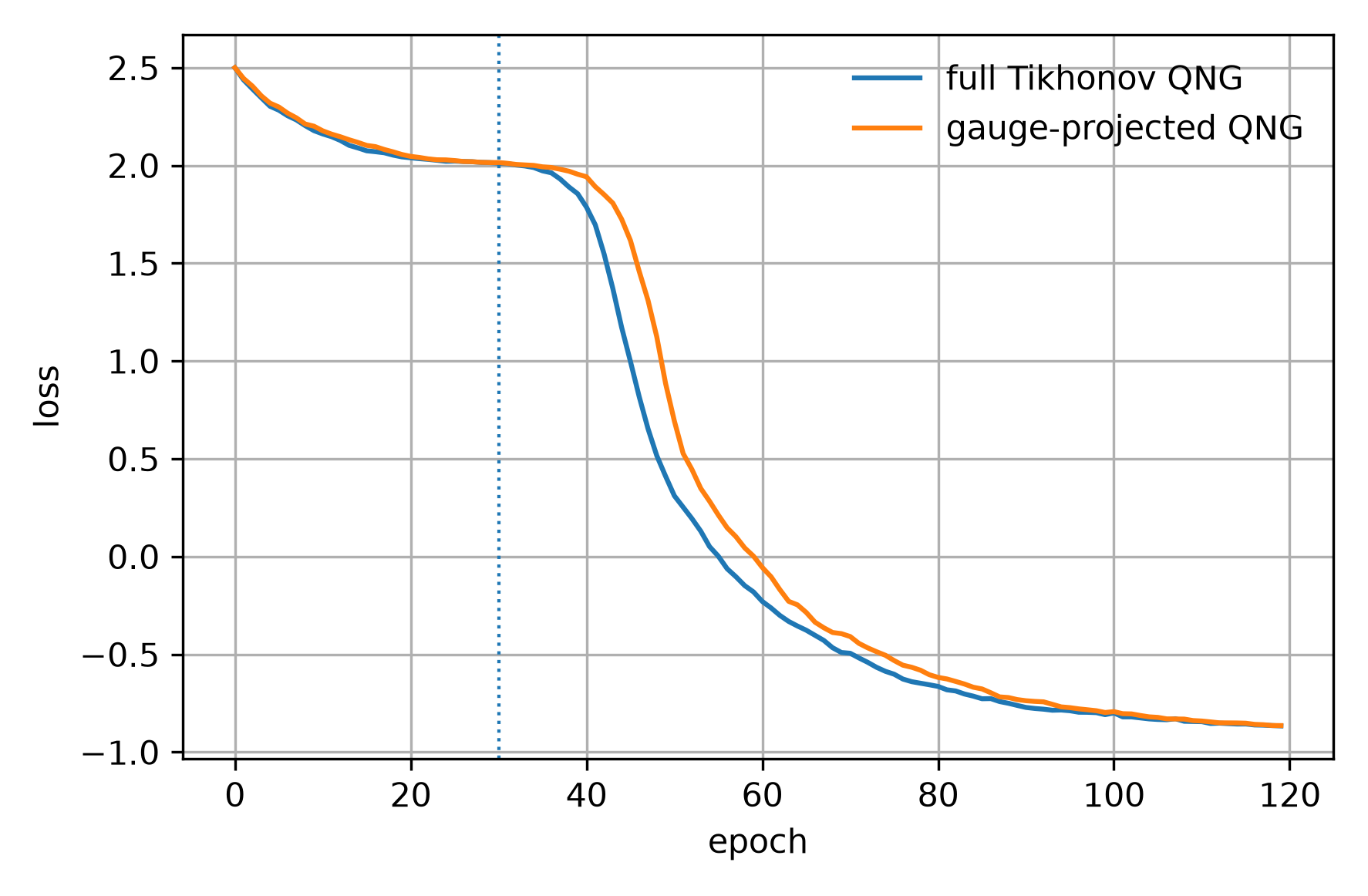}
\caption{Loss after release.}
\label{fig:qng-loss}
\end{subfigure}\hfill
\begin{subfigure}{0.48\linewidth}
\centering
\includegraphics[width=\linewidth]{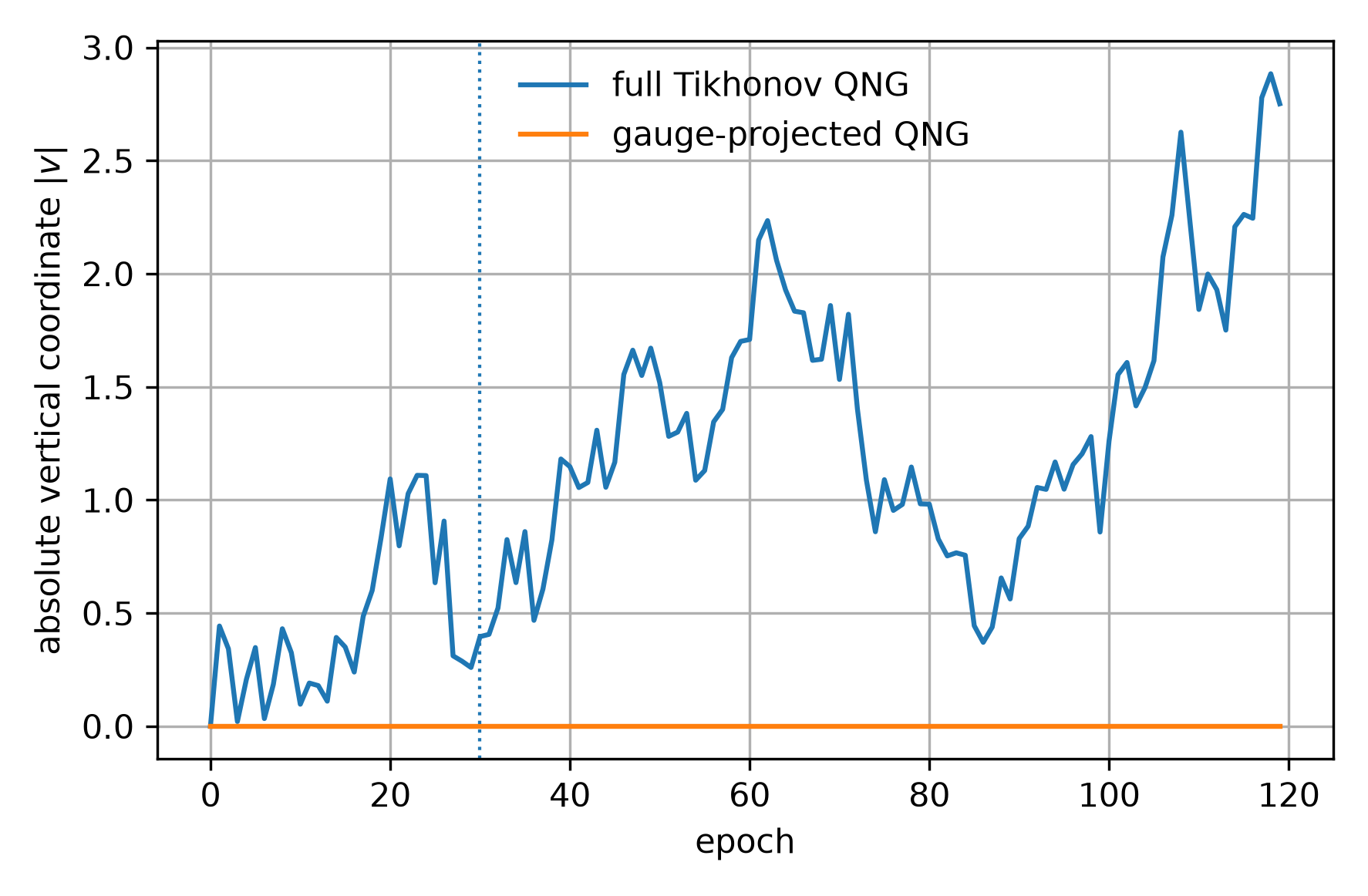}
\caption{Vertical coordinate drift.}
\label{fig:qng-drift}
\end{subfigure}
\caption{Noisy normal-form QNG. Full-space Tikhonov and gauge-projected QNG optimize the physical coordinates similarly, but only the quotient update suppresses vertical representative drift.}
\label{fig:qng-normal}
\end{figure}

Rotating the estimated horizontal projector into the vertical direction by an operator-norm error $\varepsilon_P$ tests sensitivity to imperfect gauge identification. \Cref{fig:qng-leakage} shows that vertical drift grows continuously with the projector error and remains far below the full-space Tikhonov baseline for small errors, matching \cref{thm:estimated-projector}.

\begin{figure}[!htbp]
\centering
\includegraphics[width=0.74\linewidth]{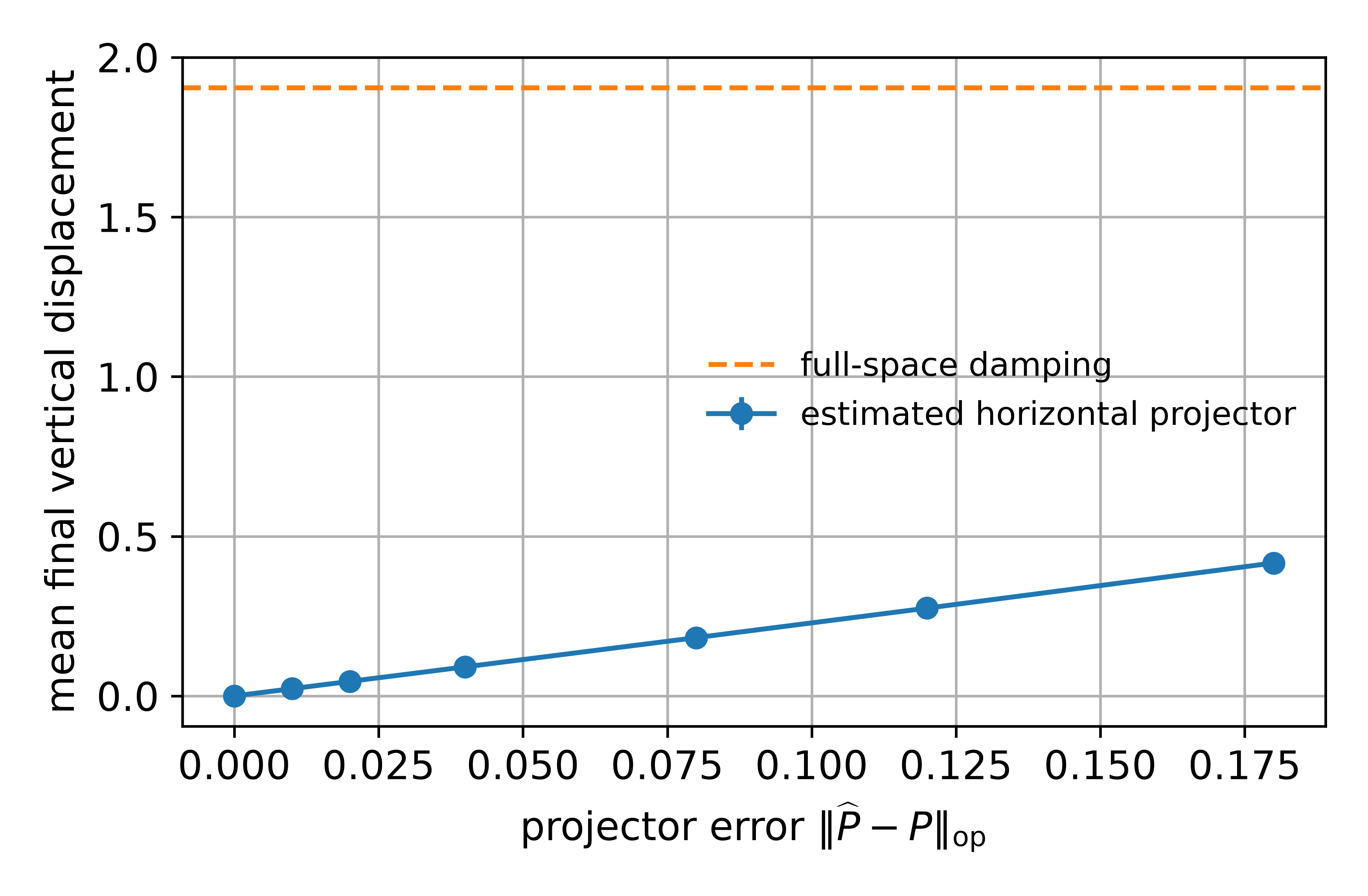}
\caption{Estimated-projector leakage. The true vertical drift of gauge-projected QNG scales with the horizontal-projector subspace error, while full-space Tikhonov damping gives a much larger vertical random walk.}
\label{fig:qng-leakage}
\end{figure}

\subsubsection{Additional mechanism checks}
The final examples isolate curvature concentration, the certified step calculation, and metric damping; \cref{fig:new-mechanisms} displays their governing quantities.

The trainability filter in \cref{cor:dla-bp-margin} is tested in a stylized branch whose curvature variance follows the dynamical Lie-algebra concentration model $C_{\mathrm{BP}}/D_S$ with $D_S=4^n$. Under a fixed polynomial shot budget, the probability of certifying this exponentially flat branch decreases with the algebra dimension and remains below the Markov--concentration bound in \eqref{eq:dla-bp-filter}; see \cref{fig:new-mechanisms}(a). It isolates the statistical decision rule and confirms that the finite-shot margin rejects branches whose curvature is already below the measurable scale.

The cubic benchmark evaluates the step-size calculation in \cref{thm:certified-release-descent}. It uses the known curvature $\lambda_{\mathrm{true}}=0.5$ and subtracts a prescribed sampling margin to obtain the reference gap. The resulting population loss drop lies above \eqref{eq:certified-descent-drop}, as shown in \cref{fig:new-mechanisms}(b). Noisy empirical curvatures are used for the naive-step comparison, whose failure rate rises at small shot budgets. This isolates the deterministic descent calculation; the Ising experiment tests the complete procedure with estimated curvature and a simultaneous confidence radius.

The soft-damping comparison in \cref{fig:new-mechanisms}(c) applies three metric penalties to the same sequence of noisy horizontal gradients. Increasing the penalty reduces the step amplitude while preserving its sign, and a large penalty approximates retention. The displayed paths illustrate the coordinate bound in \eqref{eq:soft-qng-coordinate-bound}. The inverse-square variance law is the analytical statement \eqref{eq:soft-qng-noise-bound}; this experiment reports one shared-noise trajectory for each penalty.

\begin{figure}[!htbp]
\centering
\includegraphics[width=0.98\linewidth]{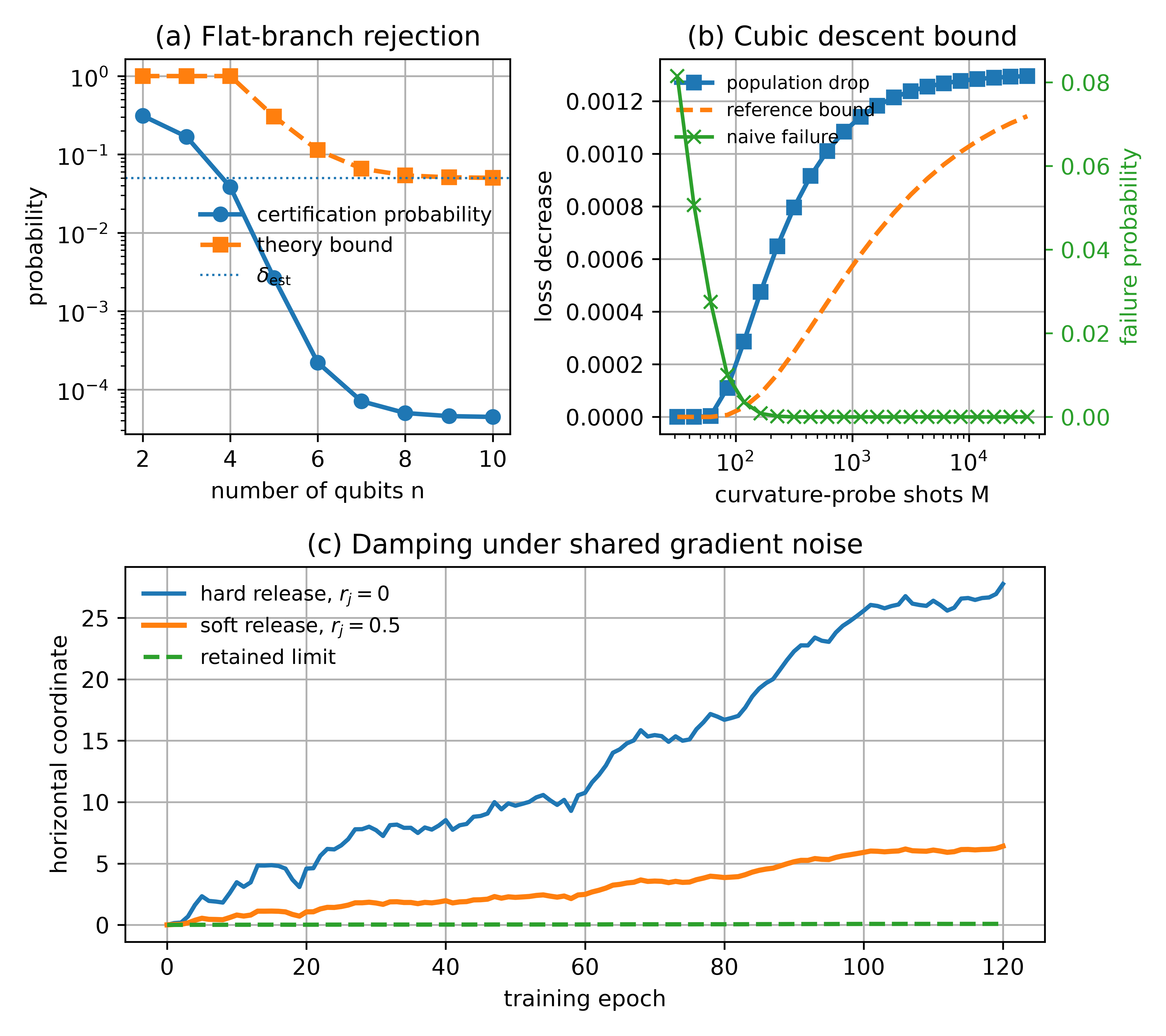}
\caption{Mechanism checks. (a) Certification probability and its bound for a scalar branch satisfying the assumed curvature concentration law. (b) A cubic model compares the reference-gap descent bound with the population drop; noisy naive steps can fail. (c) Three damping penalties applied to one shared-noise gradient sequence illustrate suppression of horizontal step amplitude.}
\label{fig:new-mechanisms}
\end{figure}

\section{Discussion and Conclusion}
\label{sec:discussion}

Symmetry becomes a statistical modeling choice once its information loss is compared with the cost of recovering that information. The group twirl identifies the readout signal removed by a hard constraint. The release gradient and curvature identify which part can improve a specified loss through available circuit directions. Their connection gives a sequence of decisions: detect useful breaking signal, certify a local enlargement, and validate the trained model. The norm-dual characterization fixes the scale of the first decision, including the rank factor needed to control all bounded-outcome readouts.

Simultaneous confidence is what allows the local decision to survive direction selection. Shared Pauli records provide an effective certificate when the release commutators are local; scalar probing remains available through the original loss observable, with truncation bias included before selection. The Ising comparison exhibits the practical consequence of this distinction: on the tested grids, the two specified estimators differ by a factor of $6300$ in the shots required to certify the leading direction. The controlling resources are the measurement ensemble, the operator structure, and the bias bound as well as the size of the release space.

Independent validation carries the local search into a comparison of trained predictors. For bounded squared loss, its fast excess-risk bound states when a nested path preserves its approximation--estimation rate. Quotient training then removes independently identified state-preserving redundancy before damping physical modes \cite{ChenQNGQuotient2026}. Physical group inference \cite{CompanionDiscovery}, statistical model selection, and parameter-gauge identification therefore contribute different information to the final guarantee.

Two extensions would sharpen the decision theory. The minimax copy complexity for an unknown composite symmetric null would determine how much the global gate can be improved beyond structured examples. Uniform or sequential confidence for continuously adapted probe families would allow discovery, direction selection, and finite-difference step choice to share measurements more efficiently. Further resource improvements depend on estimators for nonlocal curvature, controlled hardware correlations, and candidate paths with explicit approximation guarantees. These questions extend the same principle: a symmetry should be relaxed to the extent that the task signal, measurement access, and independently validated improvement support it.

\appendix

\section{Finite-Difference Curvature Estimation}
\label{app:fd}

The scalar-probe mode of \cref{subsec:shot-estimation} is the fallback whenever $[B,[B,C]]$ expands into high-weight observables and the shadow bound \eqref{eq:shadow-local-shot} becomes vacuous. Throughout, $B$ is a fixed Hermitian breaking generator in the declared coordinate normalization, $e^{-i\beta B}$ is implementable, and $\cL(\beta):=\Tr(Ce^{-i\beta B}\rho_0e^{i\beta B})$ with $\kappa_B=-\cL''(0)$.

\begin{proposition}[Release certification from scalar probes]
\label[proposition]{prop:fd-certification}
Assume each $\widehat\cL(\beta)$, $\beta\in[-\delta_\beta,\delta_\beta]$, is unbiased for $\cL(\beta)$, built from $M$ independent shots, with $\Var(\widehat\cL(\beta))\le\sigma_C^2/M$, and set
\begin{equation}
\widehat\kappa_B(\delta_\beta):=
\frac{2\widehat\cL(0)-\widehat\cL(\delta_\beta)-\widehat\cL(-\delta_\beta)}{\delta_\beta^2}.
\label{eq:fd-kappa}
\end{equation}
Then $\bbE\widehat\kappa_B(\delta_\beta)=\kappa_B+O(\delta_\beta^2\norm{B}_{\op}^4\norm{C}_{\op})$ and, for independent probe estimates, $\Var(\widehat\kappa_B)\le6\sigma_C^2/(M\delta_\beta^4)$. If the three probe estimates are independent and each is sub-Gaussian with proxy variance $\sigma_C^2/M$ and $L_4:=\sup_{\abs{\beta}\le\delta_\beta}\abs{\cL^{(4)}(\beta)}$, then with probability at least $1-\eta$
\begin{equation}
\abs{\widehat\kappa_B(\delta_\beta)-\kappa_B}
\le \frac{L_4\delta_\beta^2}{12}
+\frac{\sigma_C}{\delta_\beta^2}\sqrt{\frac{12\log(2/\eta)}{M}},
\label{eq:fd-high-prob}
\end{equation}
so for a finite-shot threshold $\tau_M$ the decision $\widehat\kappa_B>\tau_M$ has the same sign as the ideal test $\kappa_B>\tau_M$ whenever $\abs{\kappa_B-\tau_M}$ exceeds the right-hand side.
\end{proposition}

\begin{proof}
Taylor expansion gives
\begin{equation}
\cL(\pm\delta_\beta)=\cL(0)\pm\delta_\beta\cL'(0)+\tfrac12\delta_\beta^2\cL''(0)\pm\tfrac16\delta_\beta^3\cL'''(0)+O\bigl(\delta_\beta^4\norm{B}_{\op}^4\norm{C}_{\op}\bigr),
\end{equation}
the fourth-derivative bound following from repeated commutators with $B$; summing the two expansions cancels the odd terms and returns $-\cL''(0)$ up to the quartic remainder. Independence gives $\Var(\widehat\kappa_B)=\delta_\beta^{-4}(4\Var\widehat\cL(0)+\Var\widehat\cL(\delta_\beta)+\Var\widehat\cL(-\delta_\beta))$. For \eqref{eq:fd-high-prob}, the deterministic part is the central-difference remainder and the random part is a weighted sum of three independent sub-Gaussian variables with squared weights $4,1,1$ and common factor $\delta_\beta^{-4}$, hence proxy variance $6\sigma_C^2/(M\delta_\beta^4)$.
\end{proof}

For the matrix certificate, truncation must be bounded before selecting a direction. Define known Hermitian operators
\begin{equation}
C_{\delta}(B)=\frac{2C-e^{i\delta B}Ce^{-i\delta B}-e^{-i\delta B}Ce^{i\delta B}}{\delta^2},
\qquad O(B)=[B,[B,C]].
\end{equation}
For every prepared state,
\begin{equation}
|\Tr((C_{\delta}(B)-O(B))\rho_0)|
\le b_\delta(B):=\|C_{\delta}(B)-O(B)\|_{\op}
\le\frac{\delta^2}{12}\|\operatorname{ad}_B^4(C)\|_{\op}
\le\frac{4\delta^2}{3}\|B\|_{\op}^4\|C\|_{\op}.
\label{eq:fd-operator-bias}
\end{equation}
The middle bound follows from the central-difference remainder in operator norm and unitary invariance of the fourth derivative. It uses $B$ and $C$, with no knowledge of $\rho_0$.

For polarization directions $e_a$ and $e_a+e_b$, let $u_a$ and $u_{a+b}$ be simultaneous directional error bounds, each including its statistical error and \eqref{eq:fd-operator-bias}. The entry envelope is $E_{aa}=u_a$ and $E_{ab}=(u_{a+b}+u_a+u_b)/2$ for $a\ne b$. Thus $r_\delta=\|E\|_{\op}$ bounds $\|\widehat K-K\|_{\op}$ by the envelope argument in \cref{prop:empirical-curvature-certificate}. The same radius is used for release selection and the descent step. The population difference $\|K_\delta-K\|_{\op}$ is a retrospective diagnostic and never an input to the decision.

\begin{remark}[Optimal step schedule and the shadow-norm bypass]
\label[remark]{rem:fd-schedule}
The two terms in \eqref{eq:fd-high-prob} are the bias--variance trade-off of dynamical probing. With $L_4>0$ fixed and no hardware constraint on $\delta_\beta$, the bound is minimized at
\begin{equation}
\delta_{\beta,\star}
=\left(\frac{12\sigma_C}{L_4}\sqrt{\frac{12\log(2/\eta)}{M}}\right)^{1/4}
\asymp M^{-1/8},
\label{eq:fd-delta-star}
\end{equation}
where the certified curvature error scales as $O(M^{-1/4})$: the slow exponent is the price of estimating a second derivative from noisy scalar forward passes. The compensation is that the sampling scale is $\sigma_C^2/(\delta_\beta^4M)$ and the truncation scale $L_4\delta_\beta^2$, neither of which involves the shadow norm of $[B,[B,C]]$ that controls the direct estimate of \cref{cor:shadow-commutator}. Nonlocal dynamics still carry a circuit-synthesis cost, but the measurement task becomes a controlled scalar experiment along an available perturbation. In practice $\delta_\beta$ is held fixed and the residual bias is added to the confidence radius, as in \cref{subsec:executable-certification}, so the certificate refers to $K$ rather than to its finite-difference surrogate.
\end{remark}

A closed-form check is supplied by the $L=8$ chain at $h=1$, with product backbone $|+x\rangle^{\otimes L}$, energy loss, and the unnormalized generator $B=M_y$. Here $\cL(\beta)=-L\sin^2(2\beta)-hL\cos(2\beta)$, so $\kappa_B=4L(2-h)=32$. Central differences give $29.08$, $31.26$, $31.81$, and $31.95$ at steps $0.2$, $0.1$, $0.05$, and $0.025$, respectively. The directional Fisher test for pure states applies a related central difference to fidelity, with remainder $L_4s^2/6$ \cite{ChenQNGQuotient2026}; its measurement and bias budget is allocated separately from the loss probes.

\section{Experiment Details}
\label{app:repro}

\subsection{Twirling blind-spot classification}

We use $n=6$ qubits and class-conditional product states $\rho_y=((I+ymZ)/2)^{\otimes n}$ for $m\in\{0,0.05,\ldots,0.8\}$. For each $M\in\{1,2,4,8,16,32,64\}$, all $n$ qubits are measured in the $Z$ basis on $M$ independent copies. The released classifier predicts the sign of the empirical magnetization and breaks ties uniformly at random. Each point uses $5000$ Monte Carlo repetitions. The hard invariant baseline is fixed at accuracy $1/2$ by \cref{cor:binary-release-separation}.

\subsection{Quantum-native detection}

For the equal-budget comparison, \(H=\{I,X^{\otimes n}\}\), \(n=2,\ldots,7\), and every point uses \(2000\) total state copies. States are formed by mixing a rank-at-most-16 Ginibre state with its spin-flip twirl at breaking weight \(0.6\). The twirl--swap estimator uses \(500\) pairs in each arm and sixty state repetitions. The single-copy estimator uses fifty independent local Haar bases, twenty shots on \(\rho\) and twenty on \(X^{\otimes n}\rho X^{\otimes n}\) per basis, and twenty-five state repetitions. The \(S_3\) spectroscopy check uses \(200{,}000\) shared samples on a three-qubit Ginibre state. The Pauli scan uses \(d=16\), \(r^2=1/64\), \(117\) samples in each of seventeen stabilizer bases, all \(255\) nonidentity Pauli statistics, and \(2000\) repetitions under each hypothesis.

\subsection{Detection boundary}

We use $Z=\mu+\xi$, $\xi\sim\mathcal N(0,I_k)$, $k\in\{2,4,8,16,32\}$, and the norm test \eqref{eq:norm-test} with $\delta=0.05$. For each signal level, $5000$ Monte Carlo repetitions are used, with $\|\mu\|^2/\sqrt{k}$ on a $71$-point grid from $0$ to $14$ and seed $20260905$.

\subsection{Sparse search and post-selection validation}
The sparse search uses $p=64$, unit noise scale, and a one-sided maximum test at $\delta=0.05$. The budgets are $M\in\{20,40,80,160,320,640,1280\}$. A uniformly selected coordinate has mean $\mu=a\sqrt{\log(p)/M}$, with $a$ on a $16$-point uniform grid from zero to three; all other means vanish. The simulation uses $3000$ repetitions per setting, threshold $\sqrt{2\log(p/\delta)/M}$, and seed $20260602$. For the post-selection experiment, $x$ is uniform on $[-1,1]$, $y=0.75x^2+bx+\epsilon$, and $\epsilon\sim\mathcal N(0,0.35^2)$. Besides $x$, the dictionary contains $x\cos(j\arccos(2x^2-1))$ for $j=1,\ldots,12$. Fitting selects the largest training mean-squared-error improvement. On an independent validation sample of the same size $N$, let $d_i$ be the paired squared-loss improvement and let $s_d$ be its sample standard deviation. The simulation releases when $\bar d>z_{1-0.05/13}s_d/\sqrt N$, where $z_u$ is the standard normal quantile. This studentized normal approximation is a simulation threshold; the unbounded Gaussian labels do not satisfy the bounded-loss hypotheses of \cref{cor:nested-adaptivity}. We use $b\in\{0,0.25\}$, $N\in\{24,32,48,64,96,128,192,256,384,512\}$, $260$ repetitions, and seed $424242$.

\subsection{Ising exact diagonalization}

The Hamiltonian is \eqref{eq:tfim} with periodic boundary conditions. Sparse Lanczos diagonalization is used for $L=6,8,10,12,14$ and $h=0.5,1.0,1.5$. The computed quantities are $E_0$, $\langle M_z\rangle$, $\langle M_x\rangle$, $\langle M_z^2\rangle$, $\langle M_x^2\rangle$, $\langle P\rangle$, $F_{M_z}$, and $\kappa_B$.

\subsection{Multi-branch release}

The true curvature matrix is an orthogonal conjugate of
\begin{equation}
\diag(1.3,0.75,0.02,-0.02,-0.08,-0.15,-0.24,-0.35).
\end{equation}
The noise standard deviation is $0.85M^{-1/2}$, the structural penalty is $0.08I$, the illustrative spectral threshold is $0.9M^{-1/2}$, and each $M$ is repeated $400$ times.

\subsection{Finite-data--finite-shot release}

Set
\[
K=Q\diag(0.60,0.36,0.02,0,-0.04,-0.10,-0.18,-0.26)Q^\top,
\]
where $Q$ is obtained by QR factorization of a seeded Gaussian matrix. Empirical data noise has standard deviation $0.34N^{-1/2}$ per matrix entry and shot noise has standard deviation $0.90M^{-1/2}$. The release penalty is $1.05\sqrt{\log(2p)/N}+1.20\sqrt{\log(2p^2)/M}+0.06$, with $p=8$. The target is $T=(K-0.06I)_+$ and the reported excess is $\tfrac12\|\widehat T-T\|_{\mathrm F}^2$. The selected estimate retains the eigenvalues of $\widehat K-0.06I$ above the sampling margin; the naive estimate keeps every positive eigenvalue, and the soft estimate subtracts the margin before taking the positive part. Use $N=16,32,64,128,256$, $M=64,256,1024,4096,16384$, $800$ repetitions, and seed $20260906$.

\subsection{Soft release}

There are $q=14$ breaking features. Their covariance eigenvalues are $\nu_j=\exp(-(j-1)/4.5)$ and the true breaking coefficients are $\theta_j=0.35\exp(-(j-1)/3)$. The squared approximation bias is $(1-s)^2\sum_j\nu_j\theta_j^2$, and the proxy risk adds
\begin{equation}
1.20\sqrt{(2+\sum_js^2\nu_j)/N}+2s^2d_{\br}/M,
\end{equation}
with $M=500$ and $d_{\br}=q=14$. Evaluate $N=32,128,512$ and a $161$-point logarithmic grid $\alpha\in[10^{-2},10^2]$. The minimizing penalties are approximately $1.778$, $0.944$, and $0.562$, respectively.

\subsection{Structural-risk toy task}

The sample sizes are
\begin{equation}
N\in\{16,24,32,48,64,96,128,192,256,384,512\},
\end{equation}
with breaking strengths $b=0,0.12,0.25,0.45$ and $300$ repetitions per point. The hard and released dimensions are $2$ and $3$. The score is training mean-squared error plus $0.60\sqrt{d/N}$ and a release shot penalty $0.80/200$.

\subsection{QNG normal forms}

For \eqref{eq:normal-form}, $\mu=2$, $c=0.4$, $L_0=2$, $F=\diag(1,1/2,0)$, step size $0.06$, damping $\gamma=0.05$, gradient-noise standard deviation $0.18$, and release epoch $30$. The projector-leakage experiment uses the horizontal basis $[e_1,\sqrt{1-\varepsilon_P^2}e_2+\varepsilon_Pe_3]$ for $\varepsilon_P\in\{0,0.01,0.02,0.04,0.08,0.12,0.18\}$. It starts at $(1,0,0)$, seeds $b=0.05$ at epoch $30$, runs $120$ steps for $400$ repetitions, and uses seed $20260907$. 

\subsection{Mechanism checks}

For the barren-plateau filter, we use $n=2,\ldots,10$, $D_S=4^n$, $M=1000$, $C_{\mathrm{BP}}=1$, and Gaussian estimator noise of standard deviation $\sigma/\sqrt M$ with $\sigma=1$. Set $r_\delta=2\sigma z_{1-\delta_{\mathrm{est}}/2}/\sqrt M$ and $\delta_{\mathrm{est}}=0.05$, so the estimator exceeds $r_\delta/2$ in absolute value with probability exactly $0.05$, as required by \cref{cor:dla-bp-margin}. The cubic check uses $L_3=8$, known curvature $0.5$, and a reference gap over shot budgets $M\in[10^{1.5},10^{4.5}]$. The soft-damping check has horizontal Fisher eigenvalue $0.1$, learning rate $0.1$, Tikhonov damping $0.05$, and gradient-noise standard deviation $0.8$ over $120$ epochs.

\subsection{Executable certification on a quantum state}

Both experiments use $L=8$, the same $p=6$ dictionary with $\Tr(B_aB_b)/2^L=\delta_{ab}$, and $C=-(M_z/L)^2$. Scalar settings use $h=0.5,1.0,1.5$, $M=10^3,\ldots,10^7$, $\delta_\beta=0.1,0.25,0.4$, and $200$ repetitions each. The comparison uses $h=0.5,1.5$, $20$ repetitions, shadow budgets $10^3,10^4,10^5$, and scalar budgets $10^3,\ldots,10^6$ per loss estimate at $\delta_\beta=0.25$. Each fixed setting has confidence budget $0.05$. Empirical Bernstein intervals use unbiased sample variances; the shadow ranges cover all compatible Pauli measurement bases, and the scalar truncation bound uses only the specified operators. We evaluate realized operator errors, confidence coverage, empirical selection counts, and false release against $K-I/M$. The operator calculations and exact-state simulation are feasible at this size; their classical cost is not included in the state-shot count. Known-state quantities audit the result after selection and do not enter the confidence radius.

\end{document}